\documentclass[autoref]{lipics-v2021}
\nolinenumbers
\hideLIPIcs
\pdfoutput=1

\usepackage{style}
\usepackage{shortcuts}
\usepackage{figures}

\title{A Structural Investigation of the Approximability of Polynomial-Time Problems}
\author{Karl Bringmann}{Saarland University and Max Planck Institute for Informatics, Saarland Informatics Campus, Saarbrücken, Germany}{}{}{}
\author{Alejandro Cassis}{Saarland University and Max Planck Institute for Informatics, Saarland Informatics Campus, Saarbrücken, Germany}{}{}{}
\author{Nick Fischer}{Saarland University and Max Planck Institute for Informatics, Saarland Informatics Campus, Saarbrücken, Germany}{}{}{}
\author{Marvin Künnemann}{TU Kaiserslautern, Germany}{}{}{}
\authorrunning{K. Bringmann, A. Cassis, N. Fischer, M. Künnemann}

\funding{\emph{Karl Bringmann, Alejandro Cassis, Nick Fischer:} This work is part of the project TIPEA that has received funding from the European Research Council (ERC) under the European Unions Horizon 2020 research and innovation programme (grant agreement No.~850979). \emph{Marvin Künnemann:} Part of this research was performed at ETH Zürich, Institute for Theoretical Studies, supported by Dr. Max Rössler, by the Walter Haefner Foundation,
and by the ETH Zürich Foundation.}

\Copyright{Karl Bringmann, Alejandro Cassis, Nick Fischer, and Marvin Künnemann}



\begin{document}
\maketitle

\begin{abstract}
An extensive research effort targets optimal (in)approximability results for various \NP{}-hard optimization problems. Notably, the works of (Creignou'95) as well as (Khanna, Sudan, Trevisan, Williamson'00) establish a tight characterization of a large subclass of \MaxSNP{}, namely Boolean \MaxCSP{}s and further variants, in terms of their polynomial-time approximability. Can we obtain similarly encompassing characterizations for classes of \emph{polynomial-time optimization problems}?
	 
To this end, we initiate the systematic study of a recently introduced polynomial-time analogue of \MaxSNP{}, which includes a large number of well-studied problems  (including Nearest and Furthest Neighbor in the Hamming metric, Maximum Inner Product, optimization variants of $k$-XOR and Maximum $k$-Cover). Specifically, for each $k$, $\MaxSP_k$ denotes the class of $O(m^k)$-time problems of the form $\max_{x_1,\dots, x_k} \#\{y : \phi(x_1,\dots,x_k,y)\}$ where $\phi$ is a quantifier-free first-order property and $m$ denotes the size of the relational structure. Assuming central hypotheses about clique detection in hypergraphs and exact \MAXThreeSAT{}, we show that for any $\MaxSP_k$ problem definable by a quantifier-free $m$-edge \emph{graph} formula~$\phi$, the \emph{best possible} approximation guarantee in faster-than-exhaustive-search time $O(m^{k-\delta})$ falls into one of four categories:
\begin{itemize}
\item optimizable to exactness in time $O(m^{k-\delta})$,
\item an (inefficient) approximation scheme, i.e., a $(1+\varepsilon)$-approximation in time $O(m^{k-f(\varepsilon)})$,
\item a (fixed) constant-factor approximation in time $O(m^{k-\delta})$, or
\item an $m^\varepsilon$-approximation in time $O(m^{k-f(\varepsilon)})$.
\end{itemize}
We obtain an almost complete characterization of these regimes, for $\MaxSP_k$ as well as for an analogously defined minimization class $\MinSP_k$. As our main technical contribution, we show how to rule out the existence of \emph{approximation schemes} for a large class of problems admitting constant-factor approximations, under a hypothesis for \emph{exact} Sparse \MAXThreeSAT{}  algorithms posed by (Alman, Vassilevska Williams'20). As general trends for the problems we consider, we observe: (1)~Exact optimizability has a simple algebraic characterization, (2)~only few maximization problems do not admit a constant-factor approximation; these do not even have a subpolynomial-factor approximation, and (3)~constant-factor approximation of minimization problems is \emph{equivalent} to deciding whether the optimum is equal to~$0$.
\end{abstract}

\section{Introduction}
For many optimization problems, the best known exact algorithms essentially explore the complete search space, up to low-order improvements. While this holds true in particular for NP-hard optimization problems such as maximum satisfiability, also common polynomial-time optimization problems like Nearest Neighbor search are no exception. When facing such a problem, the perhaps most common approach is to relax the optimization goal and ask for \emph{approximations} rather than optimal solutions. Can we obtain an approximate value significantly faster than exhaustive search, and if so, what is the best approximation guarantee we can achieve?

Even considering polynomial-time problems only, the study of such questions has led to significant algorithmic breakthroughs, such as locality-sensitive hashing (LSH)~\cite{IndykM98,GionisIM99, AndoniIR18}, subquadratic-time approximation algorithms for Edit Distance~\cite{LandauMS98,BatuEKMRRS03,Bar-YossefJKK04,BatuES06,AndoniO09,AndoniKO10,ChakrabortyDGKS18,KouckyS20,BrakensiekR20,GoldenbergRS20, AndoniN20}, scaling algorithms for graph problems ~\cite{Gabow85, Zwick02, DuanP14}, and fast approximation algorithms via the polynomial method~\cite{AlmanCW16,AlmanCW20} (which lead to the currently fastest known exponential-time approximation schemes for MaxSAT).

These algorithmic breakthroughs have recently been complemented by exciting tools for proving hardness of approximation in P: Most notably, the distributed PCP framework~\cite{AbboudRW17} has lead to strong conditional lower bounds, including fundamental limits for Nearest Neighbor search~\cite{Rubinstein18}, as well as tight approximability results for the Maximum Inner Product problem~\cite{Chen18}. Other technical advances include evidence against deterministic approximation schemes for Longest Common Subsequence~\cite{AbboudB17,AbboudR18,ChenGLRR19}, strong (at times even tight) problem-specific hardness results such as~\cite{RodittyVW13,Bringmann14,BackursRSVWW18,BringmannKW19, KarthikM19}, the first fine-grained equivalences of approximation in P results~\cite{ChenW19,ChenGLRR19}, and related works on parameterized inapproximability~\cite{ChalermsookCKLM17,KarthikLM19,Manurangsi20}, see~\cite{FeldmannKLM20} for a survey.

These strong advances from both sides shift the algorithmic frontier and the frontier of conditional hardness towards each other. Consequently, it becomes increasingly important to generalize isolated results -- both algorithms and reductions -- towards making these frontiers explicit: A more comprehensive description of current techniques might enable us to understand general trends underlying these works and to highlight the most pressing limitations of current methods. In fact, we view this as one of the fundamental tasks and uses of fine-grained algorithm design \& complexity.

\subparagraph*{Optimization Classes in P}
To study (hardness of) approximation in P in a general way, we study a class $\MaxSP_k$ recently introduced in~\cite{BringmannCFK21} (see Section~\ref{sec:SNPmotivation} for a comparison to the classic class $\MaxSNP$, which motivated the definition of $\MaxSP_k$). This class consists of polynomial-time optimization problems of the form: 
\begin{equation*}
	\max_{x_{1},\dots,x_{k}}\#\{y : \phi(x_{1},\dots,x_{k},y)\},
\end{equation*}
where $\phi$ is a quantifier-free first-order property.\footnote{Note that \cite{BringmannCFK21} more generally defines classes $\MaxSP_{k,\ell}$ for $\ell \ge 1$. We focus on $\MaxSP_k = \MaxSP_{k,1}$, as it was determined as computationally harder than $\MaxSP_{k,\ell}$ with $\ell \ge 2$, and contains many natural problems (see below).} One obtains an analogous minimization class $\MinSP_k$ by replacing maximization by minimization. A large number of natural and well-studied problems can be expressed this way: In particular, we may think of each $x_i$ ranging over a set $X_i$ of~$n$ vectors in $\{0,1\}^d$, and the task is to maximize (or minimize) the number of coordinates~$y$ satisfying an arbitrary Boolean function over the coordinate values $x_1[y],\dots,x_k[y]$ (here, $\phi$ is defined using a binary relation $R\subseteq(X_{1}\cup\cdots\cup X_{k})\times Y$ that expresses whether the $y$-th coordinate of $x_i$ is $0$ or $1$; see Section~\ref{sec:preliminaries} for details). In particular, this class of problems includes:
\begin{itemize}
	\item (Offline) Furthest/Nearest Neighbor search in the Hamming metric\\
	($\max_{x_1,x_2} / \min_{x_1,x_2}d_{H}(x_{1},x_{2})$),
	\item Maximum/Minimum Inner Product; the latter is the optimization formulation of the Most-Orthogonal Vectors problem~\cite{AbboudBVW15}\\
	($\max_{x_1,x_2} / \min_{x_1,x_2} \innerprod{x_1}{x_2}$), 
	\item A natural similarity search problem that we call Maximum $k$-Agreement\\($\max_{x_1, \dots, x_k} \#\{y : x_1[y]=\cdots=x_k[y]\}$),
    \item Maximum $k$-Cover~\cite{Feige98, Cohen-AddadGKLL19, Manurangsi20} and its variation Maximum Exact-$k$-Cover~\cite{Knuth00}\\($\max_{x_1,\dots,x_k} \#\{y : \text{$x_i[y]=1$ for some $i$} \}$, $\max_{x_1,\dots,x_k} \#\{y : \text{$x_i[y]=1$ for exactly one $i$} \}$)\footnote{These problems are typically studied in the setting where $k$ is part of the input, while we consider them for a fixed constant $k$.},
	\item The canonical optimization variants of the \kXOR{k} problem~\cite{JafargholiV16,DietzfelbingerSW18}\\
	($\max_{x_1,\dots,x_k}/\min_{x_1,\dots,x_k} \#\{y : x_1[y]\oplus \cdots \oplus x_k[y]=0\}$).
	\medbreak\noindent
	By the standard split-and-list technique~\cite{Williams05}, it is easy to see that any $c$-approximation for Maximum \kXOR{k} or Minimum \kXOR{k} in time $O(n^{k-\delta}\poly(d))$ gives a $c$-approximation for \MaxLIN{} (maximize the number of satisfied constraints of a linear system over $\field_2$, see~\cite{Hastad01,Williams05,AlmanCW20}) or the Minimum Distance Problem\footnote{In fact, even for the Nearest Codeword Problem over $\field_2$.} (finding the minimum weight of a non-zero code word of a linear code over~$\field_2$, see~\cite{Stephens-DavidowitzV19}), respectively, in time $O(2^{n(1-\delta')})$.
\end{itemize}

By a simple baseline algorithm, all of these problems can be solved in time $O(m^k)$, where $m$ denotes the input size (for the above setting of $k$ sets of $n$ vectors in $\{0,1\}^d$ we have $m=O(nd)$). A large body of work addresses problems with $k=2$, typically inventing or adapting strong techniques to each specific problem as needed:
\begin{itemize}
	\item Abboud, Rubinstein, and Williams introduced the distributed PCP in P framework and ruled out almost-polynomial approximations for the Maximum Inner Product problem assuming the Strong Exponential Time Hypothesis (\SETH{}). Subsequently, Chen~\cite{Chen18} strengthened the lower bound and gave an approximation algorithm resulting in tight bounds on the approximability in strongly subquadratic time, assuming \SETH{}. Corresponding inapproximability results for its natural generalization to $k$-Maximum Inner Product have been obtained in~\cite{KarthikLM19}.
	\item In contrast, strong approximation algorithms have a rich history for the Nearest Neighbor search problem: Using LSH, we can obtain an $(1+\varepsilon)$-approximation in time $O(n^{2-\Theta(\varepsilon)})$~\cite{Har-PeledIM12,AndoniI06, AndoniINR14, AndoniR15}. Using further techniques, the dependence on $\varepsilon$ has been improved to $O(n^{2-\Omega(\sqrt{\varepsilon})})$~\cite{Valiant15} and \raisebox{0pt}[0pt][0pt]{$O(n^{2-\tilde{\Omega}(\sqrt[3]{\varepsilon})})$}~\cite{AlmanCW16,AlmanCW20}. On the hardness side, Rubinstein~\cite{Rubinstein18} shows that the dependence on $\varepsilon$ cannot be improved indefinitely, by proving that for every $\delta$, there exists an $\varepsilon$ such that $(1+\varepsilon)$-approximate Nearest Neighbor search requires time $\Omega(n^{2-\delta})$ assuming \SETH{}. In particular, this rules out $\poly (1/\varepsilon) n^{2-\delta}$-time algorithms with $\delta > 0$.
	\item For the dual problem of Furthest Neighbor search (i.e., diameter in the Hamming metric), \cite{BorodinOR99} gives a $O(n^{2-\Theta(\varepsilon^2)})$-time algorithm, which was improved to $O(n^{2-\Theta(\varepsilon)})$ via reduction to Nearest Neighbor search in \cite{Indyk00}. Following further improvements~\cite{GoelIV01,Indyk03}, also here \cite{AlmanCW16, AlmanCW20} give an \raisebox{0pt}[0pt][0pt]{$O(n^{2-\tilde{\Omega}(\sqrt[3]{\varepsilon})})$}-time algorithm. Analogous inapproximability to Rubinstein's result are given in~\cite{ChenW19}.
	\item For Minimum Inner Product, the well-known Orthogonal Vectors hypothesis~\cite{VassilevskaW18} (which is implied by \SETH{}~\cite{Williams05}) is precisely the assumption that already distinguishing whether the optimal value is 0 or at least 1 cannot be done in strongly subquadratic time. Interestingly, Chen and Williams~\cite{ChenW19} show a converse: a refutation of the Orthogonal Vectors hypothesis would give a subquadratic-time constant-factor approximation for Minimum Inner Product.
\end{itemize}

In the above list, we focus on the difficult case of moderate dimension $d=n^{o(1)}$ (when measuring the time complexity with respect to the input size). Lower-dimensional settings such as $d= \Theta(\log n)$, $d= \Theta(\log \log n)$ or even lower are addressed in other works~\cite{Williams18, ChenW19}.

While the above collection of results gives a detailed understanding of isolated problems, we know little about general phenomena of (in)approximability in $\MaxSP_k$ and $\MinSP_k$ using faster-than-exhaustive-search algorithms: Are there problems for which constant-factor approximations are best possible? Is maximizing (or minimizing) Inner Product the only problem without a constant-factor approximation? Which problems can we optimize exactly?

There are precursors to our work that show fine-grained \emph{equivalence classes} of approximation problems in P~\cite{ChenGLRR19,ChenW19}. However, establishing membership of a problem in these classes requires a problem-specific proof, while we are interested in \emph{syntactically} defined classes, where class membership can be immediately read off from the definition of a problem. Our aim is to understand the approximability landscape of such classes fully. Finally, the previous works either focus on lower-dimensional settings~\cite{ChenW19}\footnote{Chen and Williams also give some results for the moderate-dimensional case; we discuss these in Section~\ref{sec:technical}.}, or target more powerful problems than we consider, such as Closest-LCS-Pair~\cite{ChenGLRR19}.

\subparagraph*{$\MaxSP_k$ as Polynomial-Time Analogue of \MaxSNP{}} 
Investigating $\NP$-hard optimization problems, Papadimitriou and Yannakakis~\cite{PapadimitriouY91} introduced the class \MaxSNP, which motivates the definition of $\MaxSP_k$ as a natural polynomial-time analogue (see \cite{BringmannCFK21} and Section~\ref{sec:SNPmotivation} for details). As a general class containing prominent, constant-factor approximable optimization problems, $\MaxSNP$ was introduced to give the first evidence that \MAXThreeSAT{} does not admit a PTAS, by proving that \MAXThreeSAT{} belongs to the hardest-to-approximate problems in \MaxSNP.

Ideally, one would like to understand the approximability landscape in $\MaxSNP$ fully and give tight approximability results for each such problem. Major advances towards this goal have been achieved by Creignou~\cite{Creignou95} and Khanna, Sudan, Trevisan, and Williamson~\cite{KhannaSTW00} who gave a complete classification of a large subclass of $\MaxSNP$, namely, maximum Boolean Constraint Satisfaction Problems (\MaxCSP{}): Each Boolean \MaxCSP{} either is polynomial-time optimizable or it does not admit a PTAS unless $\P=\NP$, rendering a polynomial-time constant-factor approximation best possible. For minimization analogues, including Boolean \MinCSP{}s, the situation is more diverse with several equivalence classes needed to describe the result~\cite{KhannaSTW00}.

We initiate the study of the same type of questions in the polynomial-time regime. Our aim is to achieve a detailed understanding of $\MaxSP_k$ and $\MinSP_k$ akin to the classification theorems achieved for \MaxCSP{}s and \MinCSP{}s~\cite{Creignou95,KhannaSTW00}.

\subsection{Our Results}
We approach the classification of $\MaxSP_k$ and $\MinSP_k$ by considering the simplest, yet expressive case of a \emph{single, binary relation} involved in the first-order formula; this route was also taken in earlier classification work for first-order properties~\cite{BringmannFK19}. We may thus view the relational structure as a graph, and call such a formula a \emph{graph formula}. Note that despite its naming, this includes problems not usually viewed as graph problems, such as all examples given in the introduction, which are natural problems on sets of \emph{vectors} in $\{0,1\}^d$.

We obtain a classification of each graph formula into one of four regimes, assuming central fine-grained hardness hypotheses whose plausibility we detail in Section~\ref{sec:assumptions}. All of our hardness results are implied by the Sparse \MAXThreeSAT{} hypothesis~\cite{AlmanW20}, which states that for all $\delta > 0$ there is some $c > 0$ such that \MAXThreeSAT{} on $n$ variables and $c n$ clauses has no $O(2^{n(1-\delta)})$-time algorithm.\footnote{This hypothesis is a stronger version of the \MAXThreeSAT{} hypothesis~\cite{LincolnWW18}.} (Actually, most of our hardness results already follow from weaker assumptions. For a discussion, see Section~\ref{sec:assumptions}.)
\newcommand\RBox[2]{%
	\tikz[baseline=(label.base)]{
        \fill[#1] (0, -.35ex)
            -- ++(1.3em, 0)
            -- ++(0, 2.25ex)
            -- ++(-1.3em, 0)
            -- cycle;
		\node[anchor=base, inner sep=0] (label) at (.65em, 0) {\normalfont\color{black}\bfseries\sffamily #2};
	}}%
	\newcommand\BBox[2]{%
	\tikz[baseline=(label.base)]{
        \fill[#1] (0, -.55ex)
            -- ++(1.7em, 0)
            -- ++(0, 2.65ex)
            -- ++(-1.7em, 0)
            -- cycle;
		\node[anchor=base, inner sep=0] (label) at (.85em, 0) {\normalfont\color{black}\bfseries\sffamily #2};
    }}%
\begin{theorem}\label{thm:main}
Let $\psi$ be a $\MaxSP_k$ or $\MinSP_k$ graph formula. Assuming the Sparse \MAXThreeSAT{} hypothesis, $\psi$ belongs to precisely one of the following regimes: 
\setlength\leftmargini{1.7em}
\begin{enumerate}
	\labelsep.4em
	\smallskip
	\itemdesc[\RBox{bargreen}{R1}]{Efficiently optimizable:}\newline There is some $\delta > 0$ such that $\psi$ can be solved exactly in time $O(m^{k-\delta})$.
	\itemdesc[\RBox{baryellow}{R2}]{Admits an approximation scheme, but not an efficient one:}\newline 
	For all $\varepsilon > 0$, there is some $\delta > 0$ such that $\psi$ can be $(1+\varepsilon)$-approximated in time $O(m^{k-\delta})$. However, for all $\delta > 0$, there is some $\varepsilon > 0$ such that $\psi$ cannot be $(1+\varepsilon)$-approximated in time~$O(m^{k-\delta})$.
	\itemdesc[\RBox{barorange}{R3}]{Admits a constant-factor approximation, but no approximation scheme:}\newline
	There are $\varepsilon, \delta > 0$ such that $\psi$ can be $(1+\varepsilon)$-approximated in time $O(m^{k-\delta})$.
	However, there also exists an $\varepsilon > 0$ such that for all $\delta > 0$, $\psi$ cannot be $(1+\varepsilon)$-approximated in time $O(m^{k-\delta})$.
	\itemdesc[\RBox{barred}{R4}]{Arbitrary polynomial-factor approximation is best possible (maximization):}\newline
	For every $\varepsilon > 0$, there is some $\delta > 0$ such that $\psi$ can be $O(m^\varepsilon)$-approximated in time $O(m^{k-\delta})$. However, for every $\delta > 0$, there exists some $\varepsilon > 0$ such that $\psi$ cannot be $O(m^\varepsilon)$-approximated in time $O(m^{k-\delta})$.
	\smallbreak\noindent{\color{lipicsGray}\normalfont\bfseries\sffamily No approximation at all (minimization):}\newline
	For all $\delta > 0$, we cannot decide whether the optimum value of $\psi$ is $0$ or at least $1$ in time $O(m^{k-\delta})$.
\end{enumerate}
\end{theorem}


\begin{figure}[p]
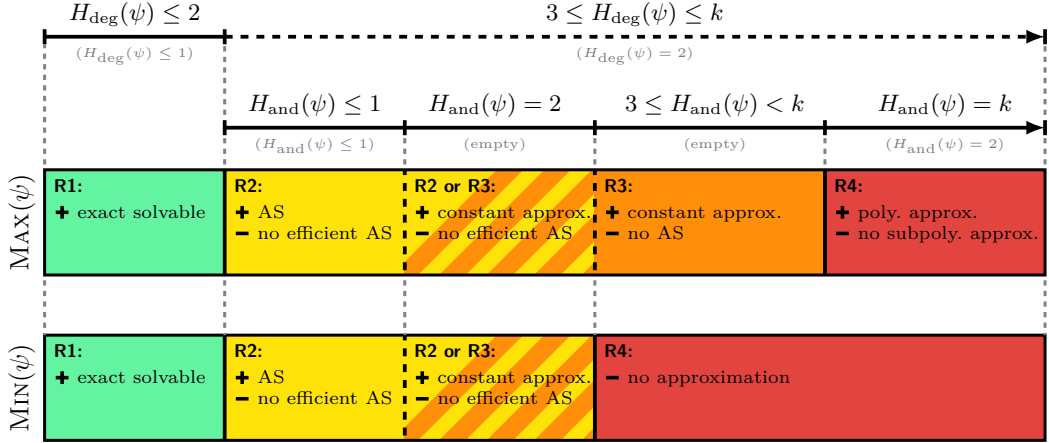

\begin{spectrumpicture}

\begin{scope}[upper]
    \drawmax
\end{scope}
\begin{scope}[lower]
    \drawmin
\end{scope}

\begin{scope}[yshift=.5 * \boxskip + \boxheight + \axisupperoffset]
    \draw[axis] (0, 0)
        -- node[above] {$\Hdeg(\psi) \leq 2$} node[below] {\note{$\Hdeg(\psi) \leq 1$}} ++(.18, 0);
    \draw[axis, dashed, ->] (.18, 0)
        -- node[above] {$3 \leq \Hdeg(\psi) \leq k$} node[below] {\note{$\Hdeg(\psi) = 2$}} ++(.82, 0);
    \drawaxismark{0}{0}
    \drawaxismark{.18}{0}
    \drawaxismark{1}{0}
\end{scope}
\begin{scope}[yshift=.5 * \boxskip + \boxheight + \axisloweroffset]
    \draw[axis, ->] (.18, 0)
    -- node[above] {$\Hand(\psi) \leq 1$} node[below] {\note{$\Hand(\psi) \leq 1$}} ++(.18, 0)
    -- node[above] {$\Hand(\psi) = 2$} node[below] {\note{empty}} ++(.18, 0)
    -- node[above] {$3 \leq \Hand(\psi) < k$} node[below] {\note{empty}} ++(.26, 0)
    -- node[above] {$\Hand(\psi) = k$} node[below] {\note{$\Hand(\psi) = 2$}} ++(.2, 0);
    \drawaxismark{.18}{0}
    \drawaxismark{.36}{0}
    \drawaxismark{.55}{0}
    \drawaxismark{.78}{0}
    \drawaxismark{1}{0}
\end{scope}

\draw[ind] (0, .5 * \boxskip + \boxheight)
    -- ++(0, \axisupperoffset - .5 * \axismarkheight);
\draw[ind] (.18, .5 * \boxskip + \boxheight)
    -- ++(0, \axisloweroffset - .5 * \axismarkheight)
    ++ (0, \axismarkheight)
    -- ++(0, \axisupperoffset - \axisloweroffset - \axismarkheight);
\draw[ind] (.36, .5 * \boxskip + \boxheight)
    -- ++(0, \axisloweroffset - .5 * \axismarkheight);
\draw[ind] (.55, .5 * \boxskip + \boxheight)
    -- ++(0, \axisloweroffset - .5 * \axismarkheight);
\draw[ind] (.78, .5 * \boxskip + \boxheight)
    -- ++(0, \axisloweroffset - .5 * \axismarkheight);
\draw[ind] (1, .5 * \boxskip + \boxheight)
    -- ++(0, \axisloweroffset - .5 * \axismarkheight)
    ++ (0, \axismarkheight)
    -- ++(0, \axisupperoffset - \axisloweroffset - \axismarkheight);

\foreach\x in {0, .18, .36, .55, 1} {
    \draw[ind] (\x, .5 * \boxskip)
        -- ++(0, -\boxskip);
};

\end{spectrumpicture}
\caption{Visualizes our classification of first-order optimization problems $\Max(\psi)$ and $\Min(\psi)$ for all $k \geq 3$, in terms of the hardness parameters $\Hand$ and $\Hdeg$ (as defined in Definitions~\ref{def:model-checking-hardness} and~\ref{def:optimization-hardness}). See~\autoref{def:approximation-scheme} for the precise definition of an (efficient) approximation scheme (AS). The pale labels indicate how to change the conditions to obtain the picture for $k = 2$.} \label{fig:spectrum}
\end{figure}

\begin{table}[p]
\caption{Some interesting examples classified according to \autoref{fig:spectrum}. For each problem $\psi$, an instance consists of $k$ sets of $n$ vectors $X_1, \dots, X_k \subseteq \{0, 1\}^d$. We write $\psi = \max_{x_1, \dots, x_k}/\min_{x_1, \dots, x_k} \phi$ and only list the inner formulas $\phi$ in the table. The computations of $\Hdeg(\psi)$ and $\Hand(\psi)$ are carried out in \autoref{sec:examples}.} \label{tab:examples}
\def\colorbullet#1{\textcolor{#1}{\rule{2em}{1.8ex}}}
\small
	\begin{minipage}{\textwidth}
\begin{tabular*}{\linewidth}{@{\extracolsep{\fill}}|llllc|}
	\hline
	\thead[l]{Problem $\psi$} & \thead[l]{$\phi$} & \thead[l]{$\Hdeg(\psi)$} & \thead[l]{$\Hand(\psi)$} & \thead[l]{Hardness regime} \\
	\hline
	\rule{0pt}{4.8ex}%
	\makecell[l]{Maximum/Minimum\\$2$-Agreement} & $\# \{ y : x_1[y] = x_2[y] \}$ & $2$ & $1$ & \makecell[r]{\BBox{baryellow}{R2}} \\[2.5ex]
	\makecell[l]{Maximum/Minimum\\$3$-Agreement} & $\# \{ y : x_1[y] = x_2[y] = x_3[y] \}$ & $2$ & $2$ & \makecell[r]{\BBox{bargreen}{R1}} \\[2.5ex]
	\makecell[l]{Maximum\\$k$-Agreement, $k \geq 4$} & $\# \{ y : x_1[y] = \cdots = x_k[y] \}$ & $2\lfloor{\frac k2}\rfloor$ & $k-1$ & \makecell[r]{\BBox{barorange}{R3}} \\[2.5ex]
	\makecell[l]{Minimum\\$k$-Agreement, $k \geq 4$} & $\# \{ y : x_1[y] = \cdots = x_k[y] \}$ & $2\lfloor{\frac k2}\rfloor$ & $k-1$ & \makecell[r]{\BBox{barred}{R4}} \\[2.5ex]
	\makecell[l]{Maximum/Minimum\\$k$-XOR} & $\# \{ y : x_1[y] \oplus \cdots \oplus x_k[y] \}$ & $k$ & $1$ & \makecell[r]{\BBox{baryellow}{R2}} \\[2.5ex]
	\makecell[l]{Maximum\\$k$-Inner Product} & $\# \{ y : x_1[y] \land \cdots \land x_k[y] \}$ & $k$ & $k$ & \makecell[r]{\BBox{barred}{R4}} \\[2.5ex]
	\makecell[l]{Minimum\\$k$-Inner Product} & $\# \{ y : x_1[y] \land \cdots \land x_k[y] \}$ & $k$ & $k$ & \makecell[r]{\BBox{barred}{R4}} \\
	\hline
\end{tabular*}
\end{minipage}
\end{table}

Note that the characteristics of the fourth (i.e., hardest) regime differ between the  maximization and the minimization case.

This theorem has immediate consequences for the approximability landscape in $\MaxSP_k$ and $\MinSP_k$, based on fine-grained assumptions: In particular, while any $\MaxSP_k$ graph formula can be approximated within a subpolynomial factor $O(m^{\varepsilon})$, we can rule out optimal approximation ratios that \emph{grow with $m$ but are strictly subpolynomial} (i.e., there are \emph{no} graph formulas whose optimal approximability within $O(m^{k-\Omega(1)})$ time is $\Theta(\log \log m)$, $\Theta(\log^2 m)$ or $2^{\Theta(\sqrt{\log m})}$). Furthermore, there are no graph formulas with an $f(1/\varepsilon)m^{k-\Omega(1)}$ approximation scheme that cannot already be optimized to exactness in time $m^{k-\Omega(1)}$.

In fact, beyond Theorem~\ref{thm:main} we even give an almost complete characterization of each regime: Specifically, we introduce integer-valued hardness parameters $0 \le \Hand(\psi) \le \Hdeg(\psi) \le k$ (defined in Section~\ref{sec:technical}). As illustrated in Figure~\ref{fig:spectrum}, we show how to place any graph formula $\psi$ into its corresponding regime depending solely on $\Hand(\psi)$ and $\Hdeg(\psi)$ -- with the \emph{single exception} of formulas $\psi$ with $\Hand(\psi) = 2$ and $\Hdeg(\psi) \ge 3$. For these formulas (e.g., Maximum Exact-3-Cover), it remains open whether they belong to Regime 2 or 3, i.e., whether or not they admit an approximation scheme (see~\autoref{open:exact-3-cover}).

In \autoref{tab:examples} we give natural problems for each regime (showing in particular that each regime is indeed non-empty). A particular highlight is that we can prove existence of \emph{constant-factor approximable formulas that do not admit an approximation scheme}, such as Maximum-4-Agreement (see Section~\ref{sec:technical}, Theorem~\ref{thm:as-maximization-classification} for a technical discussion of the lower bound). We give the details on how to calculate the hardness parameters $\Hand(\psi)$ and $\Hdeg(\psi)$ for each example in \autoref{tab:examples} in~\autoref{sec:examples}.

While the approximability of problems in the third and fourth regime seem rather unsatisfactory, we also consider the setting of additive approximation, and show that for every $\MaxSP_k$ and $\MinSP_k$ graph formula, there is an additive approximation scheme; however, assuming the Sparse MAX-3-SAT hypothesis it cannot be an efficient one unless we can optimize the problem exactly.

\begin{theorem}[Additive Approximation]\label{thm:additive-approximation}
For every $\psi$, we give an additive approximation scheme, i.e., for every $\varepsilon>0$, there is a $\delta>0$ such that we compute the optimum value up to an additive $\varepsilon |Y|$ error in time $O(m^{k-\delta})$.

If $\psi$ does not belong to Regime 1, we show that unless the Sparse \MAXThreeSAT{} hypothesis fails, for every $\delta>0$, there is some $\varepsilon > 0$ such that we cannot compute the optimum value up to an additive $\varepsilon |Y|$ error in time $O(m^{k-\delta})$.
\end{theorem}

Finally, we remark that our classification identifies general trends in $\MaxSP_k$ and $\MinSP_k$, including that exact optimizability is described by a simple algebraic criterion, and that constant-factor approximation for minimization problems is equivalent to testing whether the optimum is $0$. We address these trends in Section~\ref{sec:technical}.

\subparagraph*{On Plausibility of the Hypotheses}
As tight unconditional lower bounds for polynomial-time problems are barely existent, we give \emph{conditional} lower bounds,  based on established assumptions in fine-grained complexity theory, such as SETH (see~\cite{VassilevskaW18} for an excellent survey). The essentially only exception is the only recently introduced \emph{Sparse} \MAXThreeSAT{} hypothesis~\cite{AlmanW20}; we use this hypothesis only for giving evidence against approximation schemes and for ruling out certain additive approximation schemes. While it is possible that this hypothesis could ultimately be refuted, our classification describes the frontier of the current state of the art. At the very least, our conditional lower bound for approximation schemes reveals a rather surprising connection: To obtain an \emph{approximation scheme} for \emph{polynomial-time problems} such as Maximum $4$-Agreement, we need to give an \emph{exponential}-time improvement for \emph{exact} solutions for Sparse \MAXThreeSAT{}!

\subparagraph*{On the Use of Fast Matrix Multiplication}
From our hardness parameters, it is easy to read off which of our algorithms require the use of fast matrix multiplication. In fact, this connection also works in the other direction: Since fast matrix multiplication techniques are at times deemed too sophisticated to be used in practice, we might want to disallow the use of these techniques. If we do this, it is plausible to assume that also Sparse \MAXkSAT{2} has no $O(2^{n(1-\delta)})$-time algorithms, and the remaining gap in our classification disappears! Specifically, any formula for which we know that it belongs to regime 2 or 3 would be placed into regime 3.

At the very least, this connection suggests that to close the remaining gap of our classification, we need to use fast matrix multiplication.

\subsection{Further Related Work}\label{sec:introduction:sec:related-work}

\subparagraph*{Parameterized Inapproximability}
The problem settings we consider are related to a recent and strong line of research on parameterized inapproximability, including~\cite{ChalermsookCKLM17, KarthikLM19, Cohen-AddadGKLL19, Manurangsi20, Lin21} (see~\cite{FeldmannKLM20} for a recent survey). In these contexts, one seeks to determine optimal approximation guarantees within some running time of the form $f(k)n^{g(k)}$ (such as FPT running time $f(k)\poly(n)$ or running time $n^{o(k)}$) for some parameter $k$ (such as the solution size). Unfortunately, many of these results do not immediately establish hardness for specific values of $k$. An important exception is work by Karthik, Laekhanukit, and Manurangsi~\cite{KarthikLM19}, which among other results establishes inapproximability of $k$-Dominating Set and the Maximum $k$-Inner Product within running time $O(n^{k-\varepsilon})$, assuming SETH. We give a detailed comparison of our setting to notions used in Karthik et al.'s work in Section~\ref{sec:psp}.

\subparagraph*{Fine-Grained Complexity of First-Order Properties}
Studying the fine-grained complexity of polynomial-time problem classes defined by first-order properties has been initiated in~\cite{Williams14, GaoIKW18}, with different settings considered in~\cite{GaoI19, Gao19}. In particular, for model-checking first-order properties, \cite{GaoIKW18} provides a completeness result (a fine-grained analogue of the Cook-Levin Theorem) and \cite{BringmannFK19} provides a classification theorem (a fine-grained analogue of Schaefer's Theorem, see also Section~\ref{sec:mc-dichotomy}).

For \emph{optimization} in P, \cite{BringmannCFK21} provides completeness theorems for $\MaxSP_k$ (a fine-grained analogue of \MaxSNP-completeness of \MAXThreeSAT~\cite{PapadimitriouY91}). In contrast, the present paper gives a classification theorem (building a fine-grained analogue of the approximability classifications of optimization variants of CSPs~\cite{Creignou95, KhannaSTW00}).

\section{Preliminaries} \label{sec:preliminaries}
For an integer $k \geq 0$, we set $[k] = \{1, \ldots, k\}$. For a set $A$ and integer $k \geq 0$ we denote by $\binom Ak$ the set of all size-$k$ subsets of $A$. Let $\phi(z_1, \ldots, z_k)$ a Boolean function, and let $S \subseteq [k]$. Any function obtained by instantiating all variables $z_i$ ($i \not\in S$) in $\phi$ by constant values is called an \emph{$S$-restriction} of $\phi$. Finally, we write $\tilde\Order(\cdot)$ to hide poly-logarithmic factors and denote by $\omega < 2.373$ the exponent of square matrix multiplication~\cite{AlmanW21}.

\subsection{First-Order Model-Checking}
A \emph{relational structure} consists of $n$ objects and relations $R_1, \dots, R_\ell$ (of arbitrary arities) between these objects. A \emph{first-order formula} is a quantified formula of the form
\begin{equation*}
  \psi = (Q_1 x_1)\, \ldots\, (Q_{k+1} x_{k+1})\, \phi(x_1, \ldots, x_{k+1}),
\end{equation*}
where $\phi$ is a Boolean formula over the predicates $R(x_{i_1}, \dots, x_{i_\ell})$ and each $Q_i$ is either a universal or existential quantifier. Given a $(k+1)$-partite structure on objects $X_1 \union \dots \union X_{k+1}$, the \emph{model-checking problem} (or \emph{query evaluation problem}) is to check whether $\psi$ holds on the given structure, that is, for $x_i$ ranging over $X_i$ and by instantiating the predicates $R(x_{i_1}, \dots, x_{i_\ell})$ in $\phi$ according to the structure, $\psi$ is valid.

Following previous work in this line of research~\cite{GaoIKW18,BringmannFK19}, we usually assume that the input is represented sparsely -- that is, we assume that the relational structure is written down as an exhaustive enumeration of all records in all relations; let $m$ denote the total number of such entries\footnote{By ignoring objects not occurring in any relation, we may always assume that $n \leq \Order(m)$.}. The convention is reasonable as this data format is common in the context of database theory and also for the representation of graphs (where it is called the \emph{adjacency list} representation), see Section~\ref{sec:conclusion} for a further discussion.

A first-order formula $\psi$ is called a \emph{graph formula} if it is defined over a single binary predicate $E(x_i, x_j)$. Many natural problems fall into this category; see~\cite{BringmannFK19} for a detailed discussion on the subject.

\subsection{\texorpdfstring{$\MaxSP_k$ and $\MinSP_k$}{MaxSPk and MinSPk}}
In analogy to first-order properties with quantifier structure $\exists^k\forall$ (with maximization instead of $\exists$ and counting instead of $\forall$) and following the definition in~\cite{BringmannCFK21}, we now introduce the class of optimization problems. Let $\MaxSP_k$ be the class containing all formulas of the form 
\begin{equation} \label{eq:psi}
  \psi = \max_{x_1, \dots, x_k} \counting_y\, \phi(x_1, \ldots, x_k, y),
\end{equation}
where, as before, $\phi$ is a Boolean formula over some predicates $R(x_{i_1}, \dots, x_{i_\ell})$. We similarly define $\MinSP_k$ with ``$\min$'' in place of ``$\max$''. Occasionally, we write $\OptSP_k$ to refer to both of these classes simultaneously, and we write ``$\opt$'' as a placeholder for either ``$\max$'' or ``$\min$''. In analogy to the model-checking problem for first-order properties, we associate to each formula $\psi \in \OptSP_k$ an algorithmic problem:

\begin{definition}[$\Max(\psi)$ and $\Min(\psi)$] \label{def:optimization}
Let $\psi \in \MaxSP_k$ be as in~\eqref{eq:psi}. Given a $(k+1)$-partite~structure on objects $X_1 \union \dots \union X_k \union Y$, the $\Max(\psi)$ problem is to compute
\begin{equation*}
  \OPT = \max_{x_1 \in X_1, \dots, x_k \in X_k} \,\counting_{y \in Y}\, \phi(x_1, \dots, x_k, y).
\end{equation*}
We similarly define $\Min(\psi)$ for $\psi \in \MinSP_k$. Occasionally, for $\psi \in \OptSP_k$, we write $\Opt(\psi)$ to refer to both problems simultaneously.
\end{definition}

For convenience, we introduce some further notation: For objects $x_1 \in X_1, \dots, x_k \in X_k$, we denote by $\Val(x_1, \dots, x_k) = \counting_{y \in Y} \phi(x_1, \dots, x_k, y)$ the \emph{value} of $(x_1, \dots, x_k)$.

The problem $\Opt(\psi)$ can be solved in time $\Order(m^k)$ for all $\OptSP_k$ formulas $\psi$, by a straightforward extension of the model-checking baseline algorithm; see~\cite{BringmannCFK21} for details. As this is clearly optimal for $k = 1$, we will often implicitly assume that $k \geq 2$ in the following.

For a clean analogy between model-checking and the optimization classes $\MaxSP_k$ and $\MinSP_k$, we will from now view model-checking as ``testing for zero''. More precisely, the model-checking problem of $\exists x_1\, \dots\, \exists x_k\, \forall y\, \neg\phi(x_1, \dots, x_k, y)$ is equivalent to testing whether $\Min(\psi)$ has optimal solution $\OPT = 0$, where $\psi = \min_{x_1, \dots, x_k} \counting_y \phi(x_1, \dots, x_k, y)$. We refer to the latter problem as $\Zero(\psi)$.

\autoref{def:optimization} introduces $\Max(\psi)$ and $\Min(\psi)$ as \emph{exact} optimization problems (i.e., $\OPT$ is required to be computed exactly). We say that an algorithm computes a \emph{(multiplicative) $c$\=/approximation} for $\Max(\psi)$ if it computes any value in the interval $[c^{-1} \mult \OPT, \OPT]$. Similarly, a (multiplicative) $c$\=/approximation for $\Min(\psi)$ computes a value in $[\OPT, c \mult \OPT]$.

\section{Technical Overview} \label{sec:technical}
This section serves the purpose of stating our results formally, to provide the main proof ideas and techniques, and to convey some intuition whenever possible. We postpone the precise hypotheses for our conditional lower bounds along with a discussion on their plausibility to \autoref{sec:assumptions}.

We first introduce our hardness parameters: the and-hardness $\Hand(\psi)$ borrowed from previous work~\cite{BringmannFK19} (reviewed below), as well as a novel algebraic parameter that we call degree-hardness $\Hdeg(\psi)$. In the subsequent sections, we give an overview over our various algorithmic and hardness results based on the values $\Hdeg(\psi), \Hand(\psi)$ (see Figure~\ref{fig:spectrum}), with the formal proof of Theorem~\ref{thm:main} given at the end of this section.

\subsection{Bringmann et al.'s Model-Checking Dichotomy}\label{sec:mc-dichotomy}
Bringmann, Fischer and Künnemann~\cite{BringmannFK19} established a fine-grained classification of all $\exists^k\forall$-quantified graph properties into computationally easy and hard model-checking problems. As our work extends that classification (and also since our results are of a similar flavor), we briefly summarize their results. The hardness parameter presented in~\autoref{def:model-checking-hardness} forms the basis for the dichotomy.

Here, and for the remainder of this section, we write $\psi_0$ to denote the Boolean function obtained from $\psi$ in the following way: Let $\psi = \opt_{x_1, \dots, x_k} \counting_y \phi(x_1, \ldots, x_k, y)$ be any graph formula, where $\phi$ is a propositional formula over the atoms $E(x_i, x_j)$~($i, j \in [k]$) and $E(x_i, y)$~($i \in [k]$). Consider the Boolean function obtained from $\phi$ by replacing every atom $E(x_i, x_j)$~($i, j \in [k]$) by $\false$. What remains is a Boolean function over the $k$ atoms $E(x_i, y)$~($i \in [k])$ and we denote this formula by $\psi_0$. For example, if $\psi = \max_{x_1, x_2} \counting_y (\neg E(x_1, x_2) \land E(x_1, y) \land E(x_2, y))$ then we define $\psi_0 : \{0, 1\}^2 \to \{0, 1\}$ by $\psi_0(a_1, a_2) = a_1 \land a_2$. 

As apparent in~\autoref{def:model-checking-hardness}, the core difficulty of a formula $\psi$ is captured by $\psi_0$ and thus not affected by the predicates $E(x_i, x_j)$. We elaborate on this phenomenon in \autoref{sec:vector}.

\begin{definition}[And-Hardness] \label{def:model-checking-hardness}
Let $\phi$ be a Boolean function on $k$ inputs. The \emph{and-hardness $\Hand(\phi)$} of $\phi$ is the largest integer $0 \leq h \leq k$ such that, for any index set \raisebox{0pt}[0pt][0pt]{$S \in \binom{[k]}{h}$}, there exists some $S$\=/restriction of $\phi$ with exactly one satisfying assignment. (Set $\Hand(\phi) = 0$ for constant-valued $\phi$.) For an $\OptSP_k$ graph formula $\psi$, we define $\Hand(\psi) = \Hand(\psi_0)$.
\end{definition}

This hardness parameter essentially specifies the computational hardness of the model-checking problems $\Zero(\phi)$; here, $\Hand(\psi) \leq 2$ is the critical threshold:

\begin{theorem}[Model-Checking, \cite{BringmannFK19}] \label{thm:model-checking-classification}
Let $\psi$ be a $\MinSP_k$ graph formula.
\begin{itemize}
\item If $\Hand(\psi) \leq 2$ and $\Hand(\psi) < k$, then $\Zero(\psi)$ can be solved in time $\Order(m^{k-\delta})$ for some $\delta > 0$.
\item If $3 \leq \Hand(\psi)$ or $\Hand(\psi) = k$, then $\Zero(\psi)$ cannot be solved in time $\Order(m^{k-\delta})$ for any $\delta > 0$ unless the \MAXThreeSAT{} hypothesis fails.
\end{itemize}
\end{theorem}

\subsection{Exact Optimization}
We are now ready to detail our results. Our main contribution is a dichotomy for the exact solvability and approximability of $\Max(\psi)$ and $\Min(\psi)$ for all $\MaxSP_k$ and $\MinSP_k$ graph formulas~$\psi$. For the exact case, the decisive criterion for the hardness of some formula $\psi$ can be read off the polynomial \emph{extension} of the function $\psi_0$. Specifically, for any Boolean function $\phi$ there exists a (unique) multilinear polynomial with real coefficients that computes~$\phi$ on binary inputs. By abuse of notation, we refer to that polynomial by writing $\phi$ as well. The \emph{degree $\deg(\phi)$} of $\phi$ is the degree of its polynomial extension.

As an example, consider the Exact-$3$-Cover property: $\phi(z_1, z_2, z_3)$ is true if and only if exactly one of its inputs $z_1$, $z_2$ or $z_3$ is true. Then its unique multilinear polynomial extension is
\begin{equation*}
  \phi(z_1, z_2, z_3) = 3z_1 z_2 z_3 -2(z_1 z_2 + z_2 z_3 + z_3 z_1) + (z_1 + z_2 + z_3),
\end{equation*}
and therefore $\deg(\phi) = 3$.

\begin{restatable}[Degree Hardness]{definition}{defoptimizationhardness} \label{def:optimization-hardness}
Let $\phi$ be a Boolean function on $k$ inputs. The \emph{degree hardness $\Hdeg(\phi)$} of $\phi$ is the largest integer $0 \leq h \leq k$ such that, for any index set \raisebox{0pt}[0pt][0pt]{$S \in \binom{[k]}{h}$}, there exists some $S$\=/restriction of $\phi$ of degree~$h$. For an $\OptSP_k$ graph formula $\psi$, we define $\Hdeg(\psi) = \Hdeg(\psi_0)$.
\end{restatable}

It always holds that $\Hand(\psi) \leq \Hdeg(\psi)$ (see~\autoref{prop:hardness-relation}), but in general these parameters behave very differently (see~\autoref{sec:examples} for a list of examples). With $\Hopt$ in place of $\Hand$, we are able to recover the same classification as~\autoref{thm:model-checking-classification} for both exact maximization and minimization:

\begin{theorem}[Exact Optimization] \label{thm:exact-opt-classification}
Let $\psi$ be an $\OptSP_k$ graph formula.
\begin{itemize}
\item If $\Hdeg(\psi) \leq 2$ and $\Hdeg(\psi) < k$, then $\Opt(\psi)$ can be solved in time $\Order(m^{k-\delta})$ for some $\delta > 0$.
\item If $3 \leq \Hdeg(\psi)$ or $\Hdeg(\psi) = k$, then $\Opt(\psi)$ cannot be solved in time $\Order(m^{k-\delta})$ for any $\delta > 0$ unless the \MAXThreeSAT{} hypothesis fails.
\end{itemize}
\end{theorem}

We will formally prove~\autoref{thm:exact-opt-classification} in~\autoref{sec:exact}. The algorithmic part is along the same lines as~\cite{BringmannFK19}: We first brute-force over all but $k = 3$ quantifiers, and solve the remaining problem by a reduction to maximum-weight triangle detection with small edge weights. Intuitively, since the degree of the resulting $3$-variable problem is at most $2$, we can express the objective as a sum of three parts depending on two variables each (this approach is also used in a similar context in~\cite[Chapter 6.5]{Williams07}). This allows us to label the edges of a triangle instance with corresponding parts, which can be shown to yield average edge weight $\Order(1)$. 

The conditional lower bound is more interesting. Our reduction is inspired by a standard argument proving quadratic-time hardness of the Maximum Inner Product problem (\MaxIP{}). That lower bound is based on the \OV{} hypothesis, so it consists of a reduction from an \OV{} instance $X_1, X_2 \subseteq \{0, 1\}^d$ to an instance of \MaxIP{}. The idea is to use a \emph{gadget} that maps every entry in $x_i \in X_i$ to a constant number of new entries;\footnote{Every entry $a$ in $x_1$ is replaced by three new coordinates $(a, 1-a, 1-a)$ and every entry $b$ in $x_2$ is replaced by $(1-b, b, 1-b)$. The contribution to the inner product of the new vectors is equal to $a(1-b) + (1-a)b + (1-a)(1-b) = 1 - ab$.} let $x_i' \in \{0, 1\}^{\Order(d)}$ denote the vector after applying the gadget coordinate-wise. The crucial property is that $\innerprod{x_1'}{x_2'} = d - \innerprod{x_1}{x_2}$, and thus a pair of orthogonal vectors $x_1, x_2$ corresponds to a pair of vectors $x_1', x_2'$ of maximum inner product.

To mimic the reduction for all problems which are hard in the sense of~\autoref{thm:exact-opt-classification}, we settle for the weaker but sufficient property that the value of $(x_1', x_2')$ equals $\beta_1 d - \beta_2 \innerprod{x_1}{x_2}$, for some positive integers~$\beta_1, \beta_2$. It follows from our algebraic characterization of hard functions that a gadget with such guarantees always exists (\autoref{lem:coordinate-gadget}). Ultimately, our hardness proof makes use of that gadget in a similar way as for \MaxIP{} (Lemmas~\ref{lem:exact-hardness-ov} and~\ref{lem:exact-hardness-hyperclique}).

The hardness part of~\autoref{thm:exact-opt-classification} can in fact be stated in a more fine-grained way: If some problem $\Opt(\psi)$ has degree hardness $h = \Hdeg(\psi) \geq 3$, then the hardness proof can be conditioned on the weaker \hUniformHyperClique{h} assumption. That connection can be complemented by a partial converse, thus revealing a certain \emph{equivalence} between exact optimization and hyperclique detection. An analogous equivalence for model-checking graph formulas could not be proved and was left as an open problem in~\cite{BringmannFK19}. We postpone the proof of \autoref{thm:exact-equivalence} to \autoref{sec:exact-equivalence}.

\begin{restatable}[Equivalence of $\Opt(\psi)$ and \HyperClique{}]{theorem}{thmexactequivalence} \label{thm:exact-equivalence}
Let $\psi$ be an $\OptSP_k$ graph formula of degree hardness $h = \Hdeg(\psi) \geq 2$.
\begin{itemize}
\item If $\Opt(\psi)$ can be solved in time $\Order(m^{k-\delta})$ for some $\delta > 0$, then, for some (large) $k' = k'(k, h, \delta)$, \hUniformkHyperClique{h}{k'} can be solved in time $\Order(n^{k'-\delta'})$ for some $\delta' > 0$.
\item If \hUniformkHyperClique{h}{(h+1)} can be solved in time $\Order(n^{h+1-\delta})$ for some $\delta > 0$, then $\Opt(\psi)$ can be solved in time $\Order(m^{k-\delta'})$ for some $\delta' > 0$.
\end{itemize}
\end{restatable}

\subsection{Approximation}
Since \autoref{thm:exact-opt-classification} gives a complete classification for the exact solvability of $\Opt(\psi)$, the next natural question is to study the approximability of properties which are hard to compute exactly. Unlike the exact case, we use different techniques and tools to give our classification for maximization and minimization problems.

\subsubsection{Maximization}
We obtain a simple classification of all constant-factor approximable maximization problems, conditioned on the Strong Exponential Time Hypothesis (\SETH{}).

\begin{theorem}[Constant Approximation -- Maximization] \label{thm:constant-approx-max-classification}
Let $\psi$ be a $\MaxSP_k$ graph formula.
\begin{itemize}
\item If $\Hand(\psi) < k$ (or equivalently, if $\psi_0$ does not have exactly one satisfying assignment), then there exists a constant-factor approximation for $\Max(\psi)$ in time $\Order(m^{k-\delta})$ for some $\delta > 0$.
\item Otherwise, if $\Hand(\psi) = k$ (or equivalently, $\psi_0$ has exactly one satisfying assignment), 
then there exists no constant-factor approximation for $\Max(\psi)$ in time $\Order(m^{k-\delta})$ for any $\delta > 0$, unless \SETH{} fails.
\end{itemize}
\end{theorem}

The algorithmic part of~\autoref{thm:constant-approx-max-classification} is detailed in~\autoref{sec:approximation-algs-maximization:sec:constant-factor}, and the lower bound in~\autoref{sec:hardness-approx:sec:mip-hardness}. Here we give a high-level explanation of our proof by focusing on two representative problems: On the one hand, strong conditional hardness results have been shown for the \MaxIP{} problem~\cite{AbboudRW17,Chen18}. When $\psi_0$ has a single satisfying assignment, $\Max(\psi)$ is equivalent (up to complementation) to \MaxIP{}, so we adapt the hardness results for our setting (\autoref{sec:hardness-approx:sec:mip-hardness}).

On the other hand, consider the \FurthestNeighbor{} problem: Given two sets of bit-vectors $X_1, X_2 \subseteq \{0, 1\}^d$, compute the maximum Hamming distance between vectors $x_1 \in X_1$ and $x_2 \in X_2$. There exists a simple linear-time $3$-approximation for this problem: Fix some $x_1 \in X_1$ and compute its furthest neighbor $x_2 \in X_2$. Then compute the furthest neighbor $x_1' \in X_1$ of $x_2$ and return the distance between $x_1'$ and $x_2$ as the answer. By applying the triangle inequality twice, it is easy to see that this indeed yields a $3$-approximation. That argument generalizes for approximating $\Max(\phi)$ whenever $\psi_0$ satisfies the following property: If $\alpha$ is a satisfying assignment of $\psi_0$, then the component-wise negation of $\alpha$ is also satisfying (\autoref{lem:constant-approx-max-opposite}). Finally, if $\psi_0$ has at least two satisfying assignments, then $\Max(\psi)$ can be reduced to this special case via a reduction which worsens the approximation ratio by at most a constant factor (Lemmas~\ref{lem:constant-approx-max-exactly-two} and~\ref{lem:constant-approx-max-geq-two}). The essential insight for that last step is that we can always ``cover'' all satisfying assignments by only two satisfying assignments, as there always exists one satisfying assignment which contributes a constant fraction (depending on $k$) to the optimal value.

We give a finer-grained view of the classification in~\autoref{thm:constant-approx-max-classification} in two ways. First, we want to isolate properties which admit \emph{arbitrarily good} constant-factor approximations. We make that notion precise in the following definition:

\begin{definition}[Approximation Scheme]\label{def:approximation-scheme}
Let $\psi$ be an $\OptSP_k$ formula. We say that $\Opt(\psi)$ admits an \emph{approximation scheme} if for any $\varepsilon > 0$ there exists some $\delta > 0$ and an algorithm computing a $(1 + \varepsilon)$-approximation of $\Opt(\psi)$ in time $\Order(m^{k-\delta})$.\footnote{We stress that in our work, $\varepsilon$ is always a constant and cannot depend on $m$.}
\end{definition}

In the following theorem, we identify some formulas which admit such an approximation scheme, and some formulas for which this is unlikely:

\begin{theorem}[Approximation Scheme -- Maximization]\label{thm:as-maximization-classification}
Let $\psi$ be a $\MaxSP_k$ graph formula.
\begin{itemize}
\item If $\Hand(\psi) \leq 1$, then there exists a randomized approximation scheme for $\Max(\psi)$.
\item If $3 \leq \Hand(\psi)$ or $\Hand(\psi) = k$, then there exists no approximation scheme for $\Max(\psi)$ unless the Sparse \MAXThreeSAT{} hypothesis fails.
\end{itemize}
\end{theorem}

Unfortunately, we were not able to close the gap in~\autoref{thm:as-maximization-classification}, and it remains open whether problems with $\Hand(\phi) = 2$ admit approximation schemes. In~\autoref{sec:conclusion}, we give a specific example falling into this category.

For the first item, we give approximation schemes for any $\Max(\psi)$ with $\Hand(\psi) = 1$ via a reduction to \FurthestNeighbor{}, which, as mentioned in the introduction, admits an approximation scheme (\autoref{lem:max-approximation-scheme}). 

The lower bound is more interesting: By~\autoref{thm:constant-approx-max-classification} we know that if $\psi$ has and-hardness $\Hand(\psi) < k$, then it admits a constant-factor approximation. Nevertheless, the lower bound in \autoref{thm:constant-approx-max-classification} only addresses formulas of and-hardness $\Hand(\psi) = k$. In particular, \autoref{thm:as-maximization-classification} identifies a class of problems for which some (fixed) constant-factor approximation is \emph{best-possible in terms of approximation.}

At a technical level, we make use of interesting machinery to obtain the lower bound (\autoref{sec:hardness-approx:sec:k-3-hardness}). The starting point is the \emph{distributed PCP} framework, which was introduced in~\cite{AbboudRW17} and further strengthened in~\cite{Rubinstein18,Chen18} to give hardness of approximation for Maximum $k$-Inner Product (\kMaxIP{k}). In this case, we cannot use this tool directly, since it crucially relies on the fact that the \emph{target} problem $\Max(\psi)$ has full hardness $\Hand(\psi) = k$.\footnote{More precisely, the reduction exploits the fact that $\psi_0$ has a unique satisfying assignment to encode the communication protocol used in the reduction (\autoref{lem:pcp-reduction}).} Instead, we first show \MAXThreeSAT{}-hardness of the following intermediate problem:
\begin{problem}[\kOV{(k,3)}]
Given sets of $n$ vectors $X_1, \ldots, X_k \subseteq \{0, 1\}^d$, where each coordinate \makebox{$y \in [d]$} is associated to three \emph{active} indices $a, b, c \in [k]$, detect if there are vectors $x_1 \in X_1, \ldots, x_k \in X_k$ such that for all $y \in [d]$, it holds that $x_a[y] \mult x_b[y] \mult x_c[y] = 0$ where $a, b, c$ are the active indices at $y$.
\end{problem}

We then provide a gap introducing reduction from $\kOV{(k,3)}$ to $\Max(\psi)$, in the same spirit as the PCP reduction gives such a reduction from $\kOV{k}$ to $\kMaxIP{k}$. This involves several technical steps as outlined in~\autoref{fig:reductions}. At a high level, we decompose a $\kOV{(k,3)}$ instance into a combination of multiple $\kOV{3}$ instances and use the PCP reduction as a black box on each of these. After combining the outputs of the reduction, we obtain instances of $\Max(\psi)$ with the desired gap. The issue with this approach is that the PCP reduction blows up the dimension of the input vectors exponentially,\footnote{An instance of $\kOV{k}$ on dimension $d = c \log n$ is reduced to multiple instances of $\kMaxIP{k}$ on dimension $\exp(c)\log n$.} which makes the reduction inapplicable to our case if we start from a moderate-dimensional \kOV{(k,3)} instance. To show that $\kOV{(k,3)}$ does not even have a $\Order(m^{k-\delta})$-time algorithm when the dimension is $d = O(\log n)$, we use the stronger \emph{Sparse} \MAXThreeSAT{} hypothesis.\footnote{Morally, just as \SETH{} implies the hardness of low-dimensional \OV{} (see~\autoref{sec:assumptions}), the \emph{Sparse} \MAXThreeSAT{} Hypothesis implies the hardness of \emph{low-dimensional \kOV{(k,3)}}.} See~\autoref{sec:assumptions} for further discussion of this hypothesis.


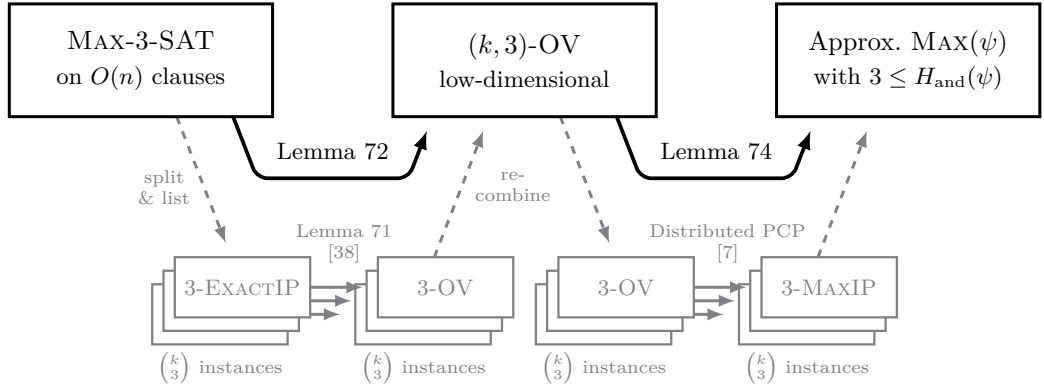
\begin{figure}[t]
\begin{center}

\def\probcolor{black}
\def\probfont{\normalfont}
\def\probstrength{.04cm}
\def\probalign{center}
\def\probwidth{3.5cm}
\def\probheight{1.5cm}
\def\probskip{(.5 * \linewidth - .5 * 3.6cm - .1cm)}

\def\auxcolor{lipicsLineGray}
\def\auxfont{\footnotesize}
\def\auxstrength{.03cm}
\def\auxalign{center}
\def\auxwidth{1.8cm}
\def\auxheight{.8cm}
\def\auxshiftx{.15cm}
\def\auxshifty{0.18cm}

\def\redcolor{black}
\def\redauxcolor{\auxcolor}
\def\redstrength{.05cm}
\def\redauxstrength{.04cm}
\def\redshorten{.2cm}
\def\redarrow{latex}
\def\reddown{2cm}
\def\redradius{.5cm}

\def\auxoffsetx{1.2cm}
\def\auxoffsety{3.2cm}

\def\redproboffset{.9cm}
\def\redparallel{.75cm}

\def\labelfont{\small}
\def\labelcolor{black}
\def\labelauxfont{\scriptsize}
\def\labelauxcolor{\auxcolor}
\def\labelalign{center}
\def\labeldistance{.13cm}
\def\labeleps{.15cm}

\def\a{((\auxoffsetx - .5 * \probwidth + \redproboffset + .5 * \redparallel) / (\auxoffsety - .5 * \auxheight - \auxshifty - .5 * \probheight))}
\def\b{1.4}

\begin{tikzpicture}[
    prob/.style={
        draw,
        color=\probcolor,
        font=\probfont,
        line width=\probstrength,
        align=\probalign,
        minimum width=\probwidth,
        minimum height=\probheight
    },
    reduction/.style={
        color=\redcolor,
        line width=\redstrength,
        >=\redarrow,
        shorten >=\redshorten,
        rounded corners
    },
    reduction aux/.style={
        reduction,
        color=\redauxcolor,
        line width=\redauxstrength
    },
    aux/.style={
        color=\auxcolor,
        line width=\auxstrength,
        minimum width=\auxwidth,
        minimum height=\auxheight,
    },
    aux caption/.style={
        color=\auxcolor,
        font=\auxfont,
        align=\auxalign,
    },
    label/.style={
        color=\labelcolor,
        font=\labelfont,
        align=\labelalign
    },
    label aux/.style={
        label,
        color=\labelauxcolor,
        font=\labelauxfont
    }
]

\def\drawred[#1]#2#3#4#5#6#7#8#9{
    \begin{scope}[#1]
        \draw[reduction, ->] ({.5 * \probwidth - \redproboffset + .5 * \redparallel}, {-.5 * \probheight})
            -- ({\auxoffsetx + (1 - \a * \b) * \redparallel}, {-\auxoffsety + .5 * \auxheight + \auxshifty + \b * \redparallel})
            -- ({\probskip - \auxoffsetx - (1 - \a * \b) * \redparallel}, {-\auxoffsety + .5 * \auxheight + \auxshifty + \b * \redparallel})
            -- ({\probskip - .5 * \probwidth + \redproboffset - .5 * \redparallel}, {-.5 * \probheight});
        \draw[reduction aux, dashed, ->] ({.5 * \probwidth - \redproboffset - .5 * \redparallel}, {-.5 * \probheight})
            -- ({\auxoffsetx}, {-\auxoffsety + .5 * \auxheight + \auxshifty});
        \draw[reduction aux, dashed, ->] ({\probskip - \auxoffsetx}, {-\auxoffsety + .5 * \auxheight + \auxshifty})
            -- ({\probskip - .5 * \probwidth + \redproboffset + .5 * \redparallel}, {-.5 * \probheight});
        \foreach\i in {-1,0,1}{
            \draw[reduction aux, ->] ({\auxoffsetx + .5 * \auxwidth + \i * \auxshiftx}, {-\auxoffsety + \i * \auxshifty})
                -- ({\probskip - \auxoffsetx - .5 * \auxwidth + \i * \auxshiftx}, {-\auxoffsety + \i * \auxshifty});
            \draw[aux, fill=white] ({\auxoffsetx + \i * \auxshiftx - .5 * \auxwidth}, {-\auxoffsety + \i * \auxshifty - .5 * \auxheight})
                rectangle ++(\auxwidth, \auxheight);
            \draw[aux, fill=white] ({\probskip - \auxoffsetx + \i * \auxshiftx - .5 * \auxwidth}, {-\auxoffsety + \i * \auxshifty - .5 * \auxheight})
                rectangle ++(\auxwidth, \auxheight);
        };
        \node[aux caption] at ({\auxoffsetx + \auxshiftx}, {-\auxoffsety + \auxshifty}) {#6};
        \node[aux caption] at ({\probskip - \auxoffsetx + \auxshiftx}, {-\auxoffsety + \auxshifty}) {#7};

        \node[label, above=\labeldistance] at ({.5 * \probskip}, {-\auxoffsety + .5 * \auxheight + \auxshifty + \b * \redparallel}) {#2};

        \node[label aux, left=\labeldistance] at ({\auxoffsetx - \a * \redparallel}, {-.5 * \probheight - .5 * (\auxoffsety - .5 * \auxheight - \auxshifty - .5 * \probheight)}) {#3};
        \node[label aux, right=\labeldistance] at ({\probskip - \auxoffsetx + \a * \redparallel}, {-.5 * \probheight - .5 * (\auxoffsety - .5 * \auxheight - \auxshifty - .5 * \probheight)}) {#5};
        \node[label aux, above=\labeldistance - \labeleps] at ({.5 * \probskip + \auxshiftx}, -\auxoffsety + .5 * \auxheight) {#4};
        \node[label aux, below=\labeldistance - \labeleps] at ({\auxoffsetx - \auxshiftx}, {-\auxoffsety - .5 * \auxheight - \auxshifty}) {#8};
        \node[label aux, below=\labeldistance - \labeleps] at ({\probskip - \auxoffsetx - \auxshiftx}, {-\auxoffsety - .5 * \auxheight - \auxshifty}) {#8};
    \end{scope}
}

\drawred[]
    {\autoref{lem:hardness_k3ov}}
    {split\\\& list}
    {\autoref{lem:red_exactip_to_ov}\\\cite{ChenW19}}
    {re-\\combine}
    {\kExactIP{3}}
    {\kOV{3}}
    {$\binom k3$ instances}
    {$\binom k3$ instances}
\drawred[xshift=\probskip]
    {\autoref{lem:reduction_k3ov_to_maxphi}}
    {}
    {Distributed PCP\\\cite{AbboudRW17}}
    {}
    {\kOV{3}}
    {\kMaxIP{3}}
    {$\binom k3$ instances}
    {$\binom k3$ instances}

\path (0, 0) node[prob] (prob0) {\MAXThreeSAT{}\\\small on $\Order(n)$ clauses}
    ++({\probskip}, 0) node[prob] (prob1) {\kOV{(k, 3)}\\\small low-dimensional}
    ++({\probskip}, 0) node[prob] (prob2) {Approx.\ $\Max(\psi)$\\\small with $3 \leq \Hand(\psi)$};

\end{tikzpicture}
\end{center}
\caption{The chain of reductions from Sparse \MAXThreeSAT{} to $\VMax(\psi)$. The main steps are proven in Lemmas~\ref{lem:hardness_k3ov} and~\ref{lem:reduction_k3ov_to_maxphi} and both proofs involve several intermediate steps as illustrated in the gray parts. In~\autoref{sec:hardness-approx} we briefly discuss an alternative approach to this reduction.} \label{fig:reductions}
\end{figure}

The second way in which we get a closer look at~\autoref{thm:constant-approx-max-classification} is by inspecting the hardest regime, i.e., when $\Hand(\psi) = k$:

\begin{theorem}[Polynomial-Factor Approximation]\label{thm:tight-polynomial-maximization}
Let $\psi$ be a $\MaxSP_k$ graph formula of full and-hardness $\Hand(\psi) = k$.
\begin{itemize}
\item For every $\varepsilon > 0$, there exists some $\delta > 0$ such that an $m^\varepsilon$-approximation for $\Max(\psi)$ can be computed in time $\Order(m^{k-\delta})$.
\item For every $\delta > 0$, there exists an $\varepsilon > 0$ such that there exists no $m^\varepsilon$-approximation for $\Max(\psi)$ in time $\Order(m^{k-\delta})$ unless \SETH{} fails.
\end{itemize}
\end{theorem}

The lower bound is obtained by applying the subsequent improvements on the distributed PCP framework by~\cite{Rubinstein18,Chen18,KarthikLM19}, which improve the parameters of the reduction via algebraic geometry codes and expander graphs (\autoref{lem:max-hardness-poly}). For the upper bound, we give a simple algorithm which exploits the sparsity of the instances (\autoref{lem:poly-approx-max}).

\subsubsection{Minimization}
There is an easier criterion for the hardness of approximating minimization problems. Namely, observe that giving a multiplicative approximation to a minimization problem $\Min(\psi)$ is at least as hard as testing if the optimal value is zero (recall that we refer to this problem as $\Zero(\psi)$). More precisely, suppose that we are given an instance with $\OPT = 0$. Then any multiplicative approximation must return an optimal solution.

It turns out that this is the only source of hardness for $\Min(\psi)$ problems. We show a fine-grained equivalence of deciding $\Zero(\psi)$ and approximating $\Min(\psi)$ within a constant factor:

\begin{theorem}[Constant Approximation is Equivalent to Testing Zero] \label{thm:constant-approx-min-classification}
Let $\psi$ be a $\MinSP_k$ graph formula. Via a randomized reduction, there exists a constant-factor approximation algorithm for $\Min(\psi)$ in time $\Order(m^{k-\delta})$ for some $\delta > 0$ if and only if $\Zero(\psi)$ can be solved in time $\Order(m^{k-\delta'})$ for some $\delta' > 0$.
\end{theorem}

To reduce approximating $\Min(\psi)$ to $\Zero(\psi)$, we make use of \emph{locality-sensitive hashing} (LSH); see \autoref{sec:approx-min} for details. This technique was for instance used to solve the Approximate Nearest Neighbors problem in Hamming spaces~\cite{Har-PeledIM12}, and was recently adapted by Chen and Williams to show that Minimum Inner Product can be reduced to \OV{}~\cite{ChenW19}. The latter result constitutes a singular known case for the general trend in $\MinSP_k$ revealed by~\autoref{thm:constant-approx-min-classification}. In comparison, our reduction is simpler and more general, but gives weaker guarantees on the constant factor.

We specifically use LSH for a reduction from approximating $\Min(\psi)$ to the intermediate problem of \emph{listing} solutions to $\Zero(\psi)$ (\autoref{lem:lsh}) -- the better the listing algorithm performs, the better the approximation guarantee. For instance, an algorithm listing~$L$ solutions in time $\Order(m^{k-\delta} \mult L^{\delta/k})$ for some $\delta > 0$ results in a constant-factor approximation (\autoref{cor:lsh-constant-approx}), while a listing algorithm in time $\tilde\Order(m^{k-\delta} + L)$ leads to an approximation scheme (\autoref{cor:lsh-as}). To finish the proof of~\autoref{thm:constant-approx-min-classification}, we show that any $\Zero(\psi)$ problem with an $\Order(m^{k-\delta})$-time decider, for some $\delta > 0$, also admits a listing algorithm in time $\Order(m^{k-\delta} \mult L^{\delta/k})$ (\autoref{lem:zero-to-listing}).

Since the hardness of $\Zero(\psi)$ is completely classified~\cite{BringmannFK19}, we obtain the following dichotomy as a consequence of \autoref{thm:constant-approx-min-classification}:

\begin{corollary}[Constant Approximation -- Minimization] \label{cor:constant-approx-min-characterization}
Let $\psi$ be a $\MinSP_k$ graph formula.
\begin{itemize}
\item If $\Hand(\psi) \leq 2$ and $\Hand(\psi) < k$, then there exists a randomized constant-factor approximation algorithm for $\Min(\psi)$ in time $\Order(m^{k-\delta})$ for some $\delta > 0$.
\item If $3 \leq \Hand(\psi)$ or $\Hand(\psi) = k$, then computing any approximation for $\Min(\psi)$ in time $\Order(m^{k-\delta})$ for any $\delta > 0$ is not possible unless the \MAXThreeSAT{} hypothesis fails.
\end{itemize}
\end{corollary}

Similar to the maximization case, let us next consider approximation schemes for minimization problems. We can reuse the general framework outlined in the previous paragraphs to obtain an approximation scheme by giving an output-linear listing algorithm in time $\tilde\Order(m^{k-\delta} + L)$, for all formulas with and-hardness at most $\Hand(\psi) \leq 1$ (\autoref{lem:h1-to-listing}).

\begin{theorem}[Approximation Scheme -- Minimization]\label{thm:as-minimization-classification}
Let $\psi$ be a $\MinSP_k$ graph formula. If $\Hand(\psi) \leq 1$, then there exists a randomized approximation scheme for $\Min(\psi)$.
\end{theorem}

\subsection{Efficient (Multiplicative) Approximation Schemes}
Theorems~\ref{thm:as-maximization-classification} and~\ref{thm:as-minimization-classification} show that if $\Hand(\psi) \leq 1$, then we can give approximation schemes for $\Opt(\psi)$. We complement this result in~\autoref{sec:hardness-approx:sec:no-eas} by ruling out the existence of \emph{efficient} approximation schemes for most regimes. In the same way as~\autoref{def:additive-approximation-scheme}, we say that $\Opt(\psi)$ admits an efficient (multiplicative) approximation scheme if there is some fixed constant $\delta > 0$ such that for any $\varepsilon > 0$, a multiplicative $(1 + \varepsilon)$-approximation for $\Opt(\psi)$ can be computed in time $O(m^{k-\delta})$.

\begin{theorem}\label{thm:no-efficient-as}
Let $\psi$ be an $\OptSP_k$ graph formula. If $3 \leq \Hdeg(\psi)$ or $\Hdeg(\psi) = k$, then there exists no efficient approximation scheme for $\Opt(\psi)$ assuming the Sparse \MAXThreeSAT{} Hypothesis.
\end{theorem}

\subsection{Proving the Main Theorem}
We can now put things together to prove \autoref{thm:main}.
\begin{proof}[Proof of~\autoref{thm:main}]
  Let $\psi$ be a $\MaxSP_k$ formula with $k \geq 3$. We show how to classify $\psi$ into one of the four stated regimes with a case distinction based on the hardness parameters $\Hand(\psi)$ and $\Hdeg(\psi)$. This case distinction can also be read off~\autoref{fig:spectrum}.
  \begin{itemize}
    \item If $\Hdeg(\psi) \leq 2$: By \autoref{thm:exact-opt-classification}, $\psi$ is efficiently optimizable, so it lies in R1.
    \item Otherwise, it holds that $3 \leq \Hdeg(\psi) \leq k$. In this case, we make a further distinction:
    \begin{itemize}
      \item $\Hand(\psi) \leq 1$: By \autoref{thm:as-maximization-classification}, $\psi$ admits an approximation scheme but by \autoref{thm:no-efficient-as} not an efficient one, so it lies in R2.
      \item $\Hand(\psi) = 2$: By \autoref{thm:constant-approx-max-classification}, $\psi$ admits an efficient constant-factor approximation but by \autoref{thm:no-efficient-as}, it does not admit an efficient approximation scheme. Thus, depending on whether $\psi$ admits an approximation scheme or not, it lies in R2 or R3. (As mentioned below \autoref{thm:main}, this is the single case where we cannot place the formula in its precise regime, see also \autoref{open:exact-3-cover}.)
      \item $3 \leq \Hand(\psi) < k$: By \autoref{thm:constant-approx-max-classification}, $\psi$ admits an efficient constant factor approximation but by \autoref{thm:as-maximization-classification} it has no approximation scheme, so it lies in R3.
      \item $\Hand(\psi) = k$: By \autoref{thm:tight-polynomial-maximization}, $\psi$ admits an efficient polynomial-factor approximation, and this is best possible, so it lies in R4.
    \end{itemize}
  \end{itemize}
  The case distinctions for the case $k = 2$ and for the minimization variant are proven analogously.
\end{proof}

\section{Discussion and Open Problems} \label{sec:conclusion}
Our investigation reveals all possible approximability types (in better-than-exhaustive-search time) for general classes of polynomial-time optimization problems, namely graph formulas in $\MaxSP_k$ and $\MinSP_k$. Our results, which give an almost complete characterization, open up the following questions:
\begin{itemize}
	\item Can we extend our classification beyond graph formulas, i.e., when we allow more binary relations, or even higher-arity relations? Such settings include, e.g., generalizations of Max $k$-XOR from $\mathbb{F}_2$ to $\mathbb{F}_q$, or variants of the densest subgraph problem on hypergraphs.
	\item In our setting, each variable $x_i$ ranges over a separate set $X_i$, also known as a \emph{multichromatic} setting. We leave it open to transfer our results to the \emph{monochromatic} setting (see e.g.~\cite{KarthikM19}).
	\item While we consider running times expressed in the input size (as usual in database contexts), it would also be natural to consider parameterization in the number $n$ of objects in the relational structure, see~\cite{Williams14}.
\end{itemize}
Besides these extensions, we ask whether one can close the remaining gap in our classification: Do formulas $\psi$ with $\Hand(\psi)=2$ and $\Hdeg(\psi)\ge 3$ admit an approximation scheme? As a specific challenge, we give the following open problem:
\begin{openproblem}\label{open:exact-3-cover}
	Is there an approximation scheme for Maximum Exact-$3$-Cover (or its minimization variant)? Specifically, can we prove or rule out that for every $\varepsilon > 0$, there is some $\delta > 0$ such that we can $(1+\varepsilon)$-approximate Maximum Exact-3-Cover in time $O(m^{3-\delta})$?
\end{openproblem}

It appears likely that showing existence of an approximation scheme for Maximum Exact-3-Cover would lead to a full characterization of $\MaxSP_k$. 

Finally, while we focused on the qualitative question whether or not exhaustive search can be beaten, a follow-up question is to determine precise \emph{approximability-time tradeoffs}. 
In this vein, consider the well-studied Maximum $k$-Cover problem: A simple linear-time greedy approach is known to establish a $(1-1/e)^{-1}$-approximation~\cite{Hochbaum96}. Subsequent lower bounds show that this is conditionally best possible in polynomial-time~\cite{Feige98} and even $f(k)m^{o(k)}$-time (under \GapETH{})~\cite{Cohen-AddadGKLL19,Manurangsi20}. On the other hand, for every fixed~$k$, we show (1) existence of an approximation scheme, but (2) rule out an efficient one assuming SETH, i.e., for every $\delta > 0$, there is some $\varepsilon > 0$ such that an $(1+\varepsilon)$-approximation requires time $\Omega(m^{k-\delta})$.

\begin{openproblem}\label{open:maximum-3-coverage}
	Let $k \ge 2$. Can we determine, for every $1 \le \gamma \le (1-1/e)^{-1}$, the optimal exponent $\alpha$ of the fastest $\gamma$-approximation for Maximum $k$-Cover running in time $O(m^{\alpha \pm o(1)})$, assuming plausible fine-grained hardness assumptions?  
\end{openproblem}
Note that the extreme cases for $\gamma = 1$ and $\gamma = (1-1/e)^{-1}$ are already settled and that \cite{Manurangsi20} shows that for all immediate cases, $\alpha$ must have a linear dependence on $k$, assuming \GapETH{}.

\bibliographystyle{plain}
\bibliography{refs}

\begin{thebibliography}{10}

\bibitem{AaronsonW09}
Scott Aaronson and Avi Wigderson.
\newblock Algebrization: {A} new barrier in complexity theory.
\newblock {\em {TOCT}}, 1(1):2:1--2:54, 2009.

\bibitem{AbboudB17}
Amir Abboud and Arturs Backurs.
\newblock Towards hardness of approximation for polynomial time problems.
\newblock In {\em Proceedings of the 8th Conference on Innovations in
  Theoretical Computer Science}, volume~67 of {\em ITCS '17}, pages
  11:1--11:26. Schloss Dagstuhl - Leibniz-Zentrum f{\"{u}}r Informatik, 2017.

\bibitem{AbboudBVW15}
Amir Abboud, Arturs Backurs, and Virginia~Vassilevska Williams.
\newblock Tight hardness results for {LCS} and other sequence similarity
  measures.
\newblock In {\em Proceedings of the 56th {IEEE} Annual Symposium on
  Foundations of Computer Science}, FOCS '15, pages 59--78. {IEEE} Computer
  Society, 2015.

\bibitem{AbboudBW18}
Amir Abboud, Arturs Backurs, and Virginia~Vassilevska Williams.
\newblock If the current clique algorithms are optimal, so is {Valiant}'s
  parser.
\newblock {\em {SIAM} J. Comput.}, 47(6):2527--2555, 2018.

\bibitem{AbboudBDN18}
Amir Abboud, Karl Bringmann, Holger Dell, and Jesper Nederlof.
\newblock More consequences of falsifying {SETH} and the orthogonal vectors
  conjecture.
\newblock In {\em Proceedings of the 50th Annual {ACM} Symposium on Theory of
  Computing}, STOC '18, pages 253--266. {ACM}, 2018.

\bibitem{AbboudR18}
Amir Abboud and Aviad Rubinstein.
\newblock Fast and deterministic constant factor approximation algorithms for
  {LCS} imply new circuit lower bounds.
\newblock In {\em Proceedings of the 9th Conference on Innovations in
  Theoretical Computer Science}, volume~94 of {\em ITCS '18}, pages
  35:1--35:14. Schloss Dagstuhl - Leibniz-Zentrum f{\"{u}}r Informatik, 2018.

\bibitem{AbboudRW17}
Amir Abboud, Aviad Rubinstein, and Ryan Williams.
\newblock Distributed {PCP} theorems for hardness of approximation in {P}.
\newblock In {\em Proceedings of the 58th {IEEE} Annual Symposium on
  Foundations of Computer Science}, FOCS '17, pages 25--36. {IEEE} Computer
  Society, 2017.

\bibitem{AbboudWY15}
Amir Abboud, Ryan Williams, and Huacheng Yu.
\newblock More applications of the polynomial method to algorithm design.
\newblock In {\em Proceedings of the 26th Annual {ACM-SIAM} Symposium on
  Discrete Algorithms}, SODA '15, pages 218--230. {SIAM}, 2015.

\bibitem{AlmanCW16}
Josh Alman, Timothy Chan, and Ryan Williams.
\newblock Polynomial representations of threshold functions and algorithmic
  applications.
\newblock In {\em Proceedings of the 57th {IEEE} Annual Symposium on
  Foundations of Computer Science}, FOCS '16, pages 467--476. {IEEE} Computer
  Society, 2016.

\bibitem{AlmanCW20}
Josh Alman, Timothy Chan, and Ryan Williams.
\newblock Faster deterministic and {Las} {Vegas} algorithms for offline
  approximate nearest neighbors in high dimensions.
\newblock In {\em Proceedings of the 31st Annual {ACM-SIAM} Symposium on
  Discrete Algorithms}, SODA '20, pages 637--649. {SIAM}, 2020.

\bibitem{AlmanW15}
Josh Alman and Ryan Williams.
\newblock Probabilistic polynomials and hamming nearest neighbors.
\newblock In {\em Proceedings of the 56th {IEEE} Annual Symposium on
  Foundations of Computer Science}, FOCS '15, pages 136--150. {IEEE} Computer
  Society, 2015.

\bibitem{AlmanW20}
Josh Alman and Virginia~Vassilevska Williams.
\newblock {OV} graphs are (probably) hard instances.
\newblock In {\em Proceedings of the 11th Conference on Innovations in
  Theoretical Computer Science}, volume 151 of {\em ITCS '20}, pages
  83:1--83:18. Schloss Dagstuhl - Leibniz-Zentrum f{\"{u}}r Informatik, 2020.

\bibitem{AlmanW21}
Josh Alman and Virginia~Vassilevska Williams.
\newblock A refined laser method and faster matrix multiplication.
\newblock In {\em Proceedings of the 32nd Annual {ACM-SIAM} Symposium on
  Discrete Algorithms}, SODA '21, pages 522--539. {SIAM}, 2021.

\bibitem{AndoniI06}
Alexandr Andoni and Piotr Indyk.
\newblock Near-optimal hashing algorithms for approximate nearest neighbor in
  high dimensions.
\newblock In {\em Proceedings of the 47th {IEEE} Annual Symposium on
  Foundations of Computer Science}, FOCS '06, pages 459--468. {IEEE} Computer
  Society, 2006.

\bibitem{AndoniINR14}
Alexandr Andoni, Piotr Indyk, Huy~L. Nguyen, and Ilya~P. Razenshteyn.
\newblock Beyond locality-sensitive hashing.
\newblock In {\em Proceedings of the 25th Annual {ACM-SIAM} Symposium on
  Discrete Algorithms}, SODA '14, pages 1018--1028. {SIAM}, 2014.

\bibitem{AndoniIR18}
Alexandr Andoni, Piotr Indyk, and Ilya~P. Razenshteyn.
\newblock Approximate nearest neighbor search in high dimensions.
\newblock In {\em Proceedings of the International Congress of Mathematicians},
  ICM '18, pages 3287--3318, 2018.

\bibitem{AndoniKO10}
Alexandr Andoni, Robert Krauthgamer, and Krzysztof Onak.
\newblock Polylogarithmic approximation for edit distance and the asymmetric
  query complexity.
\newblock In {\em Proceedings of the 51st {IEEE} Annual Symposium on
  Foundations of Computer Science}, FOCS '10, pages 377--386. {IEEE} Computer
  Society, 2010.

\bibitem{AndoniN20}
Alexandr Andoni and Negev~Shekel Nosatzki.
\newblock Edit distance in near-linear time: it's a constant factor.
\newblock In {\em Proceedings of the 61st {IEEE} Annual Symposium on
  Foundations of Computer Science}, FOCS '20, pages 990--1001. {IEEE}, 2020.

\bibitem{AndoniO09}
Alexandr Andoni and Krzysztof Onak.
\newblock Approximating edit distance in near-linear time.
\newblock In {\em Proceedings of the 41st Annual {ACM} Symposium on Theory of
  Computing}, STOC '09, pages 199--204. {ACM}, 2009.

\bibitem{AndoniR15}
Alexandr Andoni and Ilya~P. Razenshteyn.
\newblock Optimal data-dependent hashing for approximate near neighbors.
\newblock In {\em Proceedings of the 47th Annual {ACM} Symposium on Theory of
  Computing}, STOC '15, pages 793--801. {ACM}, 2015.

\bibitem{BabaiGKL03}
L{\'{a}}szl{\'{o}} Babai, Anna G{\'{a}}l, Peter~G. Kimmel, and Satyanarayana~V.
  Lokam.
\newblock Communication complexity of simultaneous messages.
\newblock {\em {SIAM} J. Comput.}, 33(1):137--166, 2003.

\bibitem{BackursRSVWW18}
Arturs Backurs, Liam Roditty, Gilad Segal, Virginia~Vassilevska Williams, and
  Nicole Wein.
\newblock Towards tight approximation bounds for graph diameter and
  eccentricities.
\newblock In {\em Proceedings of the 50th Annual {ACM} Symposium on Theory of
  Computing}, STOC '18, pages 267--280. {ACM}, 2018.

\bibitem{Bar-YossefJKK04}
Ziv Bar{-}Yossef, T.~S. Jayram, Robert Krauthgamer, and Ravi Kumar.
\newblock Approximating edit distance efficiently.
\newblock In {\em Proceedings of the 45th {IEEE} Annual Symposium on
  Foundations of Computer Science}, FOCS '04, pages 550--559. {IEEE} Computer
  Society, 2004.

\bibitem{BatuEKMRRS03}
Tugkan Batu, Funda Erg{\"{u}}n, Joe Kilian, Avner Magen, Sofya Raskhodnikova,
  Ronitt Rubinfeld, and Rahul Sami.
\newblock A sublinear algorithm for weakly approximating edit distance.
\newblock In {\em Proceedings of the 35th Annual {ACM} Symposium on Theory of
  Computing}, STOC '03, pages 316--324. {ACM}, 2003.

\bibitem{BatuES06}
Tugkan Batu, Funda Erg{\"{u}}n, and S{\"{u}}leyman~Cenk Sahinalp.
\newblock Oblivious string embeddings and edit distance approximations.
\newblock In {\em Proceedings of the 17th Annual {ACM-SIAM} Symposium on
  Discrete Algorithms}, SODA '06, pages 792--801. {SIAM}, 2006.

\bibitem{BjorklundPWZ14}
Andreas Bj{\"o}rklund, Rasmus Pagh, Virginia~Vassilevska Williams, and Uri
  Zwick.
\newblock Listing triangles.
\newblock In {\em Proceedings of the 41st International Colloquium on Automata,
  Languages, and Programming}, ICALP '14, pages 223--234. Springer Berlin
  Heidelberg, 2014.

\bibitem{BorodinOR99}
Allan Borodin, Rafail Ostrovsky, and Yuval Rabani.
\newblock Subquadratic approximation algorithms for clustering problems in high
  dimensional spaces.
\newblock In {\em Proceedings of the 31st Annual {ACM} Symposium on Theory of
  Computing}, STOC '99, pages 435--444. {ACM}, 1999.

\bibitem{BrakensiekR20}
Joshua Brakensiek and Aviad Rubinstein.
\newblock Constant-factor approximation of near-linear edit distance in
  near-linear time.
\newblock In {\em Proceedings of the 52nd Annual {ACM} Symposium on Theory of
  Computing}, STOC '20, pages 685--698. {ACM}, 2020.

\bibitem{Bringmann14}
Karl Bringmann.
\newblock Why walking the dog takes time: {Fr\'echet} distance has no strongly
  subquadratic algorithms unless {SETH} fails.
\newblock In {\em Proceedings of the 55th {IEEE} Annual Symposium on
  Foundations of Computer Science}, FOCS '15, pages 661--670. {IEEE} Computer
  Society, 2014.

\bibitem{BringmannCFK21}
Karl Bringmann, Alejandro Cassis, Nick Fischer, and Marvin K{\"{u}}nnemann.
\newblock Fine-grained completeness for optimization in {P}.
\newblock In {\em {APPROX-RANDOM}}, volume 207 of {\em LIPIcs}, pages
  9:1--9:22. Schloss Dagstuhl - Leibniz-Zentrum f{\"{u}}r Informatik, 2021.

\bibitem{BringmannFK19}
Karl Bringmann, Nick Fischer, and Marvin K{\"{u}}nnemann.
\newblock A fine-grained analogue of {Schaefer}'s theorem in {P}: {Dichotomy}
  of $\exists^k\forall$-quantified first-order graph properties.
\newblock In {\em Proceedings of the 34th Computational Complexity Conference},
  volume 137 of {\em CCC '19}, pages 31:1--31:27. Schloss Dagstuhl --
  Leibniz-Zentrum f{\"{u}}r Informatik, 2019.

\bibitem{BringmannKW19}
Karl Bringmann, Marvin K{\"{u}}nnemann, and Karol Wegrzycki.
\newblock Approximating {APSP} without scaling: {Equivalence} of approximate
  min-plus and exact min-max.
\newblock In {\em Proceedings of the 51st Annual {ACM} Symposium on Theory of
  Computing}, STOC '19, pages 943--954. {ACM}, 2019.

\bibitem{ChakrabortyDGKS18}
Diptarka Chakraborty, Debarati Das, Elazar Goldenberg, Michal Kouck{\'{y}}, and
  Michael~E. Saks.
\newblock Approximating edit distance within constant factor in truly
  sub-quadratic time.
\newblock In {\em Proceedings of the 59th {IEEE} Annual Symposium on
  Foundations of Computer Science}, FOCS '18, pages 979--990. {IEEE} Computer
  Society, 2018.

\bibitem{ChalermsookCKLM17}
Parinya Chalermsook, Marek Cygan, Guy Kortsarz, Bundit Laekhanukit, Pasin
  Manurangsi, Danupon Nanongkai, and Luca Trevisan.
\newblock From gap-{ETH} to {FPT}-inapproximability: {Clique}, dominating set,
  and more.
\newblock In {\em Proceedings of the 58th {IEEE} Annual Symposium on
  Foundations of Computer Science}, FOCS '17, pages 743--754. {IEEE} Computer
  Society, 2017.

\bibitem{ChanW16}
Timothy~M. Chan and Ryan Williams.
\newblock Deterministic {APSP}, orthogonal vectors, and more: Quickly
  derandomizing {Razborov}-{Smolensky}.
\newblock In {\em Proceedings of the 27th Annual {ACM-SIAM} Symposium on
  Discrete Algorithms}, SODA '16, pages 1246--1255. {SIAM}, 2016.

\bibitem{Chen18}
Lijie Chen.
\newblock On the hardness of approximate and exact (bichromatic) maximum inner
  product.
\newblock In {\em Proceedings of the 33th Computational Complexity Conference},
  volume 102 of {\em CCC '18}, pages 14:1--14:45. Schloss Dagstuhl -
  Leibniz-Zentrum f{\"{u}}r Informatik, 2018.

\bibitem{ChenGLRR19}
Lijie Chen, Shafi Goldwasser, Kaifeng Lyu, Guy~N. Rothblum, and Aviad
  Rubinstein.
\newblock Fine-grained complexity meets {IP} = {PSPACE}.
\newblock In {\em Proceedings of the 30th Annual {ACM-SIAM} Symposium on
  Discrete Algorithms}, SODA '19, pages 1--20. {SIAM}, 2019.

\bibitem{ChenW19}
Lijie Chen and Ryan Williams.
\newblock An equivalence class for orthogonal vectors.
\newblock In {\em Proceedings of the 30th Annual {ACM-SIAM} Symposium on
  Discrete Algorithms}, SODA '19, pages 21--40. {SIAM}, 2019.

\bibitem{ChenS15}
Ruiwen Chen and Rahul Santhanam.
\newblock Improved algorithms for sparse {MAX-SAT} and {MAX-k-CSP}.
\newblock In {\em Proceedings of the 18th International Conference on Theory
  and Applications of Satisfiability Testing}, volume 9340 of {\em SAT '15},
  pages 33--45. Springer Berlin Heidelberg, 2015.

\bibitem{Cohen-AddadGKLL19}
Vincent Cohen{-}Addad, Anupam Gupta, Amit Kumar, Euiwoong Lee, and Jason Li.
\newblock Tight {FPT} approximations for $k$-median and $k$-means.
\newblock In {\em Proceedings of the 46th International Colloquium on Automata,
  Languages, and Programming}, volume 132 of {\em ICALP '19}, pages
  42:1--42:14. Schloss Dagstuhl - Leibniz-Zentrum f{\"{u}}r Informatik, 2019.

\bibitem{Creignou95}
Nadia Creignou.
\newblock A dichotomy theorem for maximum generalized satisfiability problems.
\newblock {\em J. Comput. Syst. Sci.}, 51(3):511--522, 1995.

\bibitem{DietzfelbingerSW18}
Martin Dietzfelbinger, Philipp Schlag, and Stefan Walzer.
\newblock A subquadratic algorithm for {3XOR}.
\newblock In {\em Proceedings of the 43rd International Symposium on
  Mathematical Foundations of Computer Science}, volume 117 of {\em {MFCS}
  '18}, pages 59:1--59:15. Schloss Dagstuhl - Leibniz-Zentrum f{\"{u}}r
  Informatik, 2018.

\bibitem{DuanP14}
Ran Duan and Seth Pettie.
\newblock Linear-time approximation for maximum weight matching.
\newblock {\em J. {ACM}}, 61(1):1:1--1:23, 2014.

\bibitem{Feige98}
Uriel Feige.
\newblock A threshold of $\ln n$ for approximating set cover.
\newblock {\em J. {ACM}}, 45(4):634--652, 1998.

\bibitem{FeldmannKLM20}
Andreas~Emil Feldmann, C.~S. Karthik, Euiwoong Lee, and Pasin Manurangsi.
\newblock A survey on approximation in parameterized complexity: Hardness and
  algorithms.
\newblock {\em Electronic Colloquium on Computational Complexity {(ECCC)}},
  27:86, 2020.

\bibitem{Gabow85}
Harold~N. Gabow.
\newblock Scaling algorithms for network problems.
\newblock {\em J. Comput. Syst. Sci.}, 31(2):148--168, 1985.

\bibitem{Gao19}
Jiawei Gao.
\newblock On the fine-grained complexity of least weight subsequence in
  multitrees and bounded treewidth {DAGs}.
\newblock In {\em Proceedings of the 14th International Symposium on
  Parameterized and Exact Computation}, volume 148 of {\em {IPEC} '19}, pages
  16:1--16:17. Schloss Dagstuhl - Leibniz-Zentrum f{\"{u}}r Informatik, 2019.

\bibitem{GaoI19}
Jiawei Gao and Russell Impagliazzo.
\newblock The fine-grained complexity of strengthenings of first-order logic.
\newblock {\em Electronic Colloquium on Computational Complexity {(ECCC)}},
  26:9, 2019.

\bibitem{GaoIKW18}
Jiawei Gao, Russell Impagliazzo, Antonina Kolokolova, and Ryan Williams.
\newblock Completeness for first-order properties on sparse structures with
  algorithmic applications.
\newblock {\em {ACM} Trans. Algorithms}, 15(2):23:1--23:35, 2019.

\bibitem{GionisIM99}
Aristides Gionis, Piotr Indyk, and Rajeev Motwani.
\newblock Similarity search in high dimensions via hashing.
\newblock In {\em Proceedings of the 25th International Conference on Very
  Large Data Bases}, VLDB '99, pages 518--529. Morgan Kaufmann, 1999.

\bibitem{GoelIV01}
Ashish Goel, Piotr Indyk, and Kasturi~R. Varadarajan.
\newblock Reductions among high dimensional proximity problems.
\newblock In {\em Proceedings of the 12th Annual {ACM-SIAM} Symposium on
  Discrete Algorithms}, SODA '01, pages 769--778. {SIAM}, 2001.

\bibitem{GoldenbergRS20}
Elazar Goldenberg, Aviad Rubinstein, and Barna Saha.
\newblock Does preprocessing help in fast sequence comparisons?
\newblock In {\em Proceedings of the 52nd Annual {ACM} Symposium on Theory of
  Computing}, STOC '20, pages 657--670. {ACM}, 2020.

\bibitem{Har-PeledIM12}
Sariel Har{-}Peled, Piotr Indyk, and Rajeev Motwani.
\newblock Approximate nearest neighbor: {T}owards removing the curse of
  dimensionality.
\newblock {\em Theory of Computing}, 8(1):321--350, 2012.

\bibitem{Hastad01}
Johan H{\aa}stad.
\newblock Some optimal inapproximability results.
\newblock {\em J. {ACM}}, 48(4):798--859, 2001.

\bibitem{Hochbaum96}
Dorit~S. Hochbaum.
\newblock {\em Approximating Covering and Packing Problems: {S}et Cover, Vertex
  Cover, Independent Set, and Related Problems}, pages 94--143.
\newblock PWS Publishing Co., 1996.

\bibitem{Indyk00}
Piotr Indyk.
\newblock Dimensionality reduction techniques for proximity problems.
\newblock In {\em Proceedings of the 11th Annual {ACM-SIAM} Symposium on
  Discrete Algorithms}, SODA '00, pages 371--378. {SIAM}, 2000.

\bibitem{Indyk03}
Piotr Indyk.
\newblock Better algorithms for high-dimensional proximity problems via
  asymmetric embeddings.
\newblock In {\em Proceedings of the 14th Annual {ACM-SIAM} Symposium on
  Discrete Algorithms}, SODA '03, pages 539--545. {SIAM}, 2003.

\bibitem{IndykM98}
Piotr Indyk and Rajeev Motwani.
\newblock Approximate nearest neighbors: {Towards} removing the curse of
  dimensionality.
\newblock In {\em Proceedings of the 30th Annual {ACM} Symposium on Theory of
  Computing}, STOC ’98, pages 604--613. {ACM}, 1998.

\bibitem{JafargholiV16}
Zahra Jafargholi and Emanuele Viola.
\newblock {3SUM}, {3XOR}, triangles.
\newblock {\em Algorithmica}, 74(1):326--343, 2016.

\bibitem{KarthikLM19}
C.~S. Karthik, Bundit Laekhanukit, and Pasin Manurangsi.
\newblock On the parameterized complexity of approximating dominating set.
\newblock {\em J. {ACM}}, 66(5):33:1--33:38, 2019.

\bibitem{KarthikM19}
{Karthik {C. S.}} and Pasin Manurangsi.
\newblock On closest pair in euclidean metric: {Monochromatic} is as hard as
  bichromatic.
\newblock In {\em Proceedings of the 10th Conference on Innovations in
  Theoretical Computer Science}, volume 124 of {\em ITCS '19}, pages
  17:1--17:16. Schloss Dagstuhl - Leibniz-Zentrum f{\"{u}}r Informatik, 2019.

\bibitem{KhannaSTW00}
Sanjeev Khanna, Madhu Sudan, Luca Trevisan, and David~P. Williamson.
\newblock The approximability of constraint satisfaction problems.
\newblock {\em {SIAM} J. Comput.}, 30(6):1863--1920, 2000.

\bibitem{Knuth00}
Donald~E. Knuth.
\newblock Dancing links.
\newblock {\em CoRR}, abs/cs/0011047, 2000.

\bibitem{KouckyS20}
Michal Kouck{\'{y}} and Michael~E. Saks.
\newblock Constant factor approximations to edit distance on far input pairs in
  nearly linear time.
\newblock In {\em Proceedings of the 52nd Annual {ACM} Symposium on Theory of
  Computing}, STOC '20, pages 699--712. {ACM}, 2020.

\bibitem{KunnemannM20}
Marvin K{\"{u}}nnemann and D{\'{a}}niel Marx.
\newblock Finding small satisfying assignments faster than brute force: {A}
  fine-grained perspective into boolean constraint satisfaction.
\newblock In {\em Proceedings of the 35th Computational Complexity Conference},
  CCC '20, 2020.
\newblock To appear.

\bibitem{LandauMS98}
Gad~M. Landau, Eugene~W. Myers, and Jeanette~P. Schmidt.
\newblock Incremental string comparison.
\newblock {\em {SIAM} J. Comput.}, 27(2):557--582, 1998.

\bibitem{Lin21}
Bingkai Lin.
\newblock Constant approximating k-clique is w[1]-hard.
\newblock In {\em Proceedings of the 53rd Annual {ACM} Symposium on Theory of
  Computing}, STOC '21, pages 1749--1756. {ACM}, 2021.

\bibitem{LincolnWW18}
Andrea Lincoln, Virginia~Vassilevska Williams, and Ryan Williams.
\newblock Tight hardness for shortest cycles and paths in sparse graphs.
\newblock In {\em Proceedings of the 29th Annual {ACM-SIAM} Symposium on
  Discrete Algorithms}, SODA '18, pages 1236--1252. {SIAM}, 2018.

\bibitem{Manurangsi20}
Pasin Manurangsi.
\newblock Tight running time lower bounds for strong inapproximability of
  maximum $k$-coverage, unique set cover and related problems (via $t$-wise
  agreement testing theorem).
\newblock In {\em Proceedings of the 31st Annual {ACM-SIAM} Symposium on
  Discrete Algorithms}, SODA '20, pages 62--81. {SIAM}, 2020.

\bibitem{NesetrilP85}
Jaroslav Ne{\v{s}}et{\v{r}}il and Svatopluk Poljak.
\newblock On the complexity of the subgraph problem.
\newblock {\em Commentationes Mathematicae Universitatis Carolinae},
  26(2):415--419, 1985.

\bibitem{PapadimitriouY91}
Christos~H. Papadimitriou and Mihalis Yannakakis.
\newblock Optimization, approximation, and complexity classes.
\newblock {\em J. Comput. Syst. Sci.}, 43(3):425--440, 1991.

\bibitem{Patrascu10}
Mihai P{\v a}tra{\c s}cu.
\newblock Towards polynomial lower bounds for dynamic problems.
\newblock In {\em Proceedings of the 42nd Annual {ACM} Symposium on Theory of
  Computing}, STOC ’10, pages 603--610. {ACM}, 2010.

\bibitem{RodittyVW13}
Liam Roditty and Virginia~Vassilevska Williams.
\newblock Fast approximation algorithms for the diameter and radius of sparse
  graphs.
\newblock In {\em Proceedings of the 45th Annual {ACM} Symposium on Theory of
  Computing}, STOC '13, pages 515--524. {ACM}, 2013.

\bibitem{Rubinstein18}
Aviad Rubinstein.
\newblock Hardness of approximate nearest neighbor search.
\newblock In {\em Proceedings of the 50th Annual {ACM} Symposium on Theory of
  Computing}, STOC '18, pages 1260--1268. {ACM}, 2018.

\bibitem{Stephens-DavidowitzV19}
Noah Stephens{-}Davidowitz and Vinod Vaikuntanathan.
\newblock {SETH}-hardness of coding problems.
\newblock In {\em Proceedings of the 60th {IEEE} Annual Symposium on
  Foundations of Computer Science}, FOCS '19, pages 287--301. {IEEE} Computer
  Society, 2019.

\bibitem{Valiant15}
Gregory Valiant.
\newblock Finding correlations in subquadratic time, with applications to
  learning parities and the closest pair problem.
\newblock {\em J. {ACM}}, 62(2), 2015.

\bibitem{VassilevskaW18}
Virginia Vassilevska~Williams.
\newblock On some fine-grained questions in algorithms and complexity.
\newblock In {\em Proceedings of the International Congress of Mathematicians},
  ICM '18, pages 3447--3487, 2018.

\bibitem{Williams05}
Ryan Williams.
\newblock A new algorithm for optimal 2-constraint satisfaction and its
  implications.
\newblock {\em Theor. Comput. Sci.}, 348(2-3):357--365, 2005.

\bibitem{Williams07}
Ryan Williams.
\newblock {\em Algorithms and resource requirements for fundamental problems}.
\newblock ProQuest LLC, Ann Arbor, MI, 2007.
\newblock Thesis (Ph.D.)--Carnegie Mellon University.

\bibitem{Williams14}
Ryan Williams.
\newblock Faster decision of first-order graph properties.
\newblock In {\em Proceedings of the Joint Meeting of the 23rd {EACSL} Annual
  Conference on Computer Science Logic (CSL) and the 29th Annual {ACM/IEEE}
  Symposium on Logic in Computer Science (LICS)}, CSL-LICS ’14. {ACM}, 2014.

\bibitem{Williams18}
Ryan Williams.
\newblock On the difference between closest, furthest, and orthogonal pairs:
  {Nearly}-linear vs barely-subquadratic complexity.
\newblock In {\em Proceedings of the 29th Annual {ACM-SIAM} Symposium on
  Discrete Algorithms}, SODA '18, pages 1207--1215. {SIAM}, 2018.

\bibitem{Yao79}
Andrew~Chi{-}Chih Yao.
\newblock Some complexity questions related to distributive computing.
\newblock In {\em Proceedings of the 9th Annual {ACM} Symposium on Theory of
  Computing}, STOC '79, pages 209--213. {ACM}, 1979.

\bibitem{Zwick02}
Uri Zwick.
\newblock All pairs shortest paths using bridging sets and rectangular matrix
  multiplication.
\newblock {\em J. {ACM}}, 49(3):289--317, 2002.

\end{thebibliography}

\newpage
\appendix
\section*{Appendix}
In the following appendix sections we mention some interesting examples of our work, give background on the hardness assumptions and provide detailed proofs of our theorems. We start with a quick outline.
\newcommand\LmBox[2][lipicsYellow]{%
	\tikz[baseline=(label.base)]{
        \fill[#1] (0, -.35ex)
            -- ++(1.1em, 0)
            -- ++(0, 2.25ex)
            -- ++(-1.1em, 0)
            -- cycle;
		\node[anchor=base, inner sep=0] (label) at (.55em, 0) {\color{black}\bfseries\sffamily #2};
    }}%
\begin{itemize}
\labelsep.4em
\smallskip
\item[\LmBox{\ref{sec:examples}}] \textbf{\textsf{Examples:}} Summarizes the problems listed in~\autoref{tab:examples} and demonstrates how to apply our classification theorems to pinpoint their position in the four-regime landscape.
\item[\LmBox{\ref{sec:assumptions}}] \textbf{\textsf{Hardness Assumptions:}} Formally introduces our fine-grained hardness assumptions, and includes arguments for their plausibility.
\item[\LmBox{\ref{sec:related-work}}] \textbf{\textsf{Comparison to Related Work:}} Extends the discussion on~\autoref{sec:introduction:sec:related-work} about related work.
\item[\LmBox{\ref{sec:vector}}] \textbf{\textsf{Vector Optimization Problems:}} Introduces a subfamily of problems which we call \emph{vector (optimization) problems.} As we will prove, these problems capture the core hardness of the more general classes $\MaxSP_k$ and $\MinSP_k$ and are used as a critical ingredient in the following technical sections.
\item[\LmBox{\ref{sec:exact}}] \textbf{\textsf{Exact Optimization:}} Provides the algorithms for problems that can be optimized exactly and the conditional lower bounds for problems which cannot (for both maximization and minimization).
\item[\LmBox{\ref{sec:approx-max}}] \textbf{\textsf{Approximation Algorithms for Maximization:}} Provides the algorithms to approximate \emph{maximization} problems (that is, approximation schemes, constant-factor and polynomial-factor approximations).
\item[\LmBox{\ref{sec:approx-min}}] \textbf{\textsf{Approximation Algorithms for Minimization:}} Provides the algorithms to approximate \emph{minimization} problems (that is, approximation schemes and constant-factor approximations).
\item[\LmBox{\ref{sec:additive}}] \textbf{\textsf{Additive Approximation Algorithms:}} Provides the additive approximation algorithms (for both maximization and minimization problems).
\item[\LmBox{\ref{sec:hardness-approx}}] \textbf{\textsf{Hardness of Approximation:}} Proves our results on hardness of approximation (for both maximization and minimization problems).  
\item[\LmBox{\ref{sec:pcp}}] \textbf{\textsf{Distributed PCP and Extensions:}} Shows how to transfer the Distributed PCP theorem established by~\cite{AbboudRW17,Rubinstein18,Chen18,KarthikLM19} to our setting.
\end{itemize}

\section{Examples} \label{sec:examples}
In the introduction, we gave as examples for our classification results the following three families of natural $\MaxSP_k$ and $\MinSP_k$ graph problems. 
For each problem, an instance consists of $k$ sets of $n$ vectors $X_1, \dots, X_k \subseteq \{0, 1\}^d$.
\begin{itemize}
\item Maximum and Minimum \kXOR{k}: $\max_{x_1, \dots, x_k} / \min_{x_1, \dots, x_k} \# \{ y : x_1[y] \oplus \dots \oplus x_k[y] = 0 \}$,
\item Maximum and Minimum $k$-Agreement: $\max_{x_1, \dots, x_k} / \min_{x_1, \dots, x_k} \# \{ y : x_1[y] = \cdots = x_k[y] \}$,
\item Maximum and Minimum $k$-Inner Product: $\max_{x_1, \dots, x_k} / \min_{x_1, \dots, x_k} \innerprod{x_1}{\dots, x_k}$,\\where we write $\innerprod{x_1}{\dots, x_k} = \#\{ y : x_1[y] \land \dots \land x_k[y]\}$.
\end{itemize}
Variants of some of our conclusions for each problem were already proved in the literature (these provided the starting point for the paper); we will point these out below. The main point of this section is to demonstrate how easily we can determine the hardness regime of interesting graph formulas by a concise computation of our hardness parameters. In particular, we demonstrate in detail how to compute the hardness parameters $\Hdeg$ and $\Hand$ for all problems in the list, and how to apply our results. For each example, we condition on the weakest possible hypothesis.

We start with Maximum and Minimum \kXOR{k}. For the case of $k = 2$, these problems correspond to Furthest and Closest Pair in Hamming distance, respectively. For these problems, Rubinstein~\cite{Rubinstein18} ruled out efficient approximation schemes in the moderate-dimensional case for Minimum \kXOR{2} under \SETH{}, and Chen and Williams~\cite{ChenW19} gave an analogous result for Maximum \kXOR{2}.

\begin{example}[Maximum and Minimum \kXOR{k}]
Both Maximum \kXOR{k} and Minimum \kXOR{k} admit randomized approximation schemes, but no efficient ones unless \SETH{} fails.
\end{example}
\begin{proof}
Let $\psi = \opt_{x_1, \dots, x_k} \counting_y \phi(x_1[y], \dots, x_k[y])$, where $\phi(z_1, \dots, z_k) = (z_1 \oplus \dots \oplus z_k = 0)$; then~$\psi$ is the $\OptSP_k$ formula corresponding to Maximum \kXOR{k}, respectively Minimum \kXOR{k}. We first compute $\Hdeg(\phi)$ and $\Hand(\phi)$. The polynomial extension of $\phi$ is
\begin{equation*}
    \phi(z_1, \dots, z_k) = 1 - \sum_{\emptyset \neq S \subseteq [k]} (-2)^{|S| - 1} z^S,
\end{equation*}
and thus $\deg(\phi) = k$. By \autoref{prop:optimization-hardness-fourier} it follows immediately that $\Hdeg(\phi) = k$. Moreover, it is easy to see that $\Hand(\phi) = 1$ as any restriction of $\phi$ to $h$ inputs is still a parity function and therefore has $2^{h-1}$ satisfying assignments.

We can now apply our classification to derive the hardness of Maximum \kXOR{k} and Minimum \kXOR{k}: From Theorems~\ref{thm:as-maximization-classification} and~\ref{thm:as-minimization-classification} we obtain randomized approximation schemes for both problems, and~\autoref{thm:no-efficient-as} rules out efficient approximation schemes under \SETH{}.
\end{proof}

For Maximum and Minimum $k$-Agreement, the relation between the hardness parameters $\Hdeg$ and $\Hand$ is more interesting. Note that Maximum $k$-Agreement can be viewed as a very restricted setting of the MaxCover problem for which Karthik, Laekhanukit and Manurangsi~\cite{KarthikLM19} obtain strong inapproximability results. It turns out that our setting exhibits a more diverse landscape of approximability, which depends on $k$ in a non-monotone way. We refer to Section~\ref{sec:psp} for a detailed comparison between Maximum $k$-Agreement and the MaxCover problem studied in~\cite{KarthikLM19}.

\begin{example}[Maximum $k$-Agreement] \hspace{0cm}
\begin{itemize}
\itemdesc{If \boldmath$k = 2$:} Maximum $2$-Agreement has a randomized approximation scheme, but not an efficient one, unless the Sparse \MAXThreeSAT{} hypothesis fails.
\itemdesc{If \boldmath$k = 3$:} Maximum $3$-Agreement can be solved exactly.
\itemdesc{If \boldmath$k \geq 4$:} Maximum $k$-Agreement admits a $(k+1)$-approximation in time $\Order(m^{k-\delta})$ for some $\delta > 0$, but cannot be approximated better than $(1 - \binom k3{}^{-1})^{-1}$ in time $\Order(m^{k-\delta})$, for any $\delta > 0$, unless the Sparse \MAXThreeSAT{} hypothesis fails.
\end{itemize}
\end{example}
\begin{proof}
Let $\psi = \max_{x_1, \dots, x_k} \counting_y \phi(x_1[y], \dots, x_k[y])$, where $\phi$ is the Boolean function indicating whether all its $k$ inputs are equal; then $\psi$ is the $\MaxSP_k$ formula corresponding to Maximum $k$\=/Agreement. We first determine $\Hdeg(\phi)$ and $\Hand(\phi)$. The polynomial extension of $\phi$ is
\begin{equation*}
    \phi(z_1, \dots, z_k) = \prod_{i \in [k]} z_i + \prod_{i \in [k]} (1 - z_i) = z^{[k]} + \sum_{S \subseteq [k]} (-1)^{|S|} z^S.
\end{equation*}
In particular, observe that the leading monomial $z^{[k]}$ cancels if $k$ is odd. By~\autoref{prop:optimization-hardness-fourier}, it is easy to check that $\Hdeg(\phi) = k$ if $k$ is even, and $\Hdeg(\phi) = k - 1$ if $k$ is odd.

Next, we claim that $\Hand(\phi) = k - 1$. On the one hand, it is easy to see that $\Hand(\phi) < k$ as otherwise $\phi$ had only one satisfying assignment. On the other hand, we can the set any variable $x_i$ to, say $0$; the remaining Boolean function on $k - 1$ inputs has exactly one satisfying assignment. We can now read off the hardness of Maximum $k$-Agreement from our classification.
\begin{itemize}
\itemdesc{If \boldmath$k = 2$:} We have $\Hdeg(\phi) = 2$ and $\Hand(\phi) = 1$. By \autoref{thm:as-maximization-classification} there exists a randomized approximation scheme for Maximum $2$-Agreement, and by~\autoref{thm:no-efficient-as} there exists no efficient one unless \SETH{} fails.
\itemdesc{If \boldmath$k = 3$:} We have $\Hdeg(\phi) = 2 < k$ and therefore, Maximum $3$-Agreement can be optimized exactly in time $\Order(m^{k-\delta})$ by~\autoref{thm:exact-opt-classification}.
\itemdesc{If \boldmath$k \geq 4$:} We have $3 \leq \Hdeg(\phi) \leq k$ and $\Hand(\phi) = k - 1 < k$. It follows that we can compute a constant-factor approximation in time $\Order(m^{k-\delta})$ by~\autoref{thm:constant-approx-max-classification}. The precise approximation guarantee can be obtained from \autoref{lem:constant-approx-max-opposite}. In contrast, it is not possible to compute an approximation better than $(1 - \binom k3{}^{-1})^{-1}$ unless the Sparse \MAXThreeSAT{} hypothesis fails, as proven in~\autoref{thm:as-maximization-classification}. \qedhere
\end{itemize}
\end{proof}

With the same proof we can also classify the computational complexity of the Minimum $k$\=/Agreement problem. The only difference is that for $k \geq 4$, we cannot expect any approximation at all under the \hUniformHyperClique{(k-1)} assumption, by~\autoref{thm:constant-approx-min-classification} and \cite[Theorem 6]{BringmannFK19}.

\begin{example}[Minimum $k$-Agreement] \hspace{0cm}
\begin{itemize}
\itemdesc{If \boldmath$k = 2$:} Minimum $2$-Agreement has a randomized approximation scheme, but not an efficient one, unless the Sparse \MAXThreeSAT{} hypothesis fails.
\itemdesc{If \boldmath$k = 3$:} Minimum $3$-Agreement can be solved exactly.
\itemdesc{If \boldmath$k \geq 4$:} Minimum $k$-Agreement cannot be approximated at all (i.e. we cannot distinguish if $\OPT$ is 0 or at least 1) in time $\Order(m^{k-\delta})$ for any $\delta > 0$, unless the \hUniformHyperClique{(k-1)} hypothesis fails.
\end{itemize}
\end{example}

Finally, we consider the hardest problems Maximum and Minimum $k$-Inner Product. From our classification, we can immediately read off that Maximum $k$-Inner Product can be approximated within an arbitrary polynomial factor and nothing better, while Minimum $k$-Inner Product cannot be approximated at all. The underlying inapproximability of Maximum $k$-Inner Product is based on the works of Abboud, Rubinsten and Williams~\cite{AbboudRW17} and Chen~\cite{Chen18} for $k=2$ and Karthik, Laekhanukit and Manurangsi~\cite{KarthikLM19} for general $k$.

\begin{example}[Maximum $k$-Inner Product]
For any $\varepsilon > 0$, Maximum $k$-Inner Product has an $m^\varepsilon$\=/approximation in time $\Order(m^{k-\delta})$ for some $\delta > 0$. However, for every $\delta > 0$, there exists some $\varepsilon > 0$ such that computing an $m^\varepsilon$-approximation requires time $\Order(m^{k-\delta})$ unless \SETH{} fails.
\end{example}
\begin{proof}
Maximum and Minimum $k$-Inner Product are the prototypical problems $\psi$ of full hardness $\Hand(\psi) = \Hdeg(\psi) = k$. We only need that $\Hand(\psi) = k$ which is clear since $\phi(z_1, \dots, z_k) = z_1 \dots z_k$ has only one satisfying assignment. (By \autoref{prop:hardness-relation}, it follows that also $\Hdeg(\psi) = k$). The claim is immediate by~\autoref{thm:tight-polynomial-maximization}.
\end{proof}

\begin{example}[Minimum $k$-Inner Product]
Minimum $k$-Inner Product has no approximation at all in time $\Order(m^{k-\delta})$ for any $\delta > 0$, unless the \kOV{k} hypothesis fails.
\end{example}
\begin{proof}
As for Maximum $k$-Inner Product, it holds that $\Hand(\psi) = \Hdeg(\psi) = k$. The claim follows immediately by~\autoref{cor:constant-approx-min-characterization}.
\end{proof}

\section{Hardness Assumptions} \label{sec:assumptions}
As usual in fine-grained complexity, our lower bounds are conditioned on certain hardness assumptions. In this section we introduce the relevant hypotheses and briefly argue for their plausibility. In our work we mostly depend on two popular hypotheses on polynomial-time problems, both of which are backed by hardness assumptions on fundamental, well-studied exponential-time problems. Whenever possible, we design our reductions to start from the (weaker) polynomial-time assumptions. To state all of our results, we rely on a relatively new assumption on sparse maximum satisfiability.

\subsection{Orthogonal Vectors and SETH}
One of the most prominent hypotheses in the area is the ($k$-)Orthogonal Vectors hypothesis\footnote{In the case $k = 2$, the prefix is usually omitted and we simply speak of the Orthogonal Vectors hypothesis.}. In the $k$-Orthogonal Vectors problem (\kOV{k}), we are given $k$ sets of bit-vectors $X_1, \ldots, X_k \subseteq \{0, 1\}^d$, and the task is to check whether there exist $x_1 \in X_1, \ldots, x_k \in X_k$ with (generalized) inner product equal to zero: \raisebox{0pt}[0pt][0pt]{$\innerprod{x_1}{\ldots, x_k} := \sum_{y \in [d]} x_1[y] \cdots x_k[y] = 0$}. The naive algorithm solves \kOV{k} in time $\Order(n^k d)$ and the corresponding hypothesis states that this algorithm is essentially optimal:

\begin{hypothesis}[\kOV{k}]
For every~$\delta > 0$, \kOV{k} on $n$ vectors with dimension $d$ cannot be solved in time $\Order(n^{k-\delta} \poly(d))$.
\end{hypothesis}

The fastest known algorithm solves \kOV{k} in time $n^{k-\Omega(1/\log(d/\log n))}$~\cite{AbboudWY15,ChanW16}, which is a subpolynomial improvement over the baseline algorithm and thus not enough to refute the \kOV{k} hypothesis. It is implied by many other fine-grained hypotheses, among others the weighted \kClique{k} hypothesis~\cite{AbboudBDN18}, the assumption that not all $(k + 1)$-quantifier first-order properties can be solved in time $\Order(m^{k-\delta})$~\cite{GaoIKW18}, and -- most notably -- the Strong Exponential Time Hypothesis (\SETH{})~\cite{Williams05}. We remark that the variant we stated here is usually referred to as the \emph{moderate-dimensional} \kOV{k} hypothesis, while \SETH{} in fact implies an even stronger version postulating hardness just above logarithmic dimension:

\begin{hypothesis}[Low-Dimensional \kOV{k}]
For every $\delta > 0$, there exists some $c > 0$ such that \kOV{k} on $n$ vectors with dimension $d = c \log n$ cannot be solved in time $\Order(n^{k-\delta})$.
\end{hypothesis}

\subsection{\HyperClique{} and Maximum Satisfiability}
Another central hypothesis that gained a lot of popularity in recent years is the \HyperClique{} hypothesis~\cite{AbboudBDN18, LincolnWW18}. Let $h \geq 2$. Given an $h$-uniform $k$-partite hypergraph $H = (V_1, \dots, V_k, E)$, the \hUniformkHyperClique{h}{k} problem is to detect a $k$-hyperclique (that is, a set of $k$ vertices $v_1 \in V_1, \ldots, v_k \in V_k$ such that any $h$ of them form a hyperedge). The brute-force algorithm solves the problem in time $\Order(n^k)$, and for $h = 2$ there exists a polynomial improvement exploiting fast matrix multiplication~\cite{NesetrilP85}. However, for $h \geq 3$, it is postulated that a similar improvement is not possible:

\begin{hypothesis}[\hUniformHyperClique{h}]
For every $\delta > 0$ and $k$, the \hUniformkHyperClique{h}{k} problem cannot be solved in time $\Order(n^{k-\delta})$.
\end{hypothesis}

We refer to~\cite[Section 7]{LincolnWW18} for a detailed discussion on the plausibility of the \hUniformHyperClique{h} hypothesis, which especially addresses the question of why there is no hope to achieve improvements for $h \geq 3$ in a similar spirit as the $h = 2$ algorithm. It is known that a refutation of this hypothesis is required for further progress on model-checking of first-order properties~\cite{BringmannFK19}, and that it is \emph{equivalent} to the existence of faster-than-brute-force algorithms for a class of Boolean \CSP{}s (including \ThreeSAT{}) parameterized by solution size~\cite{KunnemannM20}.  
A further argument for this hypothesis is its connection to a well-studied exponential-time problem: The \MAXThreeSAT{} problem is the optimization variant of the \ThreeSAT{} problem, that is, given a $3$-CNF formula find an assignment maximizing the total number of satisfied clauses. The associated hypothesis postulates that we cannot significantly improve upon brute-force search:

\begin{hypothesis}[\MAXThreeSAT{}]
For every $\delta > 0$, the \MAXThreeSAT{} problem on $n$ variables cannot be solved in time $\Order(2^{n(1-\delta)})$.
\end{hypothesis}

The fastest currently known algorithm solves \MAXThreeSAT{} in time $2^{n(1 - 1/\tilde\Order(\log^2 n))}$~\cite{AlmanCW16}. The \MAXThreeSAT{} hypothesis implies both the \hUniformHyperClique{h} assumption~\cite{Williams07} (for all $h \geq 3$) as well as the \kOV{k} assumption (for all $k \geq 2$)~\cite{AbboudBDN18}. Actually, it is plausible to extend and strengthen this hypothesis further to a sparse problem variant as defined in~\cite{AlmanW20}:

\begin{hypothesis}[Sparse~\MAXThreeSAT{}]
For every $\delta > 0$, there exists $c > 0$ such that the \MAXThreeSAT{} problem on $n$ variables and $c n$ clauses cannot be solved in time $\Order(2^{n(1-\delta)})$.
\end{hypothesis}

Indeed, previous work addressed the problem of finding more efficient algorithms for Sparse \MAXThreeSAT{}~\cite{ChenS15,AlmanCW16}, and the best-known algorithm runs in time $2^{n(1-1/\tilde\Order(\log^2 c))}$ \cite{AlmanCW16}, which is compatible with the hardness assumption. Interestingly, the Sparse \MAXThreeSAT{} hypothesis unifies all previous hypotheses, including the \MAXThreeSAT{} hypothesis (obviously), but also \SETH{}.\footnote{Indeed,~\cite[Theorem 1.5]{AbboudBDN18} proves that refuting \SETH{} implies that there exists some $\delta > 0$, such that for all constants $c$ and $d$, the satisfiability of depth-$d$ threshold circuits with $cn$ wires can be determined in time $\Order(2^{n(1-\delta)})$. As a consequence, we can also solve the Sparse \MAXThreeSAT{} problem in time $\Order(2^{n(1-\delta)} \poly(n))$, by brute-forcing over all targets $0 \leq t \leq cn$ and by encoding the resulting decision problem as a depth-$2$ threshold circuit.} As such, it provides the currently \emph{easiest} barrier for all algorithmic problems that we classify as hard.
\section{Comparison to Related Work} \label{sec:related-work}
In this section, we give a more detailed comparison of similarities and differences between some related works in inapproximability and fine-grained complexity theory.

In particular, the close conceptual connection between $\MaxSP_k$ and $\MaxSNP$ is described in Section~\ref{sec:SNPmotivation}. A comparison to the parameterized inapproximability result of Karthik, Laekhanukit, and Manurangsi~\cite{KarthikLM19} is given in Section~\ref{sec:psp}.  
\subsection{Correspondence of \texorpdfstring{$\MaxSP_k$}{MaxSP} and \texorpdfstring{$\MaxSNP$}{MaxSNP}}
\label{sec:SNPmotivation}

Let us review the definition of $\MaxSNP$ and compare it to $\MaxSP_k$ (see~\cite{BringmannCFK21} for more details). The class $\SNP$ is motivated by Fagin's theorem, and is defined as  problems of the form $\exists S\, \forall\bar{y}\, \phi(\bar{y},G,S)$, where $G$ is a given relational structure, $\exists S$ ranges over a relational structure $S$ and $\forall\bar{y}\, \phi(\bar{y},G,S)$ is a $\forall^*$-quantified first-order property. Its natural optimization variant is $\MaxSNP$, defined as the set of problems expressible as $\max_{S}\#\{\bar{y} : \phi(\bar{y},G,S)\}$. 

As an analogue for $O(m^k)$-time solvable problems, $\MaxSP_k$ is a variation of  $\MaxSNP$ where the maximization over a relational structure $S$ is replaced by maxmization over $k$ variables $x_1,\dots,x_k$, and the counting is performed over a variable $y$.\footnote{An alternative would be to allow counting over $\ell$ variables $y_1, \dots, y_\ell$. This more general definition is pursued in~\cite{BringmannCFK21}, but proven to be reducible to the case of $\ell=1$, as Maximum $k$-Inner Product (a problem for $\ell=1$) is shown to be complete for this more general class.} Let us remark that \cite{BringmannCFK21} choose the name ``$\mathrm{SP}_k$'' for ``strict polynomial time $O(n^k)$'' in analogy to the name ``strict NP'' of $\SNP$.

\subsection{Comparison to Karthik, Laekhanukit and Manurangsi}
\label{sec:psp}

Let us detail some interesting aspects of the work of Karthik, Laekhanukit and Manurangsi~\cite{KarthikLM19}. To obtain inapproximability results for Dominating Set, Maximum $k$-Inner Product and other problems, the authors use a framework of reductions from so-called Product Space Problems (PSPs) via MaxCover as an intermediate problem. It is insightful to highlight how these notions relate to our problem settings:

\subparagraph*{Connection to Product Space Problems (PSPs)}
A Product Space Problem (PSP) is a very abstract problem formulation: Given a function $f: X_1 \times \cdots \times X_k \to \{0,1\}$, the corresponding product space problem $\mathrm{PSP}(f)$ is to determine, given an input $A_1 \subseteq X_1, \dots, A_k \subseteq X_k$, whether there exists $(a_1,\dots, a_k) \in A_1 \times \dots \times A_k$ such that $f(a_1, \dots, a_k)=1$. Karthik et al.\ use this notion to give a general and elegant reduction framework: Let $\mathcal{F}$ be a family of functions such that (1) each $f\in \mathcal{F}$ has an ``efficient'' communication protocol (for the precise communication setting, we refer the reader to~\cite{KarthikLM19}) and (2) computing $\mathrm{PSP}(f)$ for $f\in \mathcal{F}$ is hard to solve under some fine-grained assumption $\mathcal{A}$. Then the framework immediately derives approximation hardness for MaxCover based on the assumption $\mathcal{A}$. In particular, this framework is used in~\cite{KarthikLM19} to obtain -- via reduction from MaxCover -- approximation hardness for Dominating Set and Maximum $k$-Inner Product under a diverse set of assumptions (such as SETH, $\mathrm{W}[1]\ne \mathrm{FPT}$ and the $k$-SUM hypothesis).

Note that while PSPs defines a general family of problems, this notion would not be a suitable class to study in the context of our work: In a sense, it definition is too abstract/general for our purposes (one could not even enumerate all problems in this class, so characterizing all such problems appears hopeless). Our chosen classes, $\MaxSP_k$ and $\MinSP_k$, provide problem classes defined by a simple \emph{syntactic} definition -- they are general enough to contain a diverse set of problems with very different approximability, while on the other hand, it is always clear which problems are members of this class (notably, for every $k$, we could enumerate the problems in $\MaxSP_k$ and $\MinSP_k$). This opens up the possibility of a systematic classification, as performed in our work.

\subparagraph*{Is MaxCover in $\MaxSP_k$?}
Karthik et al.~prove their inapproximability results~\cite[Theorem 5.2]{KarthikLM19} already for a restricted setting of MaxCover that is closely related to our setting of $\MaxSP_k$. Specifically, one can view this variant (also called \emph{MaxCover with projection property}) as the following problem over vectors: Given sets $X_1, \dots, X_k \subseteq \Sigma^d$, maximize the number of coordinates~$y$ such that $x_1[y]=x_2[y]=\cdots=x_k[y]$ over all $(x_1,\dots,x_k)\in X_1\times \cdots \times X_k$. For $\Sigma = \{0,1\}$, this is precisely a $\MaxSP_k$ \emph{graph} formula studied in this paper as Maximum $k$-Agreement. For any constant~$|\Sigma|$, we can express the corresponding Maximum $k$-Agreement problem still as a $\MaxSP_k$ formula, using additional binary relations to represent the larger alphabet (this is no longer a \emph{graph} formula, see also the discussion in Section~\ref{sec:conclusion}). Importantly, Karthik et al.'s main inapproximability result for MaxCover~\cite[Theorem 5.2]{KarthikLM19} assuming hardness of $\mathrm{PSP}(\mathcal{F})$ produces instances for which $|\Sigma|$ depends both on $2^k$ and the efficiency of the communication protocol for $\mathcal{F}$. Thus, this result cannot be directly used for the $\MaxSP_k$ graph formulas studied in the present paper, but provides an important result for larger alphabet sizes. After a subsequent reduction to Maximum $k$-Inner Product~\cite[Theorem B.3]{KarthikLM19} (see Section~\ref{sec:pcp}), their result provides an important part of the basis for our hardness of approximation results given in Section~\ref{sec:hardness-approx}.

\section{Vector Optimization Problems} \label{sec:vector}
As mentioned before, the predicates $E(x_i, x_j)$ have only a minor impact on the complexity of both the model-checking problem of~$\psi$~\cite[Theorem~7]{BringmannFK19} and the optimization problems $\Opt(\psi)$. In particular, as we will show, it suffices to classify the subtype of problems featuring no predicates $E(x_i, x_j)$. For convenience in the following parts of this paper, we will now make this notion precise.

\begin{definition}[Vector Optimization Problems]
Let $\phi$ be a Boolean function on $k$ inputs, and let $\psi = \opt_{x_1, \dots, x_k} \counting_y \phi(E(x_1, y), \ldots, E(x_k, y))$. We refer to $\Max(\psi)$ and $\Min(\psi)$ as \emph{vector (optimization) problems} $\VMax(\phi)$ and $\VMin(\phi)$. Occasionally, we write $\VOpt(\psi)$ to refer to both problems simultaneously.
\end{definition}

In the context of vector problems, we often view the objects $x_i$ as vectors and view $Y$ as the set of coordinates. In addition, we define the \emph{dimension} $d := |Y|$ and, following standard notation, write $x_i[y] := E(x_i, y)$. Note that a vector problem instance is sparsely represented by one list of all nonzero entries in all vectors.

Recall that $\psi_0$ (as defined in the beginning of~\autoref{sec:technical}) is a Boolean function obtained from~$\psi$ by substituting each atom $E(x_i, x_j)$ by $\false$. Similar to~\cite[Theorem~7]{BringmannFK19}, we obtain the following equivalence of general optimization problems $\Opt(\psi)$ and vector optimization problems $\VOpt(\psi_0)$. Note that this equivalence holds in the strongest approximation-preserving way, and a similar statement holds for additive approximations. The proof of~\autoref{thm:equiv-vector-problems} is postponed to the end of this section.

\begin{restatable}[Equivalence of Vector Optimization Problems]{theorem}{thmequivvectorproblems} \label{thm:equiv-vector-problems}
Let $\psi$ be an $\OptSP_k$ graph formula, and let $c = c(m) \geq 1$.
\begin{itemize}
\item If a $c$-approximation for $\Opt(\psi)$ can be computed in time $T(m)$, then a $c$-approximation for $\VOpt(\psi_0)$ can be computed in time $O(T(m))$.
\item If a $c$-approximation for $\VOpt(\psi_0)$ can be computed in time $O(m^{k-\varepsilon})$ for some $\varepsilon > 0$, then a $c$-approximation for $\Opt(\psi)$ can be computed in time $O(m^{k - \varepsilon'})$ for some $\varepsilon' > 0$.
\end{itemize}
\end{restatable}

In the upcoming algorithms, to solve $\VOpt(\phi)$ it is often necessary to ``brute-force'' over a number of variables $x_i$ to reduce $k$ to some smaller $k'$. However, in general this cannot be considered a self-reduction, as the resulting $k'$-problem is not necessarily the same vector optimization problem $\VOpt(\phi)$. As an example, consider the \kXOR{3} problem, $\VOpt(z_1 \oplus z_2 \oplus z_3)$. After brute-forcing over all variables $x_1$, we are left with a combination of $\VOpt(0 \oplus z_2 \oplus z_3)$ and $\VOpt(1 \oplus z_2 \oplus z_3)$, where in each coordinate one of these two constraints applies. As we regularly have to deal with these combined problems, we now define them formally.

\begin{definition}[Hybrid Optimization Problems]
Let $\Phi$ be a set of Boolean functions on $k$ inputs. In a \emph{hybrid optimization problem $\VMax(\Phi)$}, we are given an $(k+1)$-partite graph $X_1 \union \dots \union X_k \union Y$, where each coordinate $y \in Y$ is associated to a function $\phi_y \in \Phi$, and the task is to compute
\begin{equation*}
  \OPT = \max_{x_1 \in X_1, \dots, x_k \in X_k} \counting_{y \in Y} \phi_y(x_1[y], \dots, x_k[y]).
\end{equation*}
We similarly define $\VMin(\Phi)$.
\end{definition}

Our hardness parameters $\Hand$ and $\Hdeg$ interplay nicely with hybrid problems: Let $\phi$ be a $k$\=/variable Boolean function of hardness, say $\Hand(\phi) \leq h$, and, for some $h \leq k' \leq k$, let $\Phi$ be the set of all $k'$-variable Boolean functions $\phi'$ of hardness $\Hand(\phi') \leq h$. Then $\VOpt(\phi)$ reduces to $n^{k-k'}$ instances of $\VOpt(\Phi)$. Indeed, by the definition of $\Hand$, we can always find a set of $k - k'$ indices which we can brute-force over, such that any remaining formula is of hardness at most $h$. The same statement holds for $\Hdeg$ as well.

We finally give the missing proof of~\autoref{thm:equiv-vector-problems}. The idea is a rather straightforward generalization of~\cite[Theorem 7]{BringmannFK19}. We essentially exploit that there are not too many tuples $(x_1, \ldots, x_k)$ for which at least one of the edges $(x_i, x_j)$ is present. As an important ingredient, we recycle a lemma already proven in~\cite{BringmannFK19}.

\begin{lemma}[{\cite[Lemma 22]{BringmannFK19}}] \label{lem:respect-inner-prod}
Given a nonempty set of index pairs $\emptyset \neq I \subseteq \binom{[k]}2$, we say that $(x_1, \ldots, x_k)$ \emph{respects} $I$ if, for any $i \neq j$, we have $E(x_i, x_j)$ if and only if $\{i,j\} \in I$. We can compute $\innerprod{x_1}{\ldots, x_k}$ for all vectors $(x_1, \ldots, x_k)$ respecting $I$ in time $O(m^{k - 1/2})$.
\end{lemma}

\begin{proof}[Proof of~\autoref{thm:equiv-vector-problems}]
The first part is clear: The vector optimization problem $\VOpt(\psi_0)$ is the special case of $\Opt(\psi)$, where we leave out all edges $E(x_i, x_j)$.

So let us focus on the second part and start with $\Max(\psi)$. Let
\begin{equation*}
  \psi = \max_{x_1, \dots, x_k} \counting_y \phi(x_1, \ldots, x_k, y).
\end{equation*}
For any set $I \subseteq \binom{[k]}{2}$ of index pairs, let $\phi_I$ denote the formula obtained from $\phi$ after substituting every atom $E(x_i, x_j)$ by $\true$ if $\{i, j\} \in I$, and by $\false$ otherwise. In particular, $\phi_\emptyset = \psi_0$. Furthermore, we define
\begin{equation*}
  \psi_I = \max_{x_1, \dots, x_k} \counting_y \left[ \left( \bigwedge_{\{i,j\} \subseteq [k]} E(x_i, x_j) \iff \{i,j\} \in I\right) \land \phi_I(x_1[y],\dots,x_k[y]) \right].
\end{equation*}
It suffices to separately solve the given $\Opt(\psi)$ instance for all $I$, treated as an instance of $\Opt(\psi_I)$, respectively. Then $\OPT = \max_I \OPT_I$, where $\OPT_I$ is the optimal value of the $\Opt(\psi_I)$ instance. We distinguish two cases:
\begin{description}
\smallskip
\item[Case 1: \boldmath$I \neq \emptyset$.] We claim that an exact solution for $\Opt(\psi_I)$ can be computed in time $\Order(m^{k-1/2})$. Clearly, we only need to focus on tuples $(x_1, \ldots, x_k)$ that respect~$I$, as any other tuple is of value~$0$. There are at most $\Order(m^{k-1})$ such tuples. Indeed, since we assumed that~$I$ is non-empty, there exists a pair of indices $\{i, j\} \in I$ such that there exists an edge $E(x_i, x_j)$ whenever $(x_1, \ldots, x_k)$ respects~$I$. However, there are at most $m$ pairs $(x_i, x_j)$ satisfying that condition, and therefore at most $n^{k-2} m \leq \Order(m^{k-1})$ relevant tuples $(x_1, \ldots, x_k)$.

We can represent the value of $(x_1, \ldots, x_k)$ in terms of the polynomial representation of $\phi_I$; in the language of~\autoref{sec:exact}:
\begin{equation*}
  \Val(x_1, \ldots, x_k) = \sum_{\substack{i_1 < \ldots < i_\ell\\0 \leq \ell \leq k}} \fourier\phi_I(\{i_1, \ldots, i_\ell\}) \mult \innerprod{x_{i_1}}{\ldots, x_{i_\ell}}.
\end{equation*}
Assuming that all inner products $\innerprod{x_{i_1}}{\ldots, x_{i_\ell}}$ are precomputed, we can determine the optimum value in time $\Order(m^{k-1})$ by iterating over all tuples respecting $I$ and applying the previous equation in each step. For $\ell < k$, we can easily compute all inner products $\innerprod{x_{i_1}}{\ldots, x_{i_\ell}}$ in time $\Order(m^\ell) \leq \Order(m^{k-1})$. For $\ell = k$, we apply~\autoref{lem:respect-inner-prod} to compute all inner products $\innerprod{x_1}{\ldots, x_k}$ in time $\Order(m^{k-1/2})$.
\medskip
\item[Case 2: \boldmath$I = \emptyset$.] Note that $\Max(\psi_\emptyset)$ is not directly equivalent to $\VMax(\psi_0)$, since an arbitrary solution $(x_1, \ldots, x_k)$ for the latter does not necessarily meet the condition that none of the edges $(x_i, x_j)$ be present. The following reduction enforces this constraint.

Computing a $c$-approximation of $\Max(\psi)$ can be reduced to the following gap problem: Given a target $t$, distinguish whether $\OPT \geq t$, or whether $\OPT < t / c$. Binary-searching over $t$ then yields a $c$-approximation of the original problem, incurring only a logarithmic overhead. Despite its name, the definition of the gap problem also makes sense for the case $c = 1$.

Let $\gamma > 0$ be a parameter to be fixed later. We call a vertex $x_i$ \emph{heavy} if it is of degree $\geq m^\gamma$, and \emph{light} otherwise. The first step is to eliminate all heavy vertices; there can exist at most $\Order(m / m^\gamma) = \Order(m^{1-\gamma})$ many such vertices~$x_i$. Fixing $x_i$, we can solve the remaining problem in time $\Order(m^{k-1})$ using the baseline algorithm. If a solution of value $\geq t / c$ is detected, we accept. It thus takes time $\Order(m^{k-\gamma})$ to safely remove all heavy vertices.

Next, partition each set $X_i$ into several groups $X_{i,1}, \ldots, X_{i,g}$ such that the total degree of all vertices in a group is $\Theta(m^\gamma)$, except for possibly the last nonempty groups. This is implemented by greedily inserting vertices into $X_{i.j}$ until its total degree exceeds $m^\gamma$. As each vector inserted in that way is light, we can overshoot by at most $m^\gamma$. It follows that $g \leq \Order(m / m^\gamma) = \Order(m^{1-\gamma})$.

We assume that $\VMax(\phi_\emptyset)$ can be $c$-approximated in time $\Order(m^{k-\delta})$ for some $\delta > 0$. Then we continue as follows:
\begin{enumerate}
\item For all combinations $(j_1, \ldots, j_k) \in [g]^k$, compute a $c$-approximation of the vector optimization problem $\VMax(\psi_0)$ with input $X_{1, j_1}, \ldots, X_{k, j_k}$. If the returned value is $\geq t / c$, then we call $(j_1, \ldots, j_k)$ a \emph{successful} combination.
\item If there are more than $m n^{k-2}$ successful combinations, we accept.
\item Otherwise, for any successful combination $(j_1, \ldots, j_k)$, solve $\Max(\psi)$ on $X_{1, j_1}, \ldots, X_{k, j_k}$ exactly using the baseline algorithm, and accept if and only if one of these calls returns a value $\geq t / c$.
\end{enumerate}
We begin with the correctness of the algorithm. First of all, whenever a solution $(x_1, \ldots, x_k)$ in $\Max(\psi_\emptyset)$ is of value at least $t$, then the respective value in $\VMax(\phi_\emptyset)$ is also at least $t$. It is therefore safe to only treat those subinstances in step 3 which we considered successful in step 1. It remains to argue why step 2 is correct. How many tuples $(x_1, \ldots, x_k)$ can be false positives, that is, be of value $\geq t / c$ in $\VMax(\phi_\emptyset)$, but of value $< t / c$ in $\Max(\psi_\emptyset)$? At most $m n^{k-2}$, since at least one edge $(x_i, x_j)$ must exist for any false positive. Thus, if we witness $> m n^{k-2}$ solutions of value $\geq t / c$ in $\VMax(\phi_\emptyset)$, among these there exists at least one true positive.

Finally, let us bound the running time of the above algorithm. Recall that removing heavy vertices accounts for $\Order(m^{k-\gamma})$ time. In step 1, the $\VMax(\phi_\emptyset)$ algorithm is applied $g^k = \Order(m^{k-\gamma k})$ times on instances of size $\Order(m^\gamma)$, which takes time $\Order(m^{k-\gamma k + \gamma (k-\delta))}) = \Order(m^{k-\gamma \delta})$. Step 3 becomes relevant only if there are at most $m n^{k-2} = \Order(m^{k-1})$ successful combinations. For any such combination, the baseline algorithm takes time $\Order((m^\gamma)^k) = \Order(m^{\gamma k})$. The total running time is $\Order(m^{k-\gamma} + m^{k-\gamma \delta} + m^{k-1 + \gamma k})$. By picking $\gamma$ in the range $0 < \gamma < \frac1k$, the claim follows.
\end{description}

This finishes the proof for the maximization variant. For minimizations, the proof is essentially analogous. We define
\begin{equation*}
  \psi_I = \min_{x_1, \dots, x_k} \counting_y \left[ \left( \bigvee_{\{i,j\} \subseteq [k]} E(x_i, x_j) \centernot\iff \{i,j\} \in I\right) \lor \phi_I(x_1[y],\dots,x_k[y]) \right]
\end{equation*}
instead, because with our former choice, any tuple not respecting $I$ is inconveniently of value~$0$. From then on, the necessary changes are straightforward and left to thorough readers.
\end{proof}

By carefully inspecting the proof of~\autoref{thm:equiv-vector-problems}, one can also see that the same approach works to show that $\Opt(\psi)$ and $\VOpt(\psi_0)$ are equivalent under additive approximations.
\section{Exact Optimization} \label{sec:exact}
This section is devoted to proving~\autoref{thm:exact-opt-classification}. In fact, we prove a stronger result, formulated in terms of vector problems:

\begin{theorem} \label{thm:exact-opt-vp}
Let $\phi$ be a Boolean function on $k \geq 2$ inputs.
\begin{itemize}
\item If $\Hdeg(\phi) \leq 1 < k$, then $\VOpt(\phi)$ can be solved in time $\Order(m^{k-1})$.
\item If $\Hdeg(\phi) \leq 2 < k$, then $\VOpt(\phi)$ can be solved in time $\Order(m^{k-\delta})$, for some $\delta > 0$, using fast matrix multiplication.
\item If $3 \leq \Hdeg(\phi) \leq k$, then $\VOpt(\phi)$ cannot be solved in time $\Order(m^{k-\delta})$, for any $\delta > 0$, unless the \hUniformHyperClique{h} hypothesis fails.
\item If $\Hdeg(\phi) = k$, then $\VOpt(\phi)$ cannot be solved in time $\Order(m^{k-\delta})$, for any $\delta > 0$, unless the \kOV{k} hypothesis fails.
\end{itemize}
\end{theorem}

By the equivalence of vector problems $\VOpt(\phi)$ and graph problems $\Opt(\psi)$ (\autoref{thm:equiv-vector-problems}) and since the \MAXThreeSAT{} assumption implies the \hUniformHyperClique{h} assumption, \autoref{thm:exact-opt-classification} follows as a corollary of \autoref{thm:exact-opt-vp}.

For the sake of convenience, let us quickly recall the definition of $\Hdeg$:

\defoptimizationhardness*

We begin with some propositions about the hardness parameters. Let $\phi(z_1, \ldots, z_k)$ be a Boolean function (as before, we also write $\phi$ to denote its multilinear polynomial extension). For $S \subseteq [k]$, we denote by $\fourier\phi(S)$ the coefficient of the monomial \raisebox{0pt}[0pt][0pt]{$z^S = \prod_{i \in S} z_i$}. These values $\fourier\phi(S)$ are sometimes referred to as the \emph{Fourier coefficients} of $\phi$. There is a simple formula for calculating $\fourier\phi(S)$ known as the inclusion-exclusion principle:

\begin{proposition}[Inclusion-Exclusion] \label{prop:inclusion-exclusion}
$\fourier\phi(S) = \sum_{T \subseteq S} (-1)^{|S| - |T|} \phi(T)$.
\end{proposition}

The following proposition relates our notion of degree hardness with certain coefficients $\fourier\phi(S)$:

\begin{proposition} \label{prop:optimization-hardness-fourier}
Let $\phi$ be a Boolean function $k$ inputs. Then $\Hdeg(\phi)$ is the maximum number $0 \leq h \leq k$ such that, for any set \raisebox{0pt}[0pt][0pt]{$S \in \binom{[k]}{h}$}, there exists some superset $T \supseteq S$ with \raisebox{0pt}[0pt][0pt]{$\fourier\phi(T) \neq 0$}.
\end{proposition}
\begin{proof}
Let $\Hopt'(\phi)$ denote the parameter defined in the statement, that is, $\Hopt'(\phi)$ is the largest number $0 \leq h \leq k$ such that, for any set \raisebox{0pt}[0pt][0pt]{$S \in \binom{[k]}{h}$}, there exists some $T \supseteq S$ with \raisebox{0pt}[0pt][0pt]{$\fourier\phi(T) \neq 0$}. We prove that $\Hdeg(\phi) = \Hopt'(\phi)$.

On the one hand, suppose that $\Hdeg(\phi) \geq h$. By definition it holds that for any set of variables \raisebox{0pt}[0pt][0pt]{$S \in \binom{[k]}{h}$}, there exists some $S$-restriction $\phi'$ of $\phi$ of degree $h$. In particular, \raisebox{0pt}[0pt][0pt]{$\fourier\phi'(S) \neq 0$} and therefore, for some $T \supseteq S$, \raisebox{0pt}[0pt][0pt]{$\fourier\phi(T) \neq 0$}. It follows that also $\Hopt'(\phi) \geq h$.

On the other hand, suppose that $\Hdeg(\phi) < h$. Then there exists some index set \raisebox{0pt}[0pt][0pt]{$S \in \binom{[k]}{h}$} such that any $S$-restriction $\phi'$ of $\phi$ is of degree less than $h$, or in other words, \raisebox{0pt}[0pt][0pt]{$\fourier\phi'(S) = 0$}. We claim that for any superset $T \supseteq S$, $\fourier\phi(T) = 0$. For the sake of contradiction, suppose otherwise and pick an inclusion-wise minimal set $T \supseteq S$ with $\fourier\phi(T) \neq 0$. Then we construct an $S$-restriction $\phi'$ of $\phi$ with $\fourier\phi'(S) \neq 0$ as follows: Fix every variable in $T \setminus S$ to $1$ and every variable in $[k] \setminus T$ to $0$. Then indeed $\fourier\phi'(S) = \fourier\phi(T) \neq 0$ as the only monomial in $\phi$ which is mapped to~$S$ under the restriction is~$T$. This is a contradiction and therefore $\Hopt'(\phi) < h$.
\end{proof}

\begin{proposition} \label{prop:optimization-hardness-negation}
For any Boolean function~$\phi$ it holds that $\Hdeg(\phi) = \Hdeg(\neg \phi)$, where $\neg \phi$ denotes the negation of $\phi$.
\end{proposition}
\begin{proof}
Negating a formula $\phi$ corresponds to changing its polynomial extension into $1 - \phi$. Any nonzero monomial (of degree $> 0$) still has a nonzero coefficient after that transformation, and therefore the claim follows by~\autoref{prop:optimization-hardness-fourier}.
\end{proof}

As an immediate consequence of~\autoref{prop:optimization-hardness-negation}, it suffices to prove~\autoref{thm:exact-opt-vp} only for $\VMax(\phi)$ (or interchangeably, $\VMin(\phi)$). Indeed, say that we have proven~\autoref{thm:exact-opt-vp} for the maximization variant, and let $\phi$ be such that $\VMax(\phi)$ admits an algorithm in time $\Order(m^{k-\delta})$. Since $\Hdeg(\phi) = \Hdeg(\neg \phi)$, there also exists an algorithm for $\VMax(\neg \phi)$ in time $\Order(m^{k-\delta})$, and solving $\VMax(\neg \phi)$ exactly is equivalent to solving $\VMin(\phi)$ under the translation $\VMin(\phi) = d - \VMax(\neg \phi)$. For that reason, we will implicitly deal with the maximization variants in this section's proofs.

Finally, it is easy to relate $\Hand(\phi)$ with $\Hdeg(\phi)$:

\begin{proposition} \label{prop:hardness-relation}
For any Boolean function $\phi$, it holds that $\Hand(\phi) \leq \Hdeg(\phi)$.
\end{proposition}
\begin{proof}
If a Boolean function $\phi' : \{ 0, 1 \}^h \to \{0, 1\}$ has a unique satisfying assignment $\alpha \in \{0, 1\}^h$, then its polynomial extension is $\phi'(z_1, \ldots, z_h) = \prod_{i=1}^h \tau_i(z_i)$, where $\tau_i(z_i) = z_i$ if $\alpha_i = 1$ and $\tau_i(z_i) = 1 - z_i$ otherwise. In particular, $\phi'$ is of degree $h$. The claim is now immediate by the definitions of $\Hand(\phi)$ and $\Hdeg(\phi)$.
\end{proof}

\subsection{Hardness Results} \label{sec:exact:sec:hardness}
As the main tool for the hardness proof, we demonstrate how to construct \emph{coordinate gadgets} for all problems classified as hard in~\autoref{thm:exact-opt-vp}. Our gadget definition is inspired by the standard reduction from \OV{} to \MaxIP{}. 

Recall that a \emph{restriction} of a Boolean function $\phi$ is another Boolean function $\phi'$ obtained by substituting constant values for some inputs of $\phi$. In the following, we need a strictly more general concept:

\begin{definition}[Unary Transformation]
Let $\phi(z_1, \ldots, z_k)$ be a Boolean function and let $\tau_1, \ldots, \tau_k : \{0, 1\} \to \{0, 1\}$ be unary functions. We call the composition $\phi(\tau_1(z_1), \ldots, \tau_k(z_k))$ a \emph{unary transformation} of $\phi$.
\end{definition}

Whenever the unary functions $\tau_i$ are either constant or the identity, a unary transformation of~$\phi$ is simply a restriction of~$\phi$. However, unary transformations also allow to flip some of the inputs.

\begin{lemma}[Coordinate Gadget] \label{lem:coordinate-gadget}
Let $\phi(z_1, \ldots, z_k)$ be a Boolean function of degree $\deg(\phi) = k$. Then there exist unary transformations $\phi_1, \ldots, \phi_\ell$ of $\phi$ such that
\begin{equation*}
    \sum_{j=1}^\ell \phi_j(z_1, \ldots, z_k) = \beta_1 - \beta_2 \mult z^{[k]},
\end{equation*}
for some positive integers $\beta_1, \beta_2$. The collection $\{\phi_1, \ldots, \phi_\ell\}$ is called a \emph{coordinate gadget} of $\phi$.
\end{lemma}
\begin{proof}
We describe an inductive procedure maintaining a polynomial~$g$ which is invariantly a sum of unary transformations of $\phi$. Initially, we set $g := \phi$ and the goal is to eventually transform $g$ into $\beta_1 - \beta_2 \mult z^{[k]}$. We can assume that $\fourier g([k]) < 0$; otherwise start with $g := \phi(1 - z_1, z_2, \ldots, z_k)$. We proceed by eliminating unnecessary monomials one at a time:

If $g$ does not have any monomial of degree $1 \leq d \leq k - 1$, then we stop and output~$g$. Otherwise, pick one such monomial $\fourier g(S) \mult z^S$ of maximum degree. Moreover, let $T \subseteq [k]$ be an inclusion-wise minimal set such that $S \subseteq T$ and $\fourier\phi(T) \neq 0$ (note that $T$ exists since $\deg(\phi)=k$). Then, we construct the unary transformation $\phi'(z_1, \ldots, z_k) := \phi(\tau_1(z_1), \ldots, \tau_k(z_k))$ with
\begin{equation*}
    \tau_i(z) :=
    \begin{cases}
        z & \text{if $i \in S$,} \\
        1 & \text{if $i \in T \setminus S$,} \\
        0 & \text{otherwise.}
    \end{cases}
\end{equation*}
We have that $\fourier\phi'(S) = \fourier\phi(T) \neq 0$. By additionally flipping an arbitrary input $z_i$, $i \in S$, we construct~$\phi''$ with $\fourier\phi''(S) = -\fourier\phi(T)$. Moreover, observe that $z^S$ is the only nonzero monomial of degree $\geq |S|$ in~$\phi''$.

Next, we cancel the monomial $z^S$ as follows: Assume without loss of generality that $\fourier\phi(T) > 0$ (and thus $\fourier\phi''(S) < 0$) and also assume $\fourier g(S) > 0$; the other cases are symmetric. We set
\begin{equation*}
    g' := \fourier\phi(T) \mult g + \fourier g(S) \mult \phi''.
\end{equation*}
Then $\fourier g'(S) = \fourier\phi(T) \mult \fourier g(S) + \fourier g(S) \mult (-\fourier\phi(T)) = 0$, and any other monomial of degree $\geq |S|$ remains untouched in $g'$. Since the number of nonzero monomials of maximum degree (less than $k$) has decreased, the construction continues with $g'$ by induction. Here, we implicitly used that the coefficients $\fourier\phi(T)$ and $\fourier g(S)$ are integral, which can for example be seen from the inclusion-exclusion formula (\autoref{prop:inclusion-exclusion}).

It remains to argue that $\beta_1$ and $\beta_2$ are positive integers. Recall that we started with $\fourier g([k]) < 0$ and any function $\phi''$ that was added in the process satisfies $\fourier\phi''([k]) = 0$. Hence, $\fourier g([k]) < 0$ and thus $\beta_2 > 0$. Finally, since $g$ is a sum of (polynomial extensions of) Boolean functions, it is clear that $0 \leq g(1, \ldots, 1) = \beta_1 - \beta_2$. It follows that $\beta_1 \geq \beta_2 > 0$.
\end{proof}

With the coordinate gadget in mind, the \kOV{k} hardness result is now almost immediate:

\begin{lemma}[No Exact Optimization for $\Hdeg(\phi) = k$] \label{lem:exact-hardness-ov}
Let $\phi$ be a Boolean function on $k$ inputs of degree hardness $\Hdeg(\phi) = k$. Then $\VOpt(\phi)$ cannot be solved in time $\Order(m^{k-\delta})$ any $\delta > 0$, unless the \kOV{k} hypothesis fails. 
\end{lemma}
\begin{proof}
We assume an $\Order(m^{k-\delta})$-time algorithm for $\VMax(\phi)$ and our goal is to find an improved algorithm for a given \kOV{k} instance $X_1, \ldots, X_k \subseteq \{0, 1\}^d$. Observe that $\Hdeg(\phi) = k$ is equivalent to $\deg(\phi) = k$, and therefore there exists a coordinate gadget $\{\phi_1, \ldots, \phi_\ell\}$ for $\phi$ by the previous lemma. Let $\tau_{j, 1}, \ldots, \tau_{j, k}$ denote the unary functions used to transform $\phi$ into $\phi_j$. Then we rewrite every vector $x_i \in X_i$ as a vector $x_i'$ of dimension $\ell d$ by mapping each entry $z$ to $(\tau_{1, i}(z), \ldots, \tau_{\ell, i}(z))$.

Consider the $\VMax(\phi)$ instance $X_1', \ldots, X_k'$. By~\autoref{lem:coordinate-gadget}, the optimal solution is
\begin{equation*}
    \OPT = \max_{x_1', \ldots, x_k'} \sum_y \phi(x_1'[y], \ldots, x_k'[y]) = \max_{x_1, \ldots, x_k} \beta_1 d - \beta_2 \innerprod{x_1}{\ldots, x_k}.
\end{equation*}
In particular, we have that $\OPT = \beta_1 d$ if we started with a YES instance of \kOV{k}, and a value strictly smaller otherwise.

Finally, we analyze the the running time. Constructing the sets $X_1', \ldots, X_k'$ takes time $\Order(nd)$, and the sparsity of the constructed $\VMax(\phi)$ instance is also $m = \Order(nd)$. The total running time is $\Order(nd + m^{k-\delta}) = \Order(n^{k-\delta} \poly(d))$, which contradicts the \kOV{k} hypothesis.
\end{proof}

\begin{lemma}[No Exact Optimization for $3 \leq \Hdeg(\phi)$] \label{lem:exact-hardness-hyperclique}
Let $\phi$ be a Boolean function on $k$ inputs of degree hardness $\Hdeg(\phi) = h$, $3 \leq h \leq k$. Then $\VOpt(\phi)$ cannot be solved in time $\Order(m^{k-\delta})$ for any $\delta > 0$, unless the \hUniformHyperClique{h} hypothesis fails.
\end{lemma}
\begin{proof}
For contradiction, let us assume that $\VMax(\phi)$ can be solved in time $\Order(m^{k-\delta})$ for some $\delta > 0$. Let $k' = k'(k, h, \delta)$ be a multiple of $k$ to be specified later. Our reduction starts from an instance~$H$ of \hUniformkHyperClique{h}{k'}. That is, $H = (V, E)$ is a $k'$-partite $h$-uniform hypergraph and we write $V = V_{1,1} \union \cdots \union V_{1,k'/k} \union \cdots \union V_{k,1} \union \cdots \union V_{k,k'/k}$ for the vertex partition into~$k'$ parts of size $n$ each. We regroup the vertices into $k$~groups $V_i := V_{i,1} \union \cdots \union V_{i,k'/k}$ and set $X_i := V_{i,1} \times \cdots \times V_{i,k'/k}$.

We begin to describe the construction of the $\VMax(\phi)$ instance: Let $I = \{ i_1 \leq \ldots \leq i_h \}$ be a set of indices and let $\bar E(I)$ denote the set of non-edges $\bar e = (v_{i_1}, \ldots, v_{i_h}) \in V_{i_1} \times \cdots \times V_{i_h}$ in $H$. In the following, we fix such a non-edge $\bar e \in \bar E(I)$. We say that some tuple $x_i \in X_i$ \emph{avoids} $\bar e$ if not all vertices in $\bar e \intersect V_i$ are part of $x_i$.

Since $\Hdeg(\phi) = h$, we know that for all index sets $I$, there exists some restriction to a subformula $\phi_I(z_{i_1}, \ldots, z_{i_h})$ of degree $\deg(\phi_I) = h$. For that subformula, we can apply~\autoref{lem:coordinate-gadget} to obtain a coordinate gadget \raisebox{0pt}[0pt][0pt]{$\{\phi_{I,1}, \ldots, \phi_{I, \ell(I)}\}$} so that
\begin{equation*}
  \sum_{j=1}^{\ell(I)} \phi_{I,j}(z_{i_1}, \ldots, z_{i_h}) = \beta_1(I) - \beta_2(I) \mult z^I.  
\end{equation*}
Recall that each $\phi_{I, j}$ is a unary transformation of $\phi'$, which in turn is a restriction of $\phi$. It follows that $\phi_{I, j}$ is also a unary transformation of $\phi$. Let $\tau_{I,j,1}, \ldots, \tau_{I,j,k}$ denote the unary functions used to transform $\phi$ into $\phi_{I,j}$. Finally, let $Y = \bigcup_I \bar E(I) \times [\ell(I)]$. We describe how to assign the entries $x_i[y]$, for $x_i \in X_i$ and $y = (\bar e, j) \in \bar E(I) \times [\ell(I)]$:
\begin{equation*}
  x_i[y] =
  \begin{cases}
    \tau_{I,j,i}(0) = \tau_{I,j,i}(1) & \text{if $i \not\in I$ (in which case $\tau_{I,j,i}$ is constant),} \\
    \tau_{I,j,i}(0) & \text{if $i \in I$ and $x_i$ avoids $\bar e$,} \\
    \tau_{I,j,i}(1) & \text{if $i \in I$ and $x_i$ does not avoid $\bar e$.}
  \end{cases}
\end{equation*}
We claim that there exists a $k'$-hyperclique in $H$ if and only if the $\VMax(\phi)$ instance over the vectors $X_1 \union \cdots \union X_k$ and dimensions $Y$ has optimal value $\beta^* := \sum_I \beta_1(I) \mult |\bar E(I)|$.

First, suppose that there exists a $k'$-hyperclique $(v_{1,1}, \ldots, v_{1,k'/k}, \ldots, v_{k,1}, \ldots, v_{k,k'/k})$ and let $(x_1, \ldots, x_k)$ be their corresponding vectors, i.e. $x_i = (v_{i,1}, \ldots, v_{i,k'/k})$. Then, for any index set $I$ and any non-edge $\bar e \in \bar E(I)$, there has to be at least one vector $x_i$, $i \in I$, which avoids $\bar e$. Hence, for any $j \in [\ell(I)]$, $x_i[(\bar e, j)] = \tau_{I,j,i}(0)$. As a consequence, the contribution of $\sum_j \phi_I(x_1[(\bar e, j)], \ldots, x_k[(\bar e, j)])$ is exactly $\beta_1(I)$. Since this argument holds for any index set $I$ and any non-edge $\bar e \in \bar E(I)$, the total value of this instance is $\beta^*$.

For the other direction, suppose that $(v_{1,1}, \ldots, v_{1,k'/k}, \ldots, v_{k,1}, \ldots, v_{k,k'/k})$ does not form a $k'$-hyperclique in $H$, and let $(x_1, \ldots, x_k)$ be their corresponding vectors. Among these $k'$ vertices, there exists at least one tuple $\bar e = (v_{i_1}, \ldots, v_{i_h})$ which forms a non-edge; let $I = \{i_1, \ldots, i_h\}$. Hence, for every $i \in I$ and every $j \in [\ell(I)]$, $x_i[(\bar e, j)] = \tau_{I, j, i}(1)$. It follows that the contribution of all coordinates $(\bar e, j)$ is:
\begin{align*}
  \sum_{j=1}^{\ell(I)} \phi\big(x_1[(\bar e, j)], \ldots, x_k[(\bar e, j)]\big)
  &= \sum_{j=1}^{\ell(I)} \phi_{I, j}\big(x_{i_1}[(\bar e, j)], \ldots, x_{i_h}[(\bar e, j)]\big) \\
  &= \beta_1(I) - \beta_2(I) \mult 1 < \beta_1(I),
\end{align*}
and therefore, the total value of $(x_1, \ldots, x_k)$ is less than $\beta^*$. This completes the correctness proof.

To analyze the running time, note that $|X_i| \leq \Order(n^{k'/k})$, $|Y| \leq \Order(n^h)$ and $m \leq \sum_i |X_i| |Y| \leq \Order(n^{k'/k+h})$. It also takes time $\Order(n^{k'/k+h})$ to construct the $\VMax(\phi)$ instance. Assuming that we can solve $\VMax(\phi)$ in time $\Order(m^{k-\delta})$, we can test whether $H$ contains a $k'$-hyperclique in time $\Order(n^{(k'/k + h)(k-\delta)})$. By choosing $k' \geq k^2 h \delta^{-1}$, this becomes $\Order(n^{k'-h\delta})$, which contradicts the \hUniformHyperClique{h} hypothesis.
\end{proof}

\subsection{Algorithmic Results} \label{sec:exact:sec:alg}
For the remaining part of~\autoref{thm:exact-opt-vp}, we present two algorithms to solve low-hardness optimization problems: First, if $\Hdeg(\phi) \leq 1 < k$, then $\VOpt(\phi)$ admits a linear-time improvement. Second, if $\Hdeg(\phi) \leq 2 < k$, then we still get a polynomial improvement using fast matrix multiplication. All in all, this approach is similar to~\cite{BringmannFK19}.

In both algorithms, we brute-force over some variables $x_i$ and solve the remaining $k = 2$ or $k = 3$ problem in a more clever way. However, recall that after instantiating some variables the remaining optimization problem is a hybrid problem.

\begin{lemma}[Exact Optimization for $\Hdeg(\phi) \leq 1$] \label{lem:exact-opt-hardness-1}
Let $\phi$ be a Boolean function on $k$ inputs of degree hardness $\Hdeg(\phi) \leq 1 < k$. Then $\VOpt(\phi)$ can be solved in time $\Order(m^{k-1})$.
\end{lemma}
\begin{proof}
By~\autoref{prop:optimization-hardness-fourier}, $\Hdeg(\phi) \leq 1$ implies that there exist distinct variables $a, b \in [k]$ such that \raisebox{0pt}[0pt][0pt]{$\fourier\phi(S) = 0$} for all $S \supseteq \{a, b\}$. We explicitly enumerate all choices of the other $k-2$ variables in $[k] \setminus \{a, b\}$. It remains to solve a two-variable hybrid vector problem $\max_{x_1, x_2} \counting_y \phi_y(x_1[y], x_2[y])$ in linear time. However, by our choice of $a, b$ we know that $\deg(\phi_y) \leq 1$ for all $y$ and thus the objective becomes
\begin{equation*}
    \max_{x_1 \in X_1} \sum_{y \in Y} \fourier\phi_y(\{1\}) \mult x_1[y] + \max_{x_2 \in X_2} \sum_{y \in Y} \fourier\phi_y(\{2\}) \mult x_2[y] + \sum_{y \in Y} \fourier\phi_y(\emptyset).
\end{equation*}
These summands can be computed independently by a single pass through the edge set and so the total running time is $\Order(n^{k-2} m) = \Order(m^{k-1})$. Note that the same approach works for minimization problems.
\end{proof}

We present the second algorithm as a reduction from $\VOpt(\phi)$ to the \hUniformkHyperClique{h}{(h+1)} problem. For $h = 2$, the improvement then follows as a corollary.

We need the following subroutine for computing maximum-weight hypercliques. That is, given an $h$-uniform $k$-partite hypergraph with integral edge weights $w$, the goal is to return a $k$-hyperclique $(v_1, \ldots, v_k)$ maximizing the total edge weight $\sum_{i_1 < \ldots < i_h} w(v_{i_1}, \ldots, v_{i_h})$.

\begin{lemma}[Maximum-Weight \HyperClique{}] \label{lem:maximum-hyperclique}
Assume that the \hUniformkHyperClique{h}{k} problem can be solved in time $\Order(n^{k-\delta})$ for some $\delta > 0$. Then, given an $h$-uniform $k$-partite hypergraph with integral edge weights $w$ satisfying \raisebox{0pt}[0pt][0pt]{$\sum_{e \in E(H)} |w(e)| \leq \Order(n^h)$} (i.e., the average edge weight is constant), we can compute a maximum-weight $k$-hyperclique in time $\Order(n^{k-\delta'})$ for some $\delta' > 0$.
\end{lemma}
\begin{proof}
Let $V_1 \union \ldots \union V_k$ be the vertex partition of the given hypergraph~$H$. We call an edge $e \in E(H)$ \emph{heavy} if $|w(e)| > n^\gamma$ (for some parameter $\gamma$ to be set later), and \emph{light} otherwise. Since \raisebox{0pt}[0pt][0pt]{$\sum_{e \in E(H)} |w(e)| \leq \Order(n^h)$}, there are at most $\Order(n^h / n^\gamma) = \Order(n^{h-\gamma})$ heavy edges. For each heavy edge $e = \{v_{i_1}, \ldots, v_{i_h}\}$ (with $v_{i_j} \in V_{i_j}$), we iterate over all $(k-h)$-tuples of remaining vertices $v_i \in V_i$ for $i \not\in \{i_1, \ldots, i_h\}$ and keep track of the maximum weight of any hyperclique. In this way, we compute the maximum weight of all hypercliques involving a heavy edge in time $\Order(n^{h-\gamma} n^{k-h}) = \Order(n^{k-\gamma})$.

After this step, we remove all heavy edges and therefore the maximum edge weight in $H$ is bounded by $n^\gamma$. Now, it is feasible to enumerate all combinations of edge weights $w_I \in \{-n^\gamma, \ldots, n^\gamma\}$ for all $I = \{i_1 < \ldots < i_h\}$. For any such combination $(w_I)_I$, we remove all edges $\{v_{i_1}, \ldots, v_{i_h}\}$ from $H$ unless $w(v_{i_1}, \ldots, v_{i_h}) = w_{\{i_1, \ldots, i_h\}}$. Therefore, any hyperclique in the remaining graph is of weight $\sum_I w_I$. We use the $\Order(n^{k-\delta})$ algorithm to test whether there exists such a clique, and by repeating the process for all weight combinations $(w_I)_I$ we can keep track of the maximum weight of any hyperclique. All in all, there are $n^{\gamma \binom kh}$ many weight combinations, so this step takes time $\Order(n^{k + \gamma \binom kh - \delta})$. 

The total running time is $\Order(n^{k-\gamma} + n^{k + \gamma \binom kh - \delta})$, which for sufficiently small $\gamma = \gamma(k, h, \delta)$ is $\Order(n^{k-\delta'})$ for some $\delta' > 0$.
\end{proof}

\begin{lemma}[Exact Optimization via \HyperClique{}] \label{lem:exact-to-hyperclique}
Let $\phi$ be a Boolean function on $k$ inputs of degree hardness $h = \Hdeg(\phi) < k$. If the \hUniformkHyperClique{h}{(h+1)} problem can be solved in time $\Order(n^{h+1-\delta'})$ for some $\delta' > 0$, then $\VOpt(\phi)$ can be solved in time $\Order(m^{k-\delta})$ for some $\delta > 0$.
\end{lemma}
\begin{proof}
The high-level idea is similar to~\autoref{lem:exact-opt-hardness-1}. Using~\autoref{prop:optimization-hardness-fourier}, there exist distinct indices $i_1, \ldots, i_{h+1} \in [k]$ so that $\fourier\phi(S) = 0$ for all $S \supseteq \{i_1, \ldots, i_{h+1}\}$. After brute-forcing over all $k - (h + 1)$ variables in $[k] \setminus \{i_1, \ldots, i_{h+1}\}$, it remains to solve a hybrid $(h+1)$-variable problem $\max_{x_1, \dots, x_{h+1}} \counting_y \phi_y(x_1[y], \ldots, x_{h+1}[y])$ in time~$\Order(m^{h+1-\delta})$. By the choice of $i_1, \ldots, i_{h+1}$, we are guaranteed that $\deg(\phi_y) \leq h$ for all coordinates $y$. Hence, we can express the objective as
\begin{equation*}
    \max_{x_1, \ldots, x_{h+1}} \sum_{\substack{i_1 < \ldots < i_\ell\\\ell \leq h}} \sum_{y \in Y} \fourier\phi_y(\{i_1, \ldots, i_\ell\}) \mult x_{i_1}[y] \ldots x_{i_\ell}[y].
\end{equation*}

The key observation is that each term depends on at most $h$ of the vectors $x_1, \ldots, x_{h+1}$. Thus, we can construct a complete $(h+1)$-partite $h$-uniform hypergraph $H$ on the vertices $X_1 \union \cdots \union X_{h+1}$ and encode the objective on the edge weights. As a first step, the encoding uses hyperedges of arities $1, \ldots, h$; we will later fix the construction to include hyperedges of arity $h$ only. Start with zero weights on all edges. For all $i_1 < \ldots < i_\ell$ with $0 < \ell \leq h$, increment all edge weights $w(x_{i_1}, \ldots, x_{i_\ell})$ by
\begin{equation*}
    \sum_{y \in Y} \fourier\phi_y(\{i_1, \ldots, i_\ell\}) \mult x_{i_1}[y] \ldots x_{i_\ell}[y].
\end{equation*}
We can compute these values in time $\Order(m^h)$ by enumerating all tuples of nonzero entries in $X_{i_1}, \ldots, X_{i_\ell}$, respectively. The empty sequence ($\ell = 0$) contributes only a constant offset in the objective and can thus be ignored. Observe that the weight of any hyperclique $(x_1, \ldots, x_{h+1})$ is exactly $\sum_{\substack{i_1 < \ldots < i_\ell,\ell \leq h}} \sum_{y \in Y} \fourier\phi_y(\{i_1, \ldots, i_\ell\}) \mult x_{i_1}[y] \ldots x_{i_\ell}[y]$. Hence, we have reduced the problem to finding a hyperclique of maximum weight in $H$. Note that each nonzero entry in a vector $x_i$ contributes to at most $n^{\ell-1}$ hyperedges of arity~$\ell$ in~$H$. Therefore, the total weight of all hyperedges of arity $\ell$ is bounded by $\Order(n^{\ell-1} m) = \Order(m^\ell)$.

We continue to fix the issue that $H$ contains hyperedges of arity less than~$h$: For any edge $(x_{i_1}, \ldots, x_{i_\ell})$, we fix some distinct indices $i_{\ell+1}, \ldots, i_h \not\in \{i_1, \ldots, i_\ell\}$ and iterate over all $(h - \ell)$-tuples of vertices $(x_{\ell+1}, \ldots, x_h) \in X_{\ell+1} \times \cdots \times X_h$. We add $w(x_1, \ldots, x_\ell)$ to all weights $w(x_1, \ldots, x_h)$ and delete the hyperedge $(x_1, \ldots, x_\ell)$ afterwards. In that way, the weight of every $(h+1)$-hyperclique remains fixed, as the weight of any small-arity hyperedge is transferred to some arity-$h$ hyperedge. Moreover, after this step $H$ is $h$-uniform. The total edge weight in $H$ is bounded by $\sum_{\ell \leq h} n^{h-\ell} \Order(m^{\ell}) = \Order(m^h)$ and the number of vertices is bounded by $\Order(m)$.

We apply~\autoref{lem:maximum-hyperclique} to find a maximum-weight hyperclique in $H$ in time $\Order(m^{h+1-\delta})$ for some $\delta > 0$. Recall that this process is repeated $n^{k-(h+1)}$ times, so the overall running time is dominated by $\Order(n^{k-(h+1)} m^{h+1-\delta}) = \Order(m^{k-\delta})$ for some $\delta > 0$, as desired.
\end{proof}

Recall that the \hUniformkHyperClique{2}{3} problem -- also known as triangle detection in standard graphs -- can be solved in time $\Order(n^{\omega})$. This yields the following corollary:

\begin{corollary}[Exact Optimization for $\Hdeg(\phi) \leq 2$] \label{cor:exact-opt-hardness-2}
Let $\phi$ be a Boolean function on $k$ inputs of degree hardness $\Hdeg(\phi) \leq 2 < k$. Then $\VOpt(\phi)$ can be solved in time $\Order(m^{k-\delta})$, for some $\delta > 0$.
\end{corollary}

\subsection{The Equivalence to \HyperClique{}} \label{sec:exact-equivalence}
In the course of the previous subsections, we proved both upper and lower bounds by connecting the optimization problems $\VOpt(\phi)$ of degree hardness $\Hdeg(\phi) = h$ to clique detection in $h$-uniform hypergraphs. \autoref{thm:exact-equivalence} summarizes both directions:

\thmexactequivalence*

\begin{proof}
By the equivalence of general problems $\Opt(\psi)$ and vector optimization problems $\VOpt(\phi)$ (\autoref{thm:equiv-vector-problems}), we only need to show the statement for vector problems $\VOpt(\phi)$. The second item is proven in~\autoref{lem:exact-to-hyperclique}. The first item is proven in~\autoref{lem:exact-hardness-hyperclique}, but only for~$h \geq 3$. Nevertheless, exactly the same proof works also for $h = 2$.
\end{proof}

\autoref{thm:exact-equivalence} strengthens the evidence for the \hUniformHyperClique{h} hypothesis (at least for the stronger version postulating no $\Order(n^{h+1-\delta})$ algorithm for the \hUniformkHyperClique{h}{h+1} problem), but we consider it an interesting result also for two additional reasons.

First, in the specific case $\Hdeg(\psi) = 2$, \autoref{thm:exact-equivalence} yields an explanation on why we can, but also \emph{need} to use fast matrix multiplication to efficiently solve $\Opt(\psi)$. Indeed, any algorithm solving $\Opt(\psi)$ in time $\Order(m^{k-\delta})$, translates directly to an algorithm for finding cliques in (ordinary) graphs in time $\Order(n^{k-\delta'})$ -- an improvement, which is labeled highly unlikely without the use of fast matrix multiplication~\cite{AbboudBW18}.

Second, it is interesting to compare \autoref{thm:exact-equivalence} to a similar, yet different equivalence result~\cite{BringmannFK19}: The \hUniformHyperClique{h} hypothesis holds if and only if for all $k \geq h$, and all $\MinSP_k$ graph formulas $\psi$ of and-hardness $\Hand(\psi) = h$, $\Zero(\psi)$ requires time $\Order(m^{k-\delta})$ for all $\delta > 0$. That equivalence is about \emph{classes} (in that it relates the class of \hUniformHyperClique{h} problems with the class of all problems $\Zero(\psi)$ of and-hardness~$h$), whereas \autoref{thm:exact-equivalence} is more \emph{problem}-centered (in the sense that \emph{every specific} graph formula $\psi$ is related with certain \hUniformHyperClique{h} problems).
\section{Approximation Algorithms for Maximization}\label{sec:approx-max}
In this section, we give our algorithmic results for approximating the maximization variant of our problems. As illustrated in \autoref{fig:spectrum-maximization}, there are three regimes of interest. 
We first present elementary constant-factor approximation algorithms (\autoref{sec:approximation-algs-maximization:sec:constant-factor}) and polynomial-factor approximation algorithms (\autoref{sec:approximation-algs-maximization:sec:polynomial-factor}). In~\autoref{sec:ptas-maximization} we then obtain approximation schemes by a reduction to the polynomial method.


\begin{figure}[t]
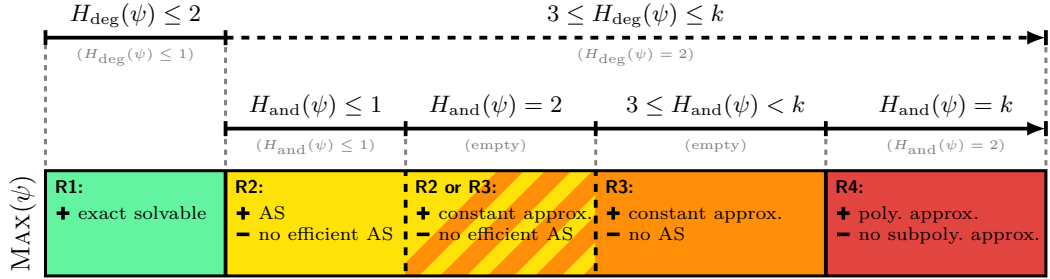

\begin{spectrumpicture}

\begin{scope}[upper]
    \drawmax
\end{scope}

\begin{scope}[yshift=.5 * \boxskip + \boxheight + \axisupperoffset]
    \draw[axis] (0, 0)
        -- node[above] {$\Hdeg(\psi) \leq 2$} node[below] {\note{$\Hdeg(\psi) \leq 1$}} ++(.18, 0);
    \draw[axis, dashed, ->] (.18, 0)
        -- node[above] {$3 \leq \Hdeg(\psi) \leq k$} node[below] {\note{$\Hdeg(\psi) = 2$}} ++(.82, 0);
    \drawaxismark{0}{0}
    \drawaxismark{.18}{0}
    \drawaxismark{1}{0}
\end{scope}
\begin{scope}[yshift=.5 * \boxskip + \boxheight + \axisloweroffset]
    \draw[axis, ->] (.18, 0)
    -- node[above] {$\Hand(\psi) \leq 1$} node[below] {\note{$\Hand(\psi) \leq 1$}} ++(.18, 0)
    -- node[above] {$\Hand(\psi) = 2$} node[below] {\note{empty}} ++(.18, 0)
    -- node[above] {$3 \leq \Hand(\psi) < k$} node[below] {\note{empty}} ++(.26, 0)
    -- node[above] {$\Hand(\psi) = k$} node[below] {\note{$\Hand(\psi) = 2$}} ++(.2, 0);
    \drawaxismark{.18}{0}
    \drawaxismark{.36}{0}
    \drawaxismark{.55}{0}
    \drawaxismark{.78}{0}
    \drawaxismark{1}{0}
\end{scope}

\draw[ind] (0, .5 * \boxskip + \boxheight)
    -- ++(0, \axisupperoffset - .5 * \axismarkheight);
\draw[ind] (.18, .5 * \boxskip + \boxheight)
    -- ++(0, \axisloweroffset - .5 * \axismarkheight)
    ++ (0, \axismarkheight)
    -- ++(0, \axisupperoffset - \axisloweroffset - \axismarkheight);
\draw[ind] (.36, .5 * \boxskip + \boxheight)
    -- ++(0, \axisloweroffset - .5 * \axismarkheight);
\draw[ind] (.55, .5 * \boxskip + \boxheight)
    -- ++(0, \axisloweroffset - .5 * \axismarkheight);
\draw[ind] (.78, .5 * \boxskip + \boxheight)
    -- ++(0, \axisloweroffset - .5 * \axismarkheight);
\draw[ind] (1, .5 * \boxskip + \boxheight)
    -- ++(0, \axisloweroffset - .5 * \axismarkheight)
    ++ (0, \axismarkheight)
    -- ++(0, \axisupperoffset - \axisloweroffset - \axismarkheight);

\end{spectrumpicture}
\caption{A reminder of our classification of $\Max(\psi)$ for $k \geq 3$. The pale labels indicate how to change the conditions to obtain the picture for $k = 2$.} \label{fig:spectrum-maximization}
\end{figure}

\subsection{Constant-Factor Approximation} \label{sec:approximation-algs-maximization:sec:constant-factor}
In this section, we provide a proof for the algorithmic part of~\autoref{thm:constant-approx-max-classification}. Phrased in terms of vectors problems (which is equivalent by~\autoref{thm:equiv-vector-problems}), we prove that:

\begin{lemma}[Constant-Factor Approximation for $\Hand(\phi) < k$] \label{lem:constant-approx-max-geq-two}
Let $\phi$ be a Boolean function on~$k$ inputs. If $\Hand(\phi) < k$ (or equivalently, $\phi$ does not have exactly one satisfying assignment), then there exists a constant-factor approximation for $\VMax(\phi)$ in time $\Order(m^{k-1})$.
\end{lemma}

We start with a definition. Let $\phi : \{0, 1\}^k \to \{0, 1\}$ be a Boolean function with the following property: Whenever $\alpha$ is a satisfying assignment of $\phi$, then the component-wise negation of $\alpha$ is also satisfying. We call $\phi$ an \emph{opposite} function. We proceed in three steps:
\begin{enumerate}
\item Exploiting the formula's structure, we design a constant approximation algorithm for all problems $\VMax(\phi)$ where $\phi$ is an opposite function (\autoref{lem:constant-approx-max-opposite}).
\item As a consequence, we turn this algorithm into a constant approximation for all problems $\VMax(\phi)$ where $\phi$ has \emph{exactly} two (arbitrary) satisfying assignments (\autoref{lem:constant-approx-max-exactly-two}).
\item Finally, let $\phi$ have \emph{at least} two satisfying assignments. We can ``cover'' all satisfying assignments by pairs and apply the previous algorithm (\autoref{lem:constant-approx-max-geq-two}).
\end{enumerate}

\begin{lemma} \label{lem:constant-approx-max-opposite}
Let $\phi$ be a $k$-variable opposite function. Then there exists a $(k + 1)$-approximation for $\VMax(\phi)$ in time $\Order(m^{k-1})$.
\end{lemma}

The proof of \autoref{lem:constant-approx-max-opposite} relies on the following simple proposition:

\begin{proposition} \label{prop:opposite}
Let $\phi$ be a $k$-variable opposite function and let $a_1, \ldots, a_k, b_1, \ldots, b_k \in \{0, 1\}$. Then:
\begin{equation*}
    \phi(a_1, \ldots, a_k) \leq \phi(b_1, \ldots, b_k) + \sum_{1 \leq i \leq k} \phi(a_1, \ldots, a_{i-1}, b_i, a_{i+1}, \ldots, a_k).
\end{equation*}
\end{proposition}
\begin{proof}
If the $a_i$'s do not satisfy $\phi$, then the claimed inequality trivially holds. Thus, assume that $\phi(a_1, \ldots, a_k) = 1$. We distinguish two cases: If there exists some index $i$ with $b_i = a_i$, we have that $\phi(a_1, \ldots, a_{i-1}, b_i, a_{i+1}, \ldots, a_k) = 1$ as well. Hence, one of the terms in the sum on the right-hand side equals one. Otherwise, we have $b_i \neq a_i$ for all $i$. Recall that $\phi$ is opposite and $\phi(a_1, \ldots, a_k) = 1$, thus $\phi(b_1, \ldots, b_k) = 1$ and we also have a contribution of at least one on the right-hand side.
\end{proof}

\begin{proof}[Proof of~\autoref{lem:constant-approx-max-opposite}]
The previous proposition immediately suggests the following algorithm for $\VMax(\phi)$. Let $X_1, \ldots, X_k$ be the given sets of vectors.
\begin{enumerate}
\item Pick $k$ arbitrary vectors $\hat x_1 \in X_1, \ldots, \hat x_k \in X_k$.
\item Compute
\begin{equation*}
    v_0 := \counting_{y \in Y} \phi(\hat x_1[y], \ldots, \hat x_k[y]).
\end{equation*}
\item Compute
\begin{equation*}
    v_i := \max_{x_1, \ldots, x_{i-1}, x_{i+1}, \ldots, x_k} \counting_{y \in Y} \phi(x_1[y], \ldots, x_{i-1}[y], \hat x_i[y], x_{i+1}[y], \ldots, x_k[y])
\end{equation*} for all $i \in [k]$, using the $\MaxSP_{k-1}$ baseline algorithm.
\item Return $v := \max_i v_i$.
\end{enumerate}
We first analyze the approximation ratio. Obviously, the algorithm never returns an overestimation: $v \leq \OPT$. Let $(x_1, \ldots, x_k)$ be vectors corresponding to an optimal solution and fix an arbitrary coordinate $y$ for which $\phi(x_1[y], \ldots, x_k[y])$ holds. From~\autoref{prop:opposite} it follows that~$y$ contributes~$1$ to at least one of the values $v_0, v_1, \ldots, v_k$. Hence, by the pigeonhole principle, there exists an index~$i$ such that $v_i \geq \OPT/(k + 1)$. In particular, we have $v \geq \OPT/(k + 1)$ and thus the algorithm indeed computes a $(k + 1)$-approximation.

Next, focus on the running time. Steps~1,~2 and~4 are executed in linear time. The third step invokes the $\MaxSP{k-1}$ baseline algorithm a constant number of times which takes time $\Order(m^{k-1})$ in total.
\end{proof}

\begin{lemma} \label{lem:constant-approx-max-exactly-two}
Let $\phi$ be a Boolean function on $k$ inputs with exactly two satisfying assignments. Then there exists a $(k + 1)$-approximation for $\VMax(\phi)$ in time $\Order(m^{k-1})$.
\end{lemma}
\begin{proof}
Let $\alpha, \beta \in \{0, 1\}^k$ be the two satisfying assignments of $\phi$. Without loss of generality, we may assume that $\alpha_i \neq \beta_i$ for all $1 \leq i \leq r$ and $\alpha_i = \beta_i$ for all $r < i \leq k$, by reordering the indices $i$ if necessary. We distinguish the following cases for $r$:
\begin{itemize}
\item $r < 1$: This case is contradictory, since $\alpha$ and $\beta$ are distinct.
\item $r = 1$: Here $\alpha_1 \neq \beta_1$ is the only position in which $\alpha$ and $\beta$ differ, so $\phi$ does not depend on the first input. Therefore, we can use the $\MaxSP_{k-1}$ baseline algorithm to solve the problem in time $\Order(m^{k-1})$.
\item $r > 1$: We brute-force over all vectors $x_i$, $r < i \leq k$. What remains is a vector maximization problem $\VMax(\phi')$ on $r$ vectors, where $\phi'$ is an opposite function. Therefore we can apply~\autoref{lem:constant-approx-max-opposite} to solve any instance in time $\Order(m^{r-1})$. We repeat this step for all $n^{k-r}$ choices of $x_{r+1}, \ldots, x_k$ and obtain a total running time of $\Order(n^{k-r} m^{r-1}) = \Order(m^{k-1})$. \qedhere
\end{itemize}
\end{proof}

\begin{proof}[Proof of~\autoref{lem:constant-approx-max-geq-two}]
If $\phi$ is the constant $\false$ function, then the problem is trivial. So let $\alpha^1, \ldots, \alpha^\ell$ denote the satisfying assignments of $\phi$, where $\ell \geq 2$.
\begin{enumerate}
\item For all indices $1 \leq i \leq k$, let $\phi^i : \{0, 1\}^k \to \{0, 1\}$ be the indicator function of $\alpha^1$ and $\alpha^i$. That is, $\phi^i(z) = 1$ if and only if $z = \alpha^1$ or $z = \alpha^i$.
\item Compute a $(k + 1)$-approximation $v^i$ for $\VMax(\phi^i)$ using~\autoref{lem:constant-approx-max-exactly-two}.
\item Return $\max_i v^i$.
\end{enumerate}
The number of iterations in step 2 is bounded by $\ell = \Order(1)$, so the running time is bounded by $\Order(m^{k-1})$.

To analyze the approximation performance, let $(x_1, \ldots, x_k)$ be an optimal solution of value $\OPT$, and let $\alpha_i$ be the satisfying assignment which occurs most frequently in $(x_1, \ldots, x_k)$. Then, by the pigeonhole principle, it follows that the optimal value of $\VMax(\phi^i)$ is at least $\OPT/\ell$. Since $v^i$ is a $(k + 1)$-approximation to that maximization problem, it follows that $v^i \geq 1/(k + 1) \mult \OPT/\ell = \Omega(\OPT)$, so our algorithm outputs a constant-factor approximation.
\end{proof}

\subsection{Polynomial-Factor Approximation}\label{sec:approximation-algs-maximization:sec:polynomial-factor}
Next, we show that even the hardest maximization problems admit polynomial-factor approximations thereby proving the upper bound in~\autoref{thm:tight-polynomial-maximization}. Recall that hereby, we mean polynomial factors in $m$. Our algorithm heavily exploits the sparse input representation.

\begin{lemma}[Polynomial-Factor Approximation] \label{lem:poly-approx-max}
Let $\phi$ be an arbitrary Boolean function on~$k$ inputs. Then, for any $\varepsilon > 0$, there exists some $\delta > 0$ such that an $m^\varepsilon$-approximation of $\VMax(\phi)$ can be computed in time $\Order(m^{k-\delta})$.
\end{lemma}

As a first step towards proving~\autoref{lem:poly-approx-max}, we prove the following simple lemma:

\begin{lemma} \label{lem:max-at-least-one}
Let $\phi$ be a Boolean function on $k$ inputs with exactly one satisfying assignment. Then, given a $\VMax(\phi)$ instance, we can decide distinguish whether $\OPT = 0$ or $\OPT > 0$ in time $\Order(m)$.
\end{lemma}

Notice that~\autoref{lem:max-at-least-one} does \emph{not} hold for minimization, which is the reason why we cannot give polynomial-factor approximations for $\VMin(\phi)$.

\begin{proof}
Let $\alpha \in \{0, 1\}^k$ be the satisfying assignment of $\phi$. Observe that $\OPT = 0$ if and only if, for all $i \in [k]$, all vectors $x_i \in X_i$ and all coordinates $y \in Y$, $x_i[y] \neq \alpha_i$. Therefore, it suffices to check individually for each set $X_i$, whether it contains some vector with an entry equal to $\alpha_i$. Clearly this takes time at most $\Order(m)$.
\end{proof}

\begin{proof}[Proof of~\autoref{lem:poly-approx-max}]
We have already shown how to compute constant approximations for all functions $\phi$ in time $\Order(m^{k-1})$, except for those with exactly one falsifying assignment (\autoref{lem:constant-approx-max-geq-two}). Thus, assume that $\phi$ has exactly one satisfying assignment $\alpha \in \{0, 1\}^k$. We distinguish two cases:
\begin{description}
\item[Case 1: \boldmath$\alpha_i = 1$ for some $i$.] Without loss of generality, we assume that $i = 1$. We split $X_1$ into a sparse part $X_1^{\text{sparse}} = \{ x_1 \in X_1 : \text{$x_1$ has Hamming weight $\leq m^\varepsilon$} \}$ and a dense part $X_1^{\text{dense}} = X_1 \setminus X_1^{\text{sparse}}$ and separately solve the instances $(X_1^{\text{sparse}}, X_2, \ldots, X_k, Y)$ and $(X_1^{\text{dense}}, X_2, \ldots, X_k, Y)$. In the sparse instance it holds that $\OPT \leq m^\varepsilon$. Thus, it suffices to check whether $\OPT > 0$ to compute an $m^\varepsilon$-approximation. By~\autoref{lem:max-at-least-one}, this can be checked in time linear time.

We are left with the dense instance, which can be solved exactly efficiently: Since every vector in $X_i^{\text{dense}}$ is of Hamming weight at least $m^\varepsilon$, there can be at most $m / m^\varepsilon = m^{1 - \varepsilon}$ vectors in $X_i^{\text{dense}}$. Therefore, we can brute-force over all $x_1 \in X_i^{\text{dense}}$ and solve the remaining $\MaxSP_{k-1}$ problem in time $\Order(m^{k-1})$ using the baseline algorithm. The total running time is $\Order(m^{1-\varepsilon} m^{k-1}) = \Order(m^{k-\varepsilon})$.

\item[Case 2: \boldmath$\alpha_i = 0$ for all $i$.] As before, let $d = |Y|$. If $d \leq m^\varepsilon$, then it again suffices to decide whether $\OPT > 0$ (which takes linear time by \autoref{lem:max-at-least-one}) to compute an $m^\varepsilon$-approximation. So we may assume that $d > m^\varepsilon$. We in fact show how to compute a constant-factor approximation. We call a vector $x_i$ \emph{light} if its Hamming weight is at most $d / (k + 1)$, and \emph{heavy} otherwise. If there exist light vectors $x_1 \in X_1, \ldots, x_k \in X_k$, then we have detected a solution of value $\geq d - k \mult d / (k + 1) = d / (k + 1) \geq \OPT / (k + 1)$.

Otherwise, some set $X_i$ contains only heavy vectors and therefore contains at most $|X_i| \leq (k + 1)m / d \leq \Order(m^{1-\varepsilon})$ vectors. After enumerating all vectors $x_i \in X_i$, it suffices to solve the remaining $\MaxSP_{k-1}$ problem using the baseline algorithm in time $\Order(m^{k-1})$. Again, the total running time is $\Order(m^{1-\varepsilon} m^{k-1}) = \Order(m^{k-\varepsilon})$. \qedhere
\end{description}
\end{proof}

\subsection{Approximation Schemes via Polynomial Method}\label{sec:ptas-maximization}
In this section, we give the approximation scheme claimed in~\autoref{thm:as-maximization-classification}. Although a similar approach also works for the minimization variants, we will omit their treatment here and instead prove the relevant claims about minimization in~\autoref{sec:approx-min}. By the equivalence of vector optimization problems, it suffices to prove the following lemma:

\begin{lemma}[Approximation Scheme for $\Hand(\phi) \leq 1$] \label{lem:max-approximation-scheme}
Let $\phi$ be a Boolean function on $k$ inputs. If $\Hand(\phi) \leq 1$, then there exists a randomized approximation scheme for $\VMax(\phi)$.
\end{lemma}

The proof of \autoref{lem:max-approximation-scheme} ultimately traces back to an application of the polynomial method \cite{AlmanW15,AlmanCW16}. However, as it turns out, it is not necessary to explicitly design algorithms in the style of the polynomial method for every problem considered here; instead, we show a black-box reduction to a problem which is already known to admit subquadratic approximation algorithms based on the polynomial method: \FurthestNeighbor{}. In the \FurthestNeighbor{} problem, we are given two sets of $n$ bit-vectors $X_1, X_2 \subseteq \{ 0,1 \}^d$ and the task is to compute the maximum Hamming distance between any two vectors $x_1 \in X_1, x_2 \in X_2$. Let $\leftrightarrow$ denote the Boolean function computing ``if and only if''. Then \FurthestNeighbor{} is equivalent to $\VMax(z_1 \centernot\leftrightarrow z_2)$.

The following result about \FurthestNeighbor{} occurs several times in the literature; in particular, it was proven by Alman, Chan and Williams~\cite{AlmanCW16}:

\begin{theorem}[{\cite[Theorem 1.5]{AlmanCW16}}] \label{lem:furthest-neighbor-polynomial-method}
For any $\varepsilon > 0$, there exists some $\delta > 0$ such that we can compute an $(1 + \varepsilon)$-approximation for \FurthestNeighbor{} in randomized time $\tilde\Order(nd + n^{2-\delta})$.
\end{theorem}

Their proof first uses a dimension reduction by locality-sensitive hashing in Hamming spaces. The improvement for the low-dimensional problem is then achieved by the polynomial method.

To state our reductions, recall the notion of a hybrid optimization problem as defined in~\autoref{sec:vector}. As a first step, we show how to obtain an approximation scheme for the hybrid problem of \FurthestNeighbor{} and its complement:

\begin{lemma} \label{lem:hybrid-to-furthest-neighbor}
Let $\Phi = \{ z_1 \centernot\leftrightarrow z_2, z_1 \leftrightarrow z_2 \}$. There exists a randomized approximation scheme for $\VMax(\Phi)$.
\end{lemma}
\begin{proof}
We introduce some notation: Let $Y_0 \subseteq Y$ denote the set of coordinates $y$ associated to $\phi_y(z_1, z_2) = (z_1 \leftrightarrow z_2)$ and let $Y_1 = Y \setminus Y_0$. Moreover, let $d_0 = |Y_0|$ and $d_1 = |Y_1|$. For a vector $x$ and a subset of dimensions $Y'$, let $\|x\|_{Y'} = |\{ y \in Y' : x[y] = 1 \}|$ denote the Hamming weight of $x$ restricted to $Y'$.

Our goal is to get rid of all coordinates in $Y_0$, as then the given hybrid instance effectively becomes an instance of \FurthestNeighbor{}. Following a naive approach, we would simply complement all entries $x_1[y]$, $y \in Y_0$. However, as $d_0$ might be as large as $n$, we cannot afford to complement directly. Instead, we will first modify the instance to guarantee that $n d_0 \leq \Order(m^{1.9})$. We have to further guarantee that $n d_1 \leq \Order(m^{1.9})$ to effectively apply the assumed \FurthestNeighbor{} algorithm.
\begin{itemize}
\item To prove that we can assume $n d_1 \leq \Order(m^{1.9})$, suppose that otherwise $n d_1 \geq m^{1.9}$. It follows that $m \leq (nd_1)^{1/1.9} \leq n^{2/1.9} \leq n^{1.1}$. Let us call a vector $x_i$ \emph{heavy} if $\|x_i\|_{Y_1} \geq n^{0.2}$, and \emph{light} otherwise. Clearly, there can be at most $m / n^{0.2} \leq n^{1.1} / n^{0.2} = n^{0.9}$ heavy vectors and thus it is feasible to explicitly enumerate all heavy vectors $x_i$ and solve the remaining $\MaxSP_{k-1}$ problem exactly using the baseline algorithm in time $\Order(n^{0.9} m) = \Order(m^{1.9})$. In the remaining instance, there are only light vectors.

We now apply the following simple hashing scheme: Let $Y_1'$ be a set of size $|Y_1'| = d_1' := n^{0.9}$ and pick a random function $h : Y_1 \to Y_1'$. For every vector $x_i$, we construct a vector $x_i' \in \{0, 1\}^{Y_0 \union Y_1'}$ with $x_i'[y'] = \bigvee_{y \in h^{-1}(y')} x_i[y]$ for all $y' \in Y_1'$ and $x_i'[y] = x_i[y]$ for all $y \in Y_0$. Notice that the value of any pair $(x_1, x_2)$ can only decrease in the hashing step. Moreover, for any fixed pair $(x_1, x_2)$, with probability $1 - 2n^{2 \mult 0.2} / n^{0.9} = 1 - \Order(n^{-0.5})$, there occurs no collision in the hashing (i.e., no two $1$-entries are mapped to the same new coordinate) and thus the value of $(x_1', x_2')$ is the same as the value for $(x_1, x_2)$. In particular, with probability $1 - \Order(n^{-0.5})$, the optimal value does not change. In that way, we have effectively reduced $d_1 = |Y_1|$ to $n^{0.9}$, and thus $n d_1 \leq n^{1.9} \leq \Order(m^{1.9})$.

\item Next, we prove that we can assume $n d_0 \leq \Order(m^{1.9})$. So suppose that $n d_0 \geq m^{1.9}$. Then $d_0 \geq m^{1.9} / n \geq \Omega(n^{0.9})$. Moreover, there exist vectors $x_1, x_2$ with $\|x_1\|_{Y_0} \leq d_0^{0.9}$ and $\|x_2\|_{Y_0} \leq d_0^{0.9}$ (as otherwise $m > n d_0^{0.9}$ and thus $m^{1.9} > n d_0$). But then the optimal value of this instance is very large: For most coordinates $y \in Y_0$, $x_1[y] = x_2[y] = 0$, which is satisfying for $(z_1 \leftrightarrow z_2)$. Thus, the optimal value is at least $\OPT \geq d_0 - d_0^{0.9} \geq (1 - \Omega(n^{-0.09})) d_0 \geq (1 - \varepsilon) d_0$ for any constant $\varepsilon > 0$ and sufficiently large $n$. Moreover, the contribution of all coordinates $y \in Y_1$ is at most $2n^{0.2}$ after eliminating all heavy vectors in the previous step. Therefore also $\OPT \leq d_0 + 2n^{0.2} \leq (1 + \varepsilon) d_0$ and we can safely return $x_1, x_2$ as an approximate optimal solution.
\end{itemize}
In combination, we can assume that $n d = n (d_0 + d_1) \leq \Order(m^{1.9})$. It is now feasible to complement all entries $x_1[y]$, $y \in Y_0$, in time $\Order(n d_0) = \Order(m^{1.9})$. What remains is a pure \FurthestNeighbor{} instance. Observe that we increased the sparsity by a lot, however, the assumed \FurthestNeighbor{} is insensitive to the sparsity~$m$. Running that algorithm takes time $\Order(n d + n^{2-\delta}) = \Order(m^{1.9} + m^{2-\delta}) = \Order(m^{2-\delta'})$ for $\delta' := \min\{\delta, 0.1\} > 0$.
\end{proof}

\begin{proof}[Proof of~\autoref{lem:max-approximation-scheme}]
As $\Hand(\phi) \leq 1$, there exist two variables, say $z_1$ and $z_2$ without loss of generality, such that for any assignment to $z_3, \ldots, z_k$, the remaining two-variable formula does not have exactly one satisfying assignment. We start by brute-forcing over all $n^{k-2}$ tuples $(x_3, \ldots, x_k)$. For each tuple, we are left to solve a hybrid optimization problem $\VMax(\Phi')$, where $\Phi'$ potentially contains all $2$-variable functions $\phi'$ of hardness $\Hand(\phi') \leq 1$. We demonstrate how to eliminate any function $\phi' \not\in \Phi$. In the process, we pay attention to not increase the sparsity $m$ beyond a constant multiple. There are several cases for $\phi'$:

\begin{description}
\item[Case 0: \boldmath$\phi'$ has 0 or 4 satisfying assignments.] Then $\phi'$ trivially accepts or rejects every input and we can delete all coordinates $y$ associated with this type.
\item[Case 1: \boldmath$\phi'$ has 1 satisfying assignment.] This case is contradictory since $\Hand(\phi') \leq 1$.
\item[Case 2: \boldmath$\phi'$ has 2 satisfying assignments.]
\item[Case 2a: \boldmath$\phi'(z_1, z_2) = (z_1 \leftrightarrow z_2)$ or $\phi'(z_1, z_2) = (z_1 \centernot\leftrightarrow z_2)$.] \vspace*{-\itemsep} Again, this case is trivial as $\phi'$ is already of the desired shape: $\phi' \in \Phi$.
\item[Case 2b: \boldmath$\phi'$ depends only on one input $z_i$.] Without loss of generality, assume that $i = 1$. If $\phi'(z_1, z_2) = z_1$, we set every vector $x_2$ to zero (at all relevant coordinates $y$) and reinterpret $\phi'$ as $z_1 \centernot\leftrightarrow z_2$. Otherwise, if $\phi'(z_1, z_2) = \neg z_1$, we instead interpret $\phi'$ as $z_1 \leftrightarrow z_2$. This transformation is clearly correct and does not increase the sparsity, reducing Case 2b to Case~2a.
\item[Case 3: \boldmath$\phi'$ has 3 satisfying assignments.] We can ``cover'' the three satisfying assignments by three pairs of assignments. For demonstration, consider the function $\phi'(z_1, z_2) = z_1 \lor z_2$. We run the previous reductions for the three functions $z_1 \centernot\leftrightarrow z_2$, $z_1$ and $z_2$ and concatenate the resulting vectors. For each of the three satisfying assignments $(1, 0)$, $(0, 1)$ and $(1, 1)$ of $\phi'$, exactly two functions among $z_1 \centernot\leftrightarrow z_2$, $z_1$ and $z_2$ are satisfied. As a result, we weighted the contribution of Case 3 by a factor of $2$. To compensate, we simply execute each of the previous cases twice and concatenate the output vectors. The optimal value of the resulting hybrid problem equals two times the former optimal value.
\end{description}
Running the approximation scheme for $\VMax(\Phi)$ takes time $\Order(m^{2-\delta})$, and we can recover the original optimal value by halving the output. Repeating these steps for all $n^{k-2}$ tuples $(x_3, \ldots, x_k)$ in the original instance takes time $\Order(n^{k-2} m^{2-\delta}) = \Order(m^{k-\delta})$.
\end{proof}
\section{Approximation Algorithms for Minimization} \label{sec:approx-min}

\begin{figure}[t]
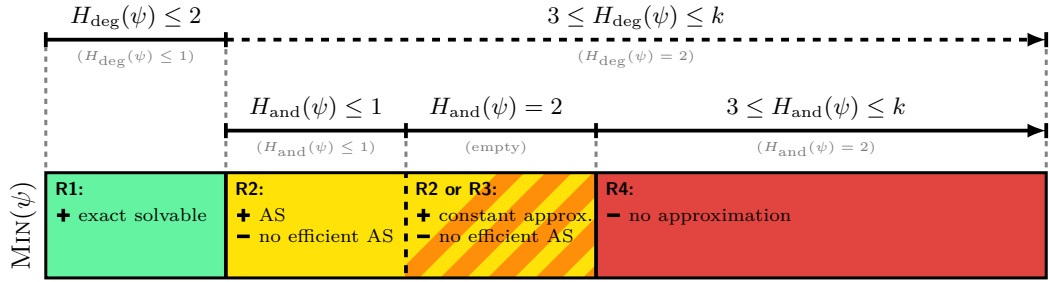

\begin{spectrumpicture}

\begin{scope}[upper]
    \drawmin
\end{scope}

\begin{scope}[yshift=.5 * \boxskip + \boxheight + \axisupperoffset]
    \draw[axis] (0, 0)
        -- node[above] {$\Hdeg(\psi) \leq 2$} node[below] {\note{$\Hdeg(\psi) \leq 1$}} ++(.18, 0);
    \draw[axis, dashed, ->] (.18, 0)
        -- node[above] {$3 \leq \Hdeg(\psi) \leq k$} node[below] {\note{$\Hdeg(\psi) = 2$}} ++(.82, 0);
    \drawaxismark{0}{0}
    \drawaxismark{.18}{0}
    \drawaxismark{1}{0}
\end{scope}
\begin{scope}[yshift=.5 * \boxskip + \boxheight + \axisloweroffset]
    \draw[axis, ->] (.18, 0)
    -- node[above] {$\Hand(\psi) \leq 1$} node[below] {\note{$\Hand(\psi) \leq 1$}} ++(.18, 0)
    -- node[above] {$\Hand(\psi) = 2$} node[below] {\note{empty}} ++(.18, 0)
    -- node[above] {$3 \leq \Hand(\psi) \leq k$} node[below] {\note{$\Hand(\psi) = 2$}} ++(.46, 0);
    \drawaxismark{.18}{0}
    \drawaxismark{.36}{0}
    \drawaxismark{.55}{0}
    \drawaxismark{1}{0}
\end{scope}

\draw[ind] (0, .5 * \boxskip + \boxheight)
    -- ++(0, \axisupperoffset - .5 * \axismarkheight);
\draw[ind] (.18, .5 * \boxskip + \boxheight)
    -- ++(0, \axisloweroffset - .5 * \axismarkheight)
    ++ (0, \axismarkheight)
    -- ++(0, \axisupperoffset - \axisloweroffset - \axismarkheight);
\draw[ind] (.36, .5 * \boxskip + \boxheight)
    -- ++(0, \axisloweroffset - .5 * \axismarkheight);
\draw[ind] (.55, .5 * \boxskip + \boxheight)
    -- ++(0, \axisloweroffset - .5 * \axismarkheight);
\draw[ind] (1, .5 * \boxskip + \boxheight)
    -- ++(0, \axisloweroffset - .5 * \axismarkheight)
    ++ (0, \axismarkheight)
    -- ++(0, \axisupperoffset - \axisloweroffset - \axismarkheight);

\end{spectrumpicture}
\caption{A reminder of our classification of $\Min(\psi)$ for $k \geq 3$. The pale labels indicate how to change the conditions to obtain the picture for $k = 2$.} \label{fig:spectrum-minimization}
\end{figure}
In this section we give the algorithms for approximating $\VMin(\phi)$. Our goal is to prove the following~\autoref{lem:approx-min}. The upper bounds for both~\autoref{thm:as-minimization-classification} and~\autoref{thm:constant-approx-min-classification} follow immediately, by the equivalence of $\Min(\psi)$ and $\VMin(\phi)$.
\begin{lemma} \label{lem:approx-min}
Let $\phi$ be a Boolean function on $k$ inputs.
\begin{itemize}
\item If $\Hand(\phi) \leq 1$, then there exists a randomized approximation scheme for $\VMin(\phi)$.
\item If $\Hand(\phi) \leq 2 < k$, then there exists a randomized constant-factor approximation algorithm for $\VMin(\phi)$ in time $\Order(m^{k-\delta})$, for some $\delta > 0$.
\medbreak\noindent
More generally: There exists a constant-factor approximation for $\VMin(\phi)$ in time $\Order(m^{k-\delta})$ for some $\delta > 0$ if and only if $\VZero(\phi)$ can be solved in time $\Order(m^{k-\delta'})$ for some $\delta' > 0$.
\end{itemize}
\end{lemma}

In contrast to the maximization case where we applied different techniques for each approximation regime, here all our results are a consequence of a general reduction from approximating $\VMin(\phi)$ to exactly solving $\VZero(\phi)$, where we use locality-sensitive hashing. Recall that $\VZero(\phi)$ is the problem of testing whether there exists a solution of value zero, or equivalently, testing whether $\OPT = 0$ for the minimization problem $\VMin(\phi)$. Our strategy is to first show that a sufficiently fast \emph{listing} algorithm for $\VZero(\phi)$ implies the claimed approximation algorithms. Afterwards, we demonstrate how derive efficient listing algorithms for $\VZero(\phi)$, given the assumptions $\Hand(\phi) \leq 1 < k$ respectively $\Hand(\phi) \leq 2 < k$.

Formally, we call an algorithm a \emph{listing algorithm} for $\VZero(\phi)$ if, given an integer $L$, it returns a list of $L$ distinct solutions $(x_1, \ldots, x_k)$ to $\VZero(\phi)$ (if there are less than~$L$ distinct solutions, then it exhaustively lists all solutions). We usually parameterize the running time of listing algorithms with $T(m, L)$, where $m$ is the size of the input (as before).

We start with the following technical lemma, which generalizes the baseline algorithm to the case where only a subset of all tuples $(x_1, \dots, x_k)$ are considered candidate solutions.

\begin{lemma} \label{lem:candidates-exactly}
Let $\phi$ be a Boolean function on $k$ inputs, and let $\delta > 0$. Given a $\VMin(\phi)$ instance, and a set of $\Order(n^{k-\delta})$ candidates $C \subseteq X_1 \times \dots \times X_k$, we can compute $\Val(x_1, \ldots, x_k)$ for all candidates $(x_1, \dots, x_k) \in C$ in time $\Order(m^{k-\delta/3})$.
\end{lemma}
\begin{proof}
The algorithm executes the following three steps:
\begin{itemize}
\item We call a vector $x_i$ \emph{heavy} if its Hamming weight exceeds $m^{\delta/3}$, and \emph{light} otherwise. As the total Hamming weight of all vectors is at most $\Order(m)$, there can be at most $\Order(m^{1-\delta/3})$ heavy vectors. Therefore, we can brute-force over all heavy vectors $x_k$ and solve the remaining $\MinSP_{k-1}$ problem in time $\Order(m^{k-1})$ using the baseline algorithm (which in fact computes all values $\Val(x_1, \ldots, x_k)$ in the same running time; see~\cite[Theorem 22 of the full version]{BringmannCFK21}). This step takes time $\Order(m^{k-\delta/3})$ and afterwards, we can assume that all vectors~$x_k$ are light.

\item Next, call a tuple $(x_1, \ldots, x_{k-1})$ a \emph{heavy prefix} if the number of vectors $x_k$ such that $(x_1, \dots, x_k)$ is a candidate is at least $n^{1-2\delta/3}$. There can be at most $n^{k-1-\delta/3}$ heavy prefixes, as otherwise $|C| > n^{k-\delta}$. Enumerating all such prefixes and solving the remaining $\MinSP_1$ problem by the baseline algorithm takes time $\Order(n^{k-1-\delta/3} m) = \Order(m^{k-\delta/3})$.

\item For any prefix $(x_1, \dots, x_{k-1})$ which is left, we can assume that there are at most $n^{1-2\delta/3}$ completions $x_k$ such that $(x_1, \dots, x_k)$ is a candidate. Moreover, recall that every such vector $x_k$ is light, and therefore of Hamming weight at most $\Order(m^{\delta/3})$. Hence, we brute-force over all remaining prefixes and solve a $\MinSP_1$ problem of size $\Order(n^{1-2\delta/3} m^{\delta/3})$. Again, we use the baseline algorithm, so this step takes time $\Order(n^{k-1} n^{1-2\delta/3} m^{\delta/3}) = \Order(m^{k-\delta/3})$ as well. \qedhere
\end{itemize}
\end{proof}

\begin{lemma}[Approximation via Listing Zeros] \label{lem:lsh}
Let $\phi$ be a variable Boolean function on $k \geq 2$ inputs so that $L$ solutions to $\VZero(\phi)$ can be listed in time $T(m, L)$. Then, for any $c > 1$ and $0 < \gamma < k$, there is a randomized algorithm that with high probability computes a $c$-approximation of $\VMin(\phi)$ in time
\begin{equation*}
    \tilde\Order\!\left(m^{k-\frac{\gamma(1-1/c)}{3}} + \sum_{i=0}^{\Order(\log n)} \min\{2^i, n^{\gamma/c}\} \mult T\!\left(\!m, \frac{n^{k-\gamma(1-1/c)}}{2^i}\!\right)\!\right)\!.
\end{equation*}
\end{lemma}
\begin{proof}
We will consider the following gap version of the $\VMin(\phi)$ problem: Along with the input, we are given an additional target $t$. The task is to compute an output satisfying the following requirements:
\begin{itemize}
\item If $\OPT \leq t$, then return ``\emph{small}'' and attach a solution of value at most $c \mult t$.
\item If $\OPT > c \mult t$, then return ``\emph{large}''.
\item Otherwise, we do not care about the output.
\end{itemize}
An algorithm for this gap problem translates into an approximation algorithm for $\VMin(\phi)$ by using binary search on $t$ to determine $\OPT$ up to a factor of $c$. For that reason, we focus on the gap problem for the remainder of this proof.

Let $X_1, \ldots, X_k \subseteq \{0, 1\}^{|Y|}$ be the given gap instance with target $t$, and let $R$ and $S$ be parameters to be specified later. Throughout the algorithm, we maintain a global counter $L$ initialized to $L = 8n^{k-\gamma(1-1/c)} + 1$. The following steps are repeated~$R$ times:
\begin{enumerate}
\item Sample $S$ coordinates $y_1, \ldots, y_S \in Y$ uniformly at random. For a vector $x_i \in X_i$, we write $x_i'$ for the projected vector $(x_i[y_1], \ldots, x_i[y_S]) \in \{0, 1\}^S$ and we set $X_i' = \{ x_i' : x_i \in X_i \}$.
\item List up to $L$ solutions to the $\VZero(\phi)$ instance $(X_1', \ldots, X_k')$. Afterwards, decrease $L$ by the number of detected solutions.
\item For every tuple $(x_1', \ldots, x_k')$ listed in the previous step, check whether $\Val(x_1, \ldots, x_k) \leq c \mult t$. In that case, report ``\emph{small}'' and attach $(x_1, \ldots, x_k)$.
\end{enumerate}
If after all iterations no solution was detected, output ``\emph{large}''.

We start with the correctness argument. For a fixed round and a tuple $(x_1, \ldots, x_k)$, we distinguish two cases:
\begin{itemize}
\item If $\Val(x_1, \ldots, x_k) \leq t$, then: $\Pr\big[\Val(x_1', \ldots, x_k') = 0\big] \geq \big(1 - \frac t{|Y|}\big)^S =: p_1^S$.
\item If $\Val(x_1, \ldots, x_k) > ct$, then: $\Pr\big[\Val(x_1', \ldots, x_k') = 0\big] \leq \big(1 - \frac{ct}{|Y|}\big)^S =: p_2^S$.
\end{itemize}
The strategy is to set $R$ large enough such that any (fixed) solution $(x_1, \ldots, x_k)$ of value at most~$t$ is mapped to a tuple $(x_1', \ldots, x_k')$ of value zero in at least one round. It suffices to choose $R = p_1^{-S}$ as then the probability of failing in all $R$ rounds is bounded by $(1 - p_1^S)^R \leq 1/e$.

Next, we specify $S$ in such a way that the number of false positives -- that is, tuples $(x_1', \ldots, x_k')$ of value zero whereas $(x_1, \ldots, x_k)$ has large value -- is bounded by $\Order(m^{k-1})$ with good probability. Specifically, set
\begin{equation*}
    S := \frac{\gamma\log n}{\log(1/p_2)} = \frac{\gamma\log n}{\log (1 - \frac{ct}{|Y|})^{-1}} \leq \frac{\gamma|Y| \log n}{ct}.
\end{equation*}
It follows that
\begin{equation*}
    R = p_1^{-S} = \left(1 - \frac t{|Y|}\right)^{-S} \leq \left(1 - \frac t{|Y|}\right)^{\textstyle-\frac{\gamma|Y| \log n}{ct}} \leq n^{\gamma/c}.
\end{equation*}
The expected number of false positives among all rounds is therefore bounded by $R \mult n^k \mult p_2^S \le n^{\gamma/c} \mult n^k \mult n^{-\gamma} = n^{k-\gamma(1-1/c)}$. Using Markov's inequality, the number of false positives is at most $8n^{k-\gamma(1-1/c)}$ with probability at least $7/8$. Hence, it suffices to list $L = 8n^{k-\gamma(1-1/c)} + 1$ many solutions to all $\VZero(\phi)$ instances together, in order find at least one true positive with probability $1 - 1/e - 1/8 > 1/2$. By repeating the whole procedure $\Order(\log n)$ times, we amplify the success probability to $1 - n^{-\Omega(1)}$.

Finally, focus on the running time analysis. The first step runs in time $\Order(m)$ per round. The second and third steps are more involved. To achieve the claimed time bound, we have to implement the second step slightly more efficiently: Instead of running the listing algorithm with the global capacity $L$ directly, we use binary search to test for~$2^i$ solutions, $i = 0, 1, \ldots, \log L$. If in the $i$-th iteration there are strictly less than $2^i$ solutions, we stop and continue as above. Among all $R$ rounds, we execute that $i$-th step at most $L / 2^i$ times. Hence, the running time of step 2 is bounded by
\begin{equation*}
    \Order\!\left(\sum_{i=0}^{\log L} \min\left\{\frac{L}{2^i}, R\right\} \mult T(m, 2^i)\right) = \Order\!\left(\sum_{i=0}^{\log L} \min\{2^i, R\} \mult T\left(\!m, \frac{L}{2^i}\right)\right),
\end{equation*}
which is as stated, after plugging in the parameters from before. In the third step, across all rounds, we test at most $\Order(n^{k-\gamma(1-1/c)})$ tuples $(x_1, \dots, x_k)$ for their exact value $\Val(x_1, \dots, x_k)$. That task can be solved in time $\Order(m^{k-\gamma(1-1/c)/3})$ using \autoref{lem:candidates-exactly}.
\end{proof}

In the following two corollaries, we show how to apply~\autoref{lem:lsh} to obtain approximation schemes and constant-factor approximations.

\begin{corollary}[Constant-Factor Approximation via Listing Zeros] \label{cor:lsh-constant-approx}
Let $\phi$ be a Boolean function on~$k$ inputs. If $\VZero(\phi)$ admits a listing algorithm in time $\Order(m^{k-\delta} \mult L^{\delta/k})$ for some $\delta > 0$, then there is a randomized algorithm computing a constant-factor approximation of $\VMin(\phi)$ in time $\Order(m^{k-\delta'})$ for some $\delta' > 0$.
\end{corollary}
\begin{proof}
We use~\autoref{lem:lsh} with $\gamma = 1$ and some constant $c > k/\delta + 1$. The running time is bounded by
\begin{align*}
    &\mathrel{\phantom=} \tilde\Order(m^{k-\frac{1-1/c}{3}} + n^{1/c} \mult m^{k-\delta} \mult (n^{k-(1-1/c)})^{\frac\delta k}) \\
    &= \tilde\Order(m^{k-\frac{1-1/c}{3}} + m^{k-\frac\delta k + \frac{1 + \delta/k}c}) \\
    &= \Order(m^{k-\delta'}),
\end{align*}
for $\delta' := \min\{\frac{1-1/c}3, \frac\delta k - \frac{1 + \delta/k}c\} > 0$. 
\end{proof}

\begin{corollary}[Approximation Scheme via Listing Zeros] \label{cor:lsh-as}
Let $\phi$ be a Boolean function on~$k$ inputs. If $\VZero(\phi)$ admits a listing algorithm in time $\tilde\Order(m^{k-\delta} + L)$ for some $\delta > 0$, then there is a randomized approximation scheme for $\VMin(\phi)$.
\end{corollary}
\begin{proof}
We use~\autoref{lem:lsh} with $\gamma = \delta$ and $c = 1 + \varepsilon$, for an arbitrarily small constant $\varepsilon > 0$. The running time is bounded by
\begin{align*}
    &\mathrel{\phantom=} \tilde\Order\!\left(m^{k-\frac{\delta(1-1/(1 + \varepsilon))}{3}} + \sum_{i=0}^{\Order(\log n)} \min\{2^i, n^{\delta/(1+\varepsilon)}\} \mult \!\left(\!m^{k-\delta} + \frac{n^{k-\delta(1-1/(1+\varepsilon))}}{2^i}\!\right)\!\right)\! \\
    &= \tilde\Order(m^{k-\frac{\delta(1-1/(1 + \varepsilon))}{3}} + n^{\delta/(1 + \varepsilon)} m^{k-\delta} + n^{k-\delta(1-1/(1+\varepsilon))}) \\
    &= \tilde\Order(m^{k-\frac{\delta(1-1/(1 + \varepsilon))}{3}}).
\end{align*}
Notice that indeed $\frac{\delta(1-1/(1+\varepsilon))}{3} > 0$.
\end{proof}

It remains to show how to list solutions to $\VZero(\phi)$ fast enough to apply Corollaries~\ref{cor:lsh-constant-approx} and~\ref{cor:lsh-as}. First we give a general algorithm, which can be then directly applied with~\autoref{lem:lsh} to yield~\autoref{cor:lsh-constant-approx}:

\begin{lemma} \label{lem:zero-to-listing}
Let $\phi$ be a Boolean function on~$k$ inputs. If $\VZero(\phi)$ can be solved in time $\Order(m^{k-\delta})$ for some $\delta > 0$, then there is an algorithm listing $L$ solutions of $\VZero(\phi)$ in time $T(m, L) = \tilde\Order(m^{k-\delta} \mult L^{\delta/k})$.
\end{lemma}
\begin{proof}
Let $X_1, \ldots, X_k \subseteq \{0, 1\}^Y$ be the given input. We write $m(X_i)$ for the total number of ones in all vectors in $X_i$. Each set $X_i$ can be partitioned into sets $X_i = X_i^1 \union \{\hat x_i\} \union X_i^2$ where $m(X_i^1) < m(X_i) / 2$ and $m(X_i^2) < m(X_i) / 2$ in the following way: We greedily add elements to $X_i^1$ until $m(X_i^1) < m(X_i) / 2$ but $m(X_i^1 \union \{\hat x_i\}) \geq m(X_i) / 2$ for the next element $\hat x_i$. We take $\{\hat x_i\}$ to be a set on its own and set $X_i^2 = X_i \setminus (X_i^1 \union \{\hat x_i\})$.

We can list all solutions $(x_1, \ldots, x_k)$ involving one of the middle elements $\hat  x_i$ in time $\Order(m^{k-1})$ using the baseline algorithm. There are $2^k$ many subproblems remaining, where each is of the form $(X_1^{j_1}, \ldots, X_k^{j_k})$ for $j_1, \ldots, j_k \in \{1, 2\}$. Using the assumed efficient algorithm for $\VZero(\phi)$, we test for each of these instances whether it contains a solution. For each subproblem which is a YES instance, continue recursively. If at some point we reported more than $L$ solutions, stop.

Since we cannot miss solutions in the branching, the correctness is clear. For the running time, note that in each recursive call the sparsity was halved, and therefore the recursion depth is bounded by $\Order(\log m)$. Each individual call on the $i$-th recursion level takes time $\Order((m/2^i)^{k-\delta})$ and there are at most $\min\{L, 2^{ik}\}$ calls on that level. Hence, we can bound the total running time by
\begin{equation*}
    T(m, L) \leq \sum_{i=0}^{\Order(\log m)} \min\{L, 2^{ik}\} \mult \Order((m/2^i)^{k-\delta}) \leq \tilde\Order(m^{k-\delta} L^{\delta/k}).
\end{equation*}
The last inequality holds because there are $\Order(\log m)$ summands and the largest one has $L = 2^{ik}$.
\end{proof}

When $\phi$ has and-hardness $\Hand(\phi) \leq 1 < k$ we can improve upon~\autoref{lem:zero-to-listing}:

\begin{lemma} \label{lem:h1-to-listing}
Let $\phi$ be a Boolean function on~$k$ inputs. If $\Hand(\phi) \leq 1 < k$, then there is an algorithm listing $L$ solutions of $\VZero(\phi)$ in time $T(m, L) = \tilde\Order(m^{k-1} + L)$.
\end{lemma}
\begin{proof}
This proof is similar to~\cite[Lemma 11]{BringmannFK19}. Let us start with the case $k = 2$; we later point out how to generalize. We distinguish several formulas $\phi$:

\begin{description}
\item[Case 0: \boldmath$\phi$ has 0 or 4 satisfying assignments.] Then $\phi$ trivially accepts or rejects every input and we can list $L$ solutions in time $\Order(L)$.
\item[Case 1: \boldmath$\phi$ has 1 satisfying assignment.] This case is contradictory since $\Hand(\phi) \leq 1$.
\item[Case 2: \boldmath$\phi$ has 2 satisfying assignments.]
\item[Case 2a: \boldmath$\phi(z_1, z_2) = (z_1 \leftrightarrow z_2)$ or $\phi(z_1, z_2) = (z_1 \centernot\leftrightarrow z_2)$.] \vspace*{-\itemsep} In this case, any solution $(x_1, x_2)$ requires that either $x_1$ and $x_2$ are equal in every coordinate, or exactly negations of each other. In any case, for each vector $x_1$ there exists at most one choice $x_2$ such that $(x_1, x_2)$ is a solution. Thus, it is easy to list all solutions in time $\tilde\Order(m)$.
\item[Case 2b: \boldmath$\phi$ depends only on one input $z_i$.] Without loss of generality, assume that $i = 1$. Then there can be at most one vector $x_1$ which participates in a solution. It takes time $\Order(m)$ to identify this vector and to list all solutions $(x_1, x_2)$, for all $x_2$.
\item[Case 3: \boldmath$\phi$ has 3 satisfying assignments.] Then $\phi$ has only one falsifying assignment, say $(1, 1)$. That means that the only possible solution $(x_1, x_2)$ is that both $x_1$ and $x_2$ are the all-ones vector. It easy to check whether such a pair exists in time $\Order(m)$.
\end{description}
Let $\Phi = \{ \phi(z_1, z_2) : \Hand(\phi) \leq 1 \}$. By combining the previous cases, we can also list $L$ solutions of the hybrid problem $\VZero(\Phi)$ in time $\tilde\Order(m + L)$. Indeed, we can list \emph{all} solutions for each individual \emph{nontrivial} problem $\VZero(\phi)$ in time $\tilde\Order(m)$, and report only those solutions which are part of every individual solution set. There are two corner cases: If there are some coordinates associated to the constant $\false$-function, then we return no solution at all, and if all coordinates are associated to the constant $\true$-function, then we proceed as in Case 0.

Finally, assume that $k > 2$. We show how to list $L$ solutions to $\VZero(\phi)$. Since $\Hand(\phi) \leq 1$, there exists a set of $k - 2$ indices, say $\{3, \ldots, k\}$, which we can brute-force over, such that the remaining problem is exactly $\VZero(\Phi)$. For any tuple $(x_3, \ldots, x_k)$, we first compute the number of solutions $(x_1, x_2, x_3, \ldots, x_k)$ in time $\tilde\Order(m)$, using the same case distinction as before. Afterwards, we split our global budget $L$ in such a way that we never list too many solutions of $\VZero(\Phi)$, and run the previous listing algorithm for every tuple $(x_3, \ldots, x_k)$. The total running time is $\tilde\Order(n^{k-2} m + L) = \tilde\Order(m^{k-1} + L)$. 
\end{proof}

Finally, we put things together and prove the main result of this section:

\begin{proof}[Proof of~\autoref{lem:approx-min}]
The first bullet follows by combining~\autoref{cor:lsh-as} and \autoref{lem:h1-to-listing}. The second bullet, and in particular the equivalence of $\VZero(\phi)$ and approximating $\VMin(\phi)$ within a constant factor, follows by combining~\autoref{cor:lsh-constant-approx} and \autoref{lem:zero-to-listing}.
\end{proof}

One might hope to give an approximation scheme for $\Hand(\phi) = 2$ in a similar way. Unfortunately, it is unlikely that a listing algorithm for $\VZero(\phi)$ in time $\tilde\Order(n^{k-\delta} + L)$ exists: There exists a fine-grained reduction from triangle detection to $\VZero(\phi)$~\cite[Lemma 20]{BringmannFK19}, and thus we are implicitly also hoping for an output-linear triangle listing algorithm. However, previous work \cite{Patrascu10,BjorklundPWZ14} suggests that we cannot avoid a dependence on $L \mult \poly(n)$.
\section{Additive Approximation Algorithms} \label{sec:additive}
We say that an algorithm computes an \emph{additive $c$-approximation} if it returns a solution whose value is contained in the interval $[\OPT - c, \OPT + c]$. Similarly as for multiplicative approximations, we can define approximation schemes in the additive setting.

\begin{definition}[Additive Approximation Scheme] \label{def:additive-approximation-scheme}
Let $\psi$ be an $\OptSP_k$ formula. We say that $\Opt(\psi)$ admits an \emph{additive approximation scheme} if, for any $\varepsilon > 0$, there exists some $\delta > 0$ and an algorithm computing an additive $\varepsilon |Y|$-approximation of $\Opt(\psi)$ in time $\Order(m^{k-\delta})$.

We say that $\Opt(\psi)$ admits an \emph{efficient} additive approximation scheme if there exists some fixed constant $\delta > 0$, such that for any $\varepsilon > 0$ there exists an algorithm computing an additive $\varepsilon |Y|$-approximation of $\Opt(\psi)$ in time $\Order(m^{k-\delta})$.
\end{definition}

The complexity landscape of multiplicative approximations turned out to be rich, featuring natural problems in all major regimes of approximation guarantees. In contrast, we draw a significantly simpler picture for additive approximations.

\begin{theorem}[Additive Approximation] \label{thm:additive-approximation-classification}
Let $\psi$ be an $\OptSP_k$ graph formula.
\begin{itemize}
\item There exists a randomized additive approximation scheme for $\Opt(\psi)$.
\item If $3 \leq \Hdeg(\psi)$ or $\Hdeg(\psi) = k$, then there exists no efficient additive approximation scheme for $\Opt(\psi)$ unless the Sparse \MAXThreeSAT{} hypothesis fails.
\end{itemize}
\end{theorem}

In summary, all problems $\Opt(\psi)$ admit approximation schemes, but no efficient ones (except for the problems that can be optimized exactly). 

In this section we will show the first item of~\autoref{thm:additive-approximation-classification}, i.e. that all problems $\Opt(\psi)$ admit additive approximation schemes. We defer the proof of the second item to~\autoref{sec:hardness-approx:sec:no-eas}. In particular, we will show the following:

\begin{lemma}[Additive Approximation Scheme] \label{lem:additive-approximation-algorithms}
Let $\phi$ be an arbitrary Boolean function. There exists a randomized additive approximation scheme for $\VOpt(\phi)$.
\end{lemma}

Combining~\autoref{lem:additive-approximation-algorithms} with the equivalence of $\Opt(\psi)$ and vector problems $\VOpt(\psi)$ (which also holds for additive approximations, see~\autoref{sec:vector} for details), we obtain a proof for the algorithmic part of~\autoref{thm:additive-approximation-classification}.

The techniques used to obtain multiplicative approximations often differed substantially for maximization and minimization. As it turns out, to obtain additive approximation schemes we can treat both in a unified way. Indeed, to prove~\autoref{lem:additive-approximation-algorithms}, we demonstrate how to reduce every problem $\VOpt(\phi)$ to an instance of low-dimensional Maximum Inner Product (\MaxIP{}), or equivalently, low-dimensional Minimum Inner Product (\MinIP{}). It is known that both \MaxIP{} and \MinIP{} admit polynomial-time improvements in the low-dimensional regime using the polynomial method~\cite{AlmanW15,AlmanCW16}:

\begin{theorem}[\cite{AlmanW15}] \label{lem:max-ip-polymethod}
For every $c > 0$, there exists some $\delta > 0$ such that \MaxIP{} on $d = c \log n$ dimensions can be solved in time $\Order(n^{2-\delta})$.
\end{theorem}

As a first algorithmic step, we establish a simple dimension reduction:

\begin{lemma} \label{lem:additive-dimension-reduction}
Let $\phi$ be an arbitrary Boolean function. For every $\varepsilon > 0$, there is a randomized linear-time reduction, turning a given an instance of $\VOpt(\phi)$ on $n$ vectors of dimension $d$ into another $\VOpt(\phi)$ instance on $n$ vectors of dimension $d' = \Order(\log n)$, such that with high probability we have $|\OPT / d - \OPT' / d'| \leq \varepsilon$, where $\OPT'$ is the optimal value of the constructed instance.
\end{lemma}
\begin{proof}
Set $d' := c \varepsilon^{-2} \log n$ for some sufficiently large constant $c$. Sample $d'$ coordinates from~$Y$ uniformly and independently at random, and let $X_1', \ldots, X_k'$ be the projections of the input vectors to these $d'$ coordinates. For any $x_1' \in X_1', \ldots, x_k' \in X_k'$, we have that $\Val(x_1', \ldots, x_k')$ is a random variable with expectation $\Ex[\Val(x_1', \ldots, x_k')] = \frac{d'}{d} \Val(x_1,\dots,x_k)$. Moreover, since each coordinate was sampled independently and uniformly at random, by an additive Chernoff bound we have
\begin{equation*}
  \Pr\left(\left|\frac{\Val(x_1,\dots,x_k)}{d} - \frac{\Val(x'_1,\dots,x'_k)}{d'}\right| 
  > \varepsilon\right) \leq 2 \exp(-2 \varepsilon^2 d') \leq \Order(n^{-2c}).
\end{equation*}
Thus, by a union bound over all tuples $(x_1, \ldots, x_k)$, it holds that $|\OPT / d - \OPT' / d'| \leq \varepsilon$ with probability $1 - n^{-\Omega(1)}$, where $\OPT'$ is the optimal value of the new instance. As claimed, we can construct the instance in time $\Order(m)$ by once iterating over all one-entries, after subsampling the~$d'$ relevant coordinates.
\end{proof}

Having applied the dimension reduction, we can complete the reduction from additively approximating $\VMax(\phi)$ to low-dimensional \MaxIP{}.

\begin{lemma} \label{lem:additive-approximation-to-max-ip}
Let $\phi$ be an arbitrary Boolean function. There exists a randomized additive approximation scheme for $\VMax(\phi)$.
\end{lemma}
\begin{proof}
Let $(X_1, \ldots, X_k, Y)$ be the given $\VMax(\phi)$ instance, and fix $\varepsilon > 0$. In order to compute an additive $\varepsilon d$-approximation of $\VMax(\phi)$, we can assume that $d = |Y| = \Order(\log n)$ by applying~\autoref{lem:additive-dimension-reduction}. Our goal is to solve the resulting instance exactly by ``covering'' every satisfying assignment of $\phi$ by the all-ones assignment (which is satisfying for \kMaxIP{k}). More precisely, let $\alpha_1, \ldots, \alpha_\ell \in \{0, 1\}^k$ denote the satisfying assignments of $\phi$. Let $Y' = Y \times [\ell]$ be a new set of dimensions and for each vector $x_i \in X_i$, we create a new vector $x_i' \in X_i'$ with
\begin{equation*}
  x_i'[(y, j)] =
  \begin{cases}
    x_i[y] &\text{if $\alpha_j[i] = 1$,} \\
    1 - x_i[y] &\text{otherwise.}
  \end{cases}
\end{equation*}
It is easy to check that $\Val(x_1, \ldots, x_k) = \innerprod{x_1'}{\ldots, x_k'}$, and therefore it suffices to solve the \kMaxIP{k} instance $(X_1', \ldots, X_k', Y')$. We brute-force over $k-2$ variables, and solve the remaining \MaxIP{} instance using the assumed efficient algorithm (notice that indeed $d' = |Y'| = d \ell = \Order(\log n)$). The total running time is $\Order(n^{k-2} n^{2-\delta}) = \Order(n^{k-\delta}) = \Order(m^{k-\delta})$ for some $\delta > 0$.
\end{proof}

Instead of proving an analogous reduction to from $\VMin(\phi)$ to \MinIP{}, we instead show the following more general statement, which we will reuse in~\autoref{sec:hardness-approx:sec:no-eas} to rule out efficient additive approximation schemes. 

\begin{lemma} \label{lem:additive-max-min}
Let $\phi$ be an arbitrary Boolean function. Via a randomized reduction, there exists an additive approximation scheme for $\VMax(\phi)$ if and only if there exists an additive approximation scheme for $\VMin(\phi)$. The same holds for efficient additive approximating schemes.
\end{lemma}
\begin{proof}
We focus only on one direction, the other one is analogous. Suppose there exists an additive approximation scheme for $\VMax(\phi)$, and fix some $\varepsilon > 0$. We can compute an additive $\varepsilon d$-approximation for a given $\VMin(\phi)$ instance $(X_1, \ldots, X_k, Y)$ on dimension-$d$ vectors using the following algorithm:
\begin{enumerate}
\item Apply the dimension reduction from~\autoref{lem:additive-dimension-reduction} with parameter $\varepsilon / 2$ to obtain an equivalent $\VMin(\phi)$ instance $(X_1', \ldots, X_k', Y')$ on $d' = |Y'| = \Order(\log n)$ dimensions.
\item Complement all vectors in $X_1', \ldots, X_k'$ to obtain a $\VMax(\phi)$ instance $(X_1'', \ldots, X_k'', Y')$.
\item Use the assumed additive approximation scheme to compute an additive $\frac{\varepsilon d'}{2}$-approx-imation~$V$ of the $\VMax(\phi)$ instance $(X_1'', \ldots, X_k'', Y')$, and return $d - \frac{d}{d'} V$.
\end{enumerate}
Let $\OPT$ denote the optimal value of the original $\VMin(\phi)$ instance, and let $\OPT'$ and $\OPT''$ be the optimal values of the $\VMin(\phi)$ and $\VMax(\phi)$ instances computed in the second and third steps, respectively. The algorithm guarantees that $|\OPT'' - V| \leq \frac{\varepsilon d'}{2}$ and by construction it holds that $\OPT' = d' - \OPT''$. Moreover, with high probability the dimension reduction succeeds in which case $|\OPT - \frac{d}{d'}\OPT'| \leq \frac{\varepsilon d}{2}$. Plugging these three bounds together yields that $|\OPT - (d - \frac{d}{d'} V)| \leq \varepsilon d$.

Running the first step of the algorithm takes time $\Order(m)$ by~\autoref{lem:additive-dimension-reduction}. The complementation takes time $\tilde\Order(m)$, as the dimension was reduced to $d' = \Order(\log n)$. Finally, applying the additive approximation scheme for $\VMax(\phi)$ takes time $\tilde\Order(m^{k-\delta})$ for some $\delta > 0$, which dominates the total running time. If the $\VMax(\phi)$ additive approximation scheme is efficient, then so is the constructed one for $\VMin(\phi)$.
\end{proof}

Finally, note that \autoref{lem:additive-approximation-algorithms} follows by combining \autoref{lem:additive-approximation-to-max-ip} and 
\autoref{lem:additive-max-min}.
\section{Hardness of Approximation}\label{sec:hardness-approx}
In this section, we give the proofs of our results for hardness of approximation. Recall that in the minimization case, \autoref{cor:constant-approx-min-characterization} immediately rules out the existence of multiplicative approximations for any $\psi$ with $\Hand(\psi) = k$ or $3 \leq \Hand(\psi) \leq k$. This was done by arguing that \emph{any} approximation for $\VMin(\psi)$ is at least as hard as $\Zero(\psi)$, for which~\cite{BringmannFK19} gave a complete classification.

The maximization case is more complicated. To show hardness of approximation for maximization problems, we make use of the distributed PCP framework which was introduced in~\cite{AbboudRW17} and further developed in~\cite{Rubinstein18,Chen18,KarthikLM19}. We use this as the main tool (\autoref{lem:reduction_clean}) to give our results for hardness of approximation. In~\autoref{sec:hardness-approx:sec:mip-hardness} we address the \emph{hardest regime}, giving the lower bound of~\autoref{thm:tight-polynomial-maximization}. In~\autoref{sec:hardness-approx:sec:k-3-hardness} we develop our main technical contribution and use it to rule out approximation schemes, completing the lower bound of~\autoref{thm:as-maximization-classification}. Finally, in~\autoref{sec:hardness-approx:sec:no-eas} we rule out \emph{efficient} approximation schemes for both maximization and minimization problems, which gives the lower bounds for Theorems~\ref{thm:additive-approximation-classification} and~\ref{thm:no-efficient-as}.

\subsection{Distributed PCP}
All the hardness results in this section rely on the distributed PCP reduction. In a nutshell, this gives a reduction from $\kOV{k}$ in the low dimensional regime, i.e.\ $d = c \log n$ where $c$ is a constant, to a subpolynomial number of instances of $\kMaxIP{k}$ with dimension $\exp(c) \cdot \log n$ and introduces a gap in the optimal value of the instances depending on whether we started from a YES or a NO instance. In this way, an algorithm giving an approximation for $\kMaxIP{k}$ can distinguish between both cases and decide the original instance.

For our exposition, we extract the main properties of this reduction in the following lemma, which is the main technical tool that we will use. In~\autoref{sec:pcp} we elaborate on how this can be derived as a simple combination of the results in~\cite{Chen18} and~\cite{KarthikLM19}, both of which make use of the key improvement of Rubinstein~\cite{Rubinstein18} over the basic reduction introduced in~\cite{AbboudRW17}.

\begin{restatable}[\cite{Chen18,KarthikLM19,Rubinstein18}]{lemma}{lemreductionclean} \label{lem:reduction_clean}
  Let $\phi: \{0,1\}^k \to \{0,1\}$ have exactly one satisfying assignment.
  There is a universal constant $c_1$ such that for every $c$ and parameters $0 < \gamma \leq 1, \tau \geq 2$, there is a reduction from \kOV{k} on sets of size $n$ and dimension $d = c \log n$, to $n^{\gamma}$ instances $I_1,\dots,I_{n^{\gamma}}$ of $\VMax(\phi)$, each on $d'$ dimensions and with sparsity $m'$. Let $T := c \log n \cdot \tau^{c_1}$. Then, the reduction has the following properties:
  \begin{itemize}
    \item $d' = \tau^{\poly(k, c/\gamma)} \cdot \log n$ and $m' \leq n d'$.
    \itemdesc{Completeness:} If we start from a YES instance of $\kOV{k}$, then there is some $j \in [n^{\gamma}]$ such that in~$I_j$, it holds that $\OPT = T$.
    \itemdesc{Soundness:} If we start from a NO instance of $\kOV{k}$, then $\OPT \leq T / \tau$ for all instances $I_1,\ldots,I_{n^\gamma}$.
  \end{itemize}
  The running time of the reduction is linear in the output size: $\Order(n^{1+\gamma}d')$.
\end{restatable}

\subsection{\texorpdfstring{Hardness for $\Hand(\phi) = k$}{Hardness for Hand(phi) = k}}\label{sec:hardness-approx:sec:mip-hardness}

In this section we give a proof of the lower bound part of~\autoref{thm:tight-polynomial-maximization}.

\begin{lemma}[No Subpolynomial-Factor Approximation for $\Hand(\phi) = k$] \label{lem:max-hardness-poly}
Let $\phi: \{0,1\}^k \to \{0,1\}$ have and-hardness \makebox{$\Hand(\phi) = k$}. Then for all $\delta > 0$ there is an $\varepsilon > 0$ such that there is no $m^{\varepsilon}$-approximation for $\VMax(\phi)$ in time $\Order(m^{k - \delta})$ unless the low-dimensional \kOV{k} hypothesis fails.
\end{lemma}
\begin{proof}
    Assume for contradiction that there is a $\delta > 0$ such that for all $\varepsilon > 0$ there is an algorithm $\calA$ which gives an $m^{\varepsilon}$-approximation to $\VMax(\phi)$ in time $\Order(m^{k-\delta})$. We will show that by appropriately setting the parameters of the reduction in~\autoref{lem:reduction_clean}, we can reduce any $\kOV{k}$ instance on $d = \Theta(\log n)$ to $n^\gamma$ instances of $\VMax(\phi)$ introducing a gap in the optimal value of $\tau = (m')^{\varepsilon}$, where~$m'$ is the sparsity of the constructed instances of $\VMax(\phi)$. By making $\gamma$ small enough and ensuring that $m' \leq n^{1 + \beta}$ for some small constant $\beta > 0$, we can thus decide the $\kOV{k}$ instance by running $\calA$ on the constructed instances, in time $\Order(n^{\gamma}\cdot (n^{1+\beta})^{k - \delta}) = \Order(n^{k - \delta'})$ for some $\delta' > 0$, which contradicts the low-dimensional \kOV{k} hypothesis.
    
    Now we give the full details. By \autoref{def:model-checking-hardness}, note that $\Hand(\phi) = k$ if and only if $\phi$ has exactly one satisfying assignment, so we can use~\autoref{lem:reduction_clean}. Set $\beta := \delta/(3k)$. Then, for every constant $c > 0$, we apply Lemma~\ref{lem:reduction_clean} with parameters $\gamma := \delta/3$ and $\tau := (n^{1+\beta})^{\varepsilon}$ to reduce any instance of \kOV{k} on $d = c \log n$ dimensions, to $n^{\gamma}$ instances of $\VMax(\phi)$ on $d' = \tau^{\poly(k, c/\gamma)} \cdot \log n$ dimensions each. 
    
    Note that by setting $\varepsilon < \frac{\beta}{(1+\beta)\cdot \poly(k, c/\gamma)}$, we can bound $d' \leq n^{\beta}$ and further, the sparsity of the constructed instances by $m' \leq n^{1+\beta}$. Since by assumption $\calA$ computes an $(m')^{\varepsilon} \leq (n^{1+\beta})^\varepsilon = \tau$\=/approximation, we can decide the original $\kOV{k}$ instance by running $\calA$ on the constructed instances. This takes time
    \begin{equation*}
      \Order(n^{\gamma} \cdot (m')^{k - \delta})
      =  \Order(n^{\gamma} \cdot (n^{1+\beta})^{k-\delta})
      = \Order(n^{k + \beta k - \delta + \gamma})
      = \Order(n^{k + \delta/3 - \delta + \gamma})
      = \Order(n^{k - \gamma}).
    \end{equation*}
    Since the time needed to run the reduction is $\tilde\Order(n^{1 + \gamma})$, we obtain an algorithm to decide the $\kOV{k}$ instance in time $\Order(n^{k - \gamma})$. This contradicts the low-dimensional \kOV{k} hypothesis.
\end{proof}

\subsection{\texorpdfstring{Hardness for $3 \leq \Hand(\phi) < k$}{Hardness for 3 <= Hand(phi) < k}}\label{sec:hardness-approx:sec:k-3-hardness}

If $\Hand(\phi) = \Hdeg(\phi) = k$, then~\autoref{lem:max-hardness-poly} rules out even the existence of subpolynomial multiplicative approximations for $\VMax(\phi)$. But consider the case when $3 \leq \Hand(\phi) < k$: By~\autoref{thm:constant-approx-max-classification} we know that $\VMax(\phi)$ admits \emph{some} constant approximation, but how small can the approximation factor be? In particular, can we hope to get an approximation scheme? In this section we answer this negatively by showing that any such function $\phi$ cannot be approximated within $(1 - \binom k3{}^{-1} + \varepsilon)^{-1}$ in time $\Order(m^{k-\delta})$ for any $\varepsilon > 0$ and $\delta > 0$, assuming the Sparse \MAXThreeSAT{} hypothesis. In other words, we can rule out approximation schemes for \emph{all such $\phi$'s}.

Our reduction is outlined in~\autoref{fig:reductions}. At a high level there are two main steps: we first introduce an intermediate problem that we call $\kOV{(k,3)}$ and show that it is hard in the low-dimensional regime assuming the Sparse \MAXThreeSAT{} Hypothesis (\autoref{lem:hardness_k3ov}). We then give a gap introducing reduction from $\kOV{(k,3)}$ to $\VMax(\phi)$ by using the distributed PCP reduction as a blackbox (\autoref{lem:reduction_k3ov_to_maxphi} and~\autoref{thm:no-as-maximization}).\footnote{We note that an alternative to our reduction is to reduce \MAXThreeSAT{} to $\kExactIP{(k,3)}$, and then reduce $\kExactIP{(k,3)}$ to approximate $\VMax(\phi)$ using the distributed PCP. Although this is slightly more direct, it requires to show that the Disjointness protocol used in the distributed PCP can be extended for the more general $\kExactIP{k}$. Moreover, we think the reduction via $\kOV{(k, 3)}$ provides a conceptually clean and natural abstraction that might inspire future applications.}

We start by formally defining our intermediate problems:

\begin{definition}[\kOV{(k,3)}] \label{def:k-3-phi}
Given sets of $n$ vectors $X_1, \ldots, X_k \subseteq \{0, 1\}^d$, where to each coordinate \makebox{$y \in [d]$} we associate three \emph{active} indices $a, b, c \in [k]$, detect if there are vectors $x_1 \in X_1, \ldots, x_k \in X_k$ such that for all $y \in [d]$, it holds that $x_a[y] \mult x_b[y] \mult x_c[y] = 0$ where $a, b, c$ are the active indices at $y$.
\end{definition}

In the same spirit, we also define the ``$(k,3)$-'' problems for $\kMaxIP{k}$ and \kExactIP{k}: Here, let the \emph{value} $\Val(x_1, \dots, x_k)$ denote the number of coordinates $y \in [d]$ with $x_a[y] \mult x_b[y] \mult x_c[y] = 1$, where again $a, b, c$ are the active indices at $y$. Then the \kMaxIP{(k,3)} problem is to maximize the value among all vectors $x_1, \dots, x_k$, and the \kExactIP{(k,3)} problem is to check whether there exist vectors $x_1, \dots, x_k$ of value equal to a given target.

\subsubsection{Hardness of Low-Dimensional \texorpdfstring{$(k,3)$-OV}{(k,3)-OV}}

As mentioned earlier, the first step in the reduction it to show that we can reduce a sparse instance of \MAXThreeSAT{} to an instance of $\kOV{(k,3)}$ with dimension $d = \Theta(\log n)$. We proceed by first reducing $\MAXThreeSAT$ to $\kExactIP{(k,3)}$, and then make use of the following known reduction from $\kExactIP{3}$ to $\kOV{3}$:

\begin{lemma}[Reduction from $\kExactIP{3}$ to $\kOV{3}$ {\cite[Lemma 4.2]{ChenW19}}]\label{lem:red_exactip_to_ov}
  For all constants $c > 0$ and $\gamma > 0$, there is a reduction from $\kExactIP{3}$ with target $\beta$ on $n$ vectors and dimension $d = c \log n$ to $n^{\gamma \log(c/\gamma)}$ instances $I_1,\dots,I_{n^{\gamma \log(c/\gamma)}}$ of $\kOV{3}$ on $n$ vectors and dimension $d' = \Order(2^{c/\gamma} \cdot \log n)$ with the following properties:
  \begin{itemize}
    \item Each vector in the $\kExactIP{3}$ is in one-to-one correspondence with a vector in each of the constructed instances.
    \itemdesc{Completeness:} If we start from a YES instance of \kExactIP{3}, then there exists a $j \in [n^{\gamma \log(c/\gamma)}]$ such that $I_j$ is a YES instance of \kOV{3}.
    \itemdesc{Soundness:} If we start from a NO instance of $\kExactIP{3}$, then for all $j \in [n^{\gamma \log(c/\gamma)}]$, $I_j$ is a NO instance.
  \end{itemize}
  The running time of the reduction is linear in the output size: $\tilde\Order(n^{1+\gamma \log(c/\gamma)})$.
\end{lemma}

\begin{lemma}[Low-Dimensional \kOV{(k,3)} is Hard under Sparse \MAXThreeSAT{}]\label{lem:hardness_k3ov}
For every $\delta > 0$, there exists some $c > 0$ such that \kOV{(k,3)} in dimension $d = c \log n$ cannot be solved in time $\Order(n^{k-\delta})$, unless the Sparse \MAXThreeSAT{} hypothesis fails.
\end{lemma}
\begin{proof}
  Take an instance $\Phi$ of $\MAXThreeSAT$ on $N$ variables and $M = c N$ clauses. As a first step, we reduce to $O(N)$ instances of $\kExactIP{(k, 3)}$ using the split-and-list technique due to Williams~\cite{Williams05}: We split the variables into $k$ groups, and let $V_1,\dots,V_k$ be the sets of partial assignments to each group, i.e.\ $|V_i| = 2^{N/k}$. Given a partial assignment $v_i \in V_i$, we construct a vector $x_i \in \{0,1\}^{M}$ and assign its entries as $x_i[j] = 0$ if and only if $v_i$ contains a literal which satisfies the clause $C_j$. In this way, we obtain $k$ sets of vectors $X_1,\dots,X_k$ with $n := |X_i| = 2^{N/k}$ and dimension $d = M = c N = c \cdot k \log n$. For each coordinate $j \in [d]$ we say that the set $X_i$ \emph{is active in dimension $j$} iff its corresponding group contains a variable inside the clause $C_j$. Since we start from a $\MAXThreeSAT$ instance, each dimension has at most 3 active sets. The key property is that for an assignment $x_1,\dots,x_k$, we have that $d - \langle x_1,\dots,x_k \rangle$ counts the number of satisfied clauses. Moreover, in each dimension, the inner product \emph{only depends on the active vectors} (a set~$X_i$ that is not active for a given clause $C_j$ has $x_i[j] = 1$ for all $x_i \in X_i$). Thus, we can decide the $\MAXThreeSAT$ instance, by deciding $d$ instances of $\kExactIP{(k,3)}$: we try all $d$ possible target values and output~$d$ minus the minimum. Note that the way we defined $\kExactIP{(k,3)}$
  requires \emph{exactly} 3 active sets per dimension, while in this reduction we have \emph{at most} 3. To fix this, if some dimension has $a < 3$ active sets, we select $3 - a$ other arbitrary sets as active. The instance remains equivalent since for that dimension all vectors in the newly active sets will have a one.
  
  Next, we use \autoref{lem:red_exactip_to_ov} to reduce each of the $d$ instances of $\kExactIP{(k,3)}$ to $\kOV{(k,3)}$: Fix an instance $X_1,\dots,X_k$ of $\kExactIP{(k,3)}$ with target $\beta$. For each $\{a,b,c\} \in \binom{[k]}3$, we project $X_a, X_b, X_c$ to the dimensions where $\{a,b,c\}$ are active. Let $X'_a, X'_b, X'_c$ be the resulting sets. Then, for each $\alpha \in [\beta]$, we can see $(X'_a,X'_b,X'_c)$ with target $\alpha$ as a $\kExactIP{3}$ instance. For each of these, we apply the reduction in~\autoref{lem:red_exactip_to_ov}, and obtain a collection of $\kOV{3}$ instances. We extend them to instances of $\kOV{(k,3)}$ by making the sets corresponding to $X_a,X_b,X_c$ \emph{active} in \emph{all} dimensions, and adding the sets $X_j$ for $j \in  [k] \setminus \{a,b,c\}$ with $n$ zero-vectors each (i.e.\ we add dummy vectors for each non-active set, which are irrelevant for the instance). We denote the resulting collection of $\kOV{(k,3)}$ instances by $\calI_{\{a,b,c\}}$. By the properties of \autoref{lem:red_exactip_to_ov}, we have that $|\calI_{\{a,b,c\}}| \leq \beta \cdot n^{\gamma \log (c/\gamma)} = \Order(n^{\gamma \log(c/\gamma)}\log n)$ for all $\{a,b,c\} \in \binom{[k]}3$, and the dimension of each instance is $\Order(2^{c/\gamma} \log n)$. As a final step, we combine the $\calI$'s into $\Order(n^{\binom k3 \gamma \log(c/\gamma)} \log n)$ instances of $\kOV{(k,3)}$ in the natural way: for each combination of $(I_1,\dots,I_{\binom k3}) \in \calI_{1} \times \dots \times \calI_{\binom k3}$ (where we consider all $\{a,b,c\} \in \binom{[k]}3$ in some fixed order), we construct an instance of $\kOV{(k,3)}$ by concatenating the $\binom k3$ vectors corresponding to the original $\kExactIP{(k,3)}$ instance (note that in each $I_j$, there is a one-to-one correspondence between vectors in the original and the reduced instances) and summing up the targets (we discard any constructed instance not summing up to~$\beta$).

  In summary, we reduced the given $\MAXThreeSAT$ instance~$\Phi$ of $N$ variables and $cN$ clauses to $\tilde\Order(n^{\binom k3\gamma \log(c/\gamma)})$ instances of $\kOV{(k,3)}$ with dimension $d' = \Order(\binom k3 2^{c / \gamma} \log n)$. By the completeness and soundness of \autoref{lem:red_exactip_to_ov}, we can decide $\Phi$ by checking whether any of the constructed $\kOV{(k,3)}$ instances is a YES-instance. Thus, assume for contradiction that there is an algorithm $\calA$ and some $\delta > 0$ such that for all $c'$, $\calA$ decides $\kOV{(k,3)}$ on $d = c' \log n$ dimensions in time $\Order(n^{k - \delta})$. Then, by picking $\gamma$ sufficiently small we can make $\gamma \binom k3 \log(c/\gamma) < \delta/2$. Since the dimension of the $\kOV{(k,3)}$ instances is $d' = \Order(\binom k3 2^{c / \gamma} \log n) = \Order(\log n)$, we can solve the $\MAXThreeSAT$ by running $\calA$. This takes time
  \begin{equation*}
    \tilde\Order(n^{\binom k3\gamma\log(c / \gamma)} \cdot n^{k - \delta}) = \tilde\Order(n^{k - \delta/2}).
  \end{equation*}
  Since $n = 2^{N/k}$, this gives a $\MAXThreeSAT$ algorithm in time $\Order(2^{N(1 - \delta/(2k))})$, contradicting the Sparse $\MAXThreeSAT$ hypothesis.
\end{proof}

\subsubsection{Gap Reduction from Low-Dimensional \texorpdfstring{\kOV{(k,3)}}{(k,3)-OV} to \texorpdfstring{$\VMax(\phi)$}{VMax(phi)}}
Having established the hardness of $\kOV{(k,3)}$ in the low-dimensional regime, we will use the distributed PCP reduction from~\autoref{lem:reduction_clean} to show hardness of approximation for $\VMax(\phi)$ when $\Hand(\phi) \geq 3$. As an intermediate step we will need the following reduction:

\begin{lemma}[Reduction from \kMaxIP{(k,3)} to $\VMax(\phi)$]\label{lem:red_k3mip_to_h3}
Let $\phi : \{0,1\}^k \to \{0,1\}$ have and-hardness $3 \leq \Hand(\phi) \leq k$. Then, given an instance of $\kMaxIP{(k,3)}$ on $n$ vectors of dimension $d$, there is a reduction to $\VMax(\phi)$ on $n$ vectors of dimension $d$ such that the value of any solution $x_1, \dots, x_k$ in the original instance is the same as the value of the corresponding solution $x_1', \dots, x_k'$ in the constructed instance. The reduction runs in time $\Order(nd)$.
\end{lemma}
\begin{proof}
  For each coordinate $y \in [d]$, the $\kMaxIP{(k,3)}$ instance has three active indices $a_y, b_y, c_y$. By definition of $\Hand(\phi)$, there is an $S = \{a_y, b_y, c_y\}$-restriction $\phi'$ of $\phi$ which has a unique satisfying assignment. Let $\alpha_y : [k] \setminus S \to \{0,1\}$ be the partial assignment which restricts $\phi$ to $\phi'$, and denote the satisfying assignment by $\beta_y: S \to \{0,1\}$.

  For every $j \in [k]$ and all $x_j \in X_j$, we create a vector $x'_j \in X'_j$ in the $\VMax(\phi)$ instance, and set the entries for each coordinate $y \in [d]$ as follows:
  \[
    x'[y] := 
    \begin{cases}  
      \alpha_y(j) & \text{if $j \notin \{a_y, b_y, c_y\}$,} \\
      \beta_i(j) & \text{if $j \in \{a_y, b_y, c_y\}$ and $x_y[j] = 1$,} \\
      1 - \beta_i(j) & \text{otherwise.}
    \end{cases}
  \]
  It is immediate that the $\VMax(\phi)$ instance has the same number of vectors and the same dimension as the original instance. Moreover, the value of each possible solution is the same in both instances. The running time is $\Order(nd)$.
\end{proof}

Next, we give a gap introducing reduction from low-dimensional \kOV{(k,3)} to $\VMax(\phi)$. The steps are similar to those in the proof of~\autoref{lem:hardness_k3ov}:

\begin{lemma}[Reduction from Low-Dimensional $\kOV{(k,3)}$ to $\VMax(\phi)$]\label{lem:reduction_k3ov_to_maxphi}
  Let $\phi: \{0,1\}^k \to \{0,1\}$ have hardness $3 \leq \Hand(\phi) \leq k$. Then, there is a universal constant $c_1$ such that for every $c > 0$, and parameters $0 < \gamma \leq 1$ and $\tau \geq 2$, there is a reduction from \kOV{(k,3)} in dimension $c \log n$ to~$n^{\gamma}$ instances $I_1,\dots,I_{n^{\gamma}}$ of $\VMax(\phi)$. For $T := c \log n \cdot \tau^{c_1}$, the following properties hold:
  \begin{itemize}
    \item The constructed $\VMax(\phi)$ instances have dimension $d' = \binom k3\cdot \tau^{\poly(k,c/\gamma)} \log n$, and sparsity $m' = \Order(n \cdot \tau^{\poly(k, c/\gamma)}\log n)$.
    \itemdesc{Completeness:} If we start from a YES instance of $\kOV{(k,3)}$, then there is some $j \in [n^{\gamma}]$ such that in $I_j$ it holds that $\OPT = \binom k3 \cdot T$.
    \itemdesc{Soundness:} If we start from a NO instance of $\kOV{(k,3)}$, then for every $I_j$, $j \in [n^{\gamma}]$ we have that $\OPT \leq (\binom k3 - 1) \cdot T + T/\tau$.
  \end{itemize}
  The running time of the reduction is linear in the output size: $\tilde\Order(n^{1 + \gamma})$.
\end{lemma}
\begin{proof}
  Let $X_1,\dots,X_k$ be the input instance of $\kOV{(k, 3)}$ in $d = c \log n$ dimensions. As a first step, we make sure that for all $\{a,b,c\} \in \binom{[k]}3$, there is at least one dimension in which $\{a,b,c\}$ are active, by adding a fresh dimension with all zeros for all vectors if necessary.
  
  Now we describe the reduction. For each $\{a,b,c\} \in \binom{[k]}3$, we project $X_a, X_b, X_c$ to the dimensions where $\{a,b,c\}$ `are active. Let $X'_a, X'_b, X'_c$ be the resulting sets and note that we can see it as a $\kOV{3}$ instance. We pad the vectors with zeros, so that the dimension of this $\kOV{3}$ instance is also $d = c \log n$. Next, we invoke the reduction from \autoref{lem:reduction_clean} on $X'_a, X'_b, X'_c$ with parameters $\gamma' := \gamma / \binom k3$ and $\tau' := \tau$. This produces $n^{\gamma'}$ instances of $\kMaxIP{3}$, which we extend to instances of $\kMaxIP{(k, 3)}$ by adding zeros to the sets $X_j$'s for $j \in [k] \setminus \{a,b,c\}$ (i.e.\ we put zero-vectors corresponding to each vector from the non-participating sets). We denote by $\calI_{\{a,b,c\}}$ the collection of $n^{\gamma'}$ constructed instances of $\kMaxIP{3}$.

  Having constructed $\calI_{\{a,b,c\}}$ corresponding to each $\{a,b,c\} \in \binom{[k]}3$, we merge them into $n^{\gamma' \cdot \binom k3} = n^{\gamma}$ instances of $\kMaxIP{(k,3)}$ as follows.  For each $(I_1,\dots,I_{\binom k3}) \in \calI_{1} \times \dots \times \calI_{\binom k3}$ we construct one instance of $\kMaxIP{(k,3)}$ in the natural way: Each $I_j$ has a vector associated to every $x_i \in X_i$ from the original $\kOV{(k,3)}$ instance, so we concatenate the $\binom k3$ copies across the $I$'s. This gives an instance of $\kMaxIP{(k,3)}$ with the same number of vectors as the original instance, dimension $d' = \binom k3 \cdot \tau^{\poly(k, c/\gamma)} \log n$. and sparsity $m' = \Order(n \cdot d')$.

  As a final step, we apply \autoref{lem:red_k3mip_to_h3} to each of the $n^{\gamma}$ instances of $\kMaxIP{(k,3)}$ and obtain $n^{\gamma}$  instances of $\VMax(\phi)$, each with sparsity $m' = \Order(n \cdot d') = \Order(n \cdot \tau^{\poly(k, c/\gamma)}\log n)$. Now we analyze the completeness and soundness of the reduction:

  \begin{description}
    \item[Completeness:] If $X_1,\dots,X_k$ is a YES instance of $\kOV{(k,3)}$, then there are vectors $x_1 \in X_1, \dots, \allowbreak x_k \in X_k$ such that $\langle x'_a, x'_b, x'_c \rangle = 0$ in the corresponding projected $\kOV{3}$ instances $X'_a,X'_b,X'_c$ for all $\{a, b, c\} \in \binom{[k]}3$. By the completeness of the reduction in \autoref{lem:reduction_clean}, for each $\{a,b,c\} \in \binom k3$, in at least one of the $n^{\gamma'}$ instances of $\kMaxIP{3}$ the vectors corresponding to $x'_a, x'_b, x'_c$ attain a value of $T$. In particular, in at least one of the $n^{\gamma}$ final instances of $\VMax(\phi)$, the value of the vectors corresponding to $x_1,\dots,x_k$ is at least $\binom k3 \cdot T$.
    \item[Soundness:] If $X_1,\dots,X_k$ is a NO instance, then for all $x_1 \in X_1, \dots, x_k \in X_k$, there is at least one $\{a,b,c\} \in \binom{[k]}3$ such that $\langle x'_a, x'_b, x'_c \rangle > 0$ in the corresponding projected $\kOV{3}$ instance $X'_a,X'_b,X'_c$. Hence, by the soundness of the reduction in \autoref{lem:reduction_clean}, in all the final instances of $\VMax(\phi)$, the value attained by the vectors corresponding to $x_1,\dots,x_k$ is at most $(\binom k3 - 1)\cdot T + T / \tau$.
  \end{description}
  The running time is dominated by constructing the $n^{\gamma}$ instances of $\kMaxIP{(k,3)}$, which takes time $\tilde\Order(n^{1 + \gamma})$. 
\end{proof}

Finally, we are ready to give the main result, which proves the lower bound of~\autoref{thm:as-maximization-classification}:

\begin{theorem}[No Approximation Scheme for $3 \leq \Hand(\phi)$] \label{thm:no-as-maximization}
  Let $\phi: \{0,1\}^k \to \{0,1\}$ have~and-hardness $3 \leq \Hand(\phi) \leq k$. Then, for all $0 < \varepsilon < \binom k3{}^{-1}$ and for all $\delta > 0$, we cannot compute a \raisebox{0pt}[0pt][0pt]{$(1 - \binom k3{}^{-1} + \varepsilon)^{-1}$}-approximation for $\VMax(\phi)$ in time $\Order(m^{k - \delta})$ unless the Sparse $\MAXThreeSAT$ hypothesis fails. 
\end{theorem}
\begin{proof}
  Suppose for contradiction there is some \raisebox{0pt}[0pt][0pt]{$0 < \varepsilon < \binom k3{}^{-1}$} and $\delta > 0$ such that there is an algorithm $\calA$ computing a multiplicative $(1 - \binom k3{}^{-1} + \varepsilon)^{-1}$-approximation for $\VMax(\phi)$ in time $\Order(m^{k - \delta})$. We will show that by applying~\autoref{lem:reduction_k3ov_to_maxphi}, we can use $\calA$ to decide any low-dimensional \kOV{(k,3)} instance in time $\Order(n^{k-\delta'})$ for some $\delta' > 0$, which would contradict the Sparse \MAXThreeSAT{} hypothesis.
  
  For every constant $c > 0$, we invoke the reduction from \autoref{lem:reduction_k3ov_to_maxphi} with parameters $\gamma := \delta / 2$ and $\tau > (\varepsilon \binom k3)^{-1}$ to reduce a \kOV{(k,3)} instance on $n$ vectors of dimension $c \log n$ to $n^{\gamma}$ instances of $\VMax(\phi)$ each with sparsity $m' = \Order(n \tau^{\poly(k, c/\gamma)} \log n) = \Order(n \log n)$.
  
  By the soundness and completeness of the reduction, a $(1 - \binom k3{}^{-1} + \varepsilon)^{-1}$-approximation is sufficient to distinguish between the two cases:
  \[
    \Bigg(1 - \binom k3^{\!-1}\! + \varepsilon\Bigg) \cdot \binom k3 \cdot T > 
    \left(\binom k3 - 1\right)\cdot T + T / \tau.
  \]

  Thus, we can decide the original $\kOV{(k,3)}$ instance by running the reduction and using $\calA$ on the constructed instances. The time to run the reduction is $\tilde\Order(n^{1 + \gamma})$, and running $\calA$ on $n^{\gamma}$ instances of $\kMaxIP{(k,3)}$ with sparsity $m' = \Order(n \log n)$ takes time $\Order(n^\gamma \mult (m')^{k-\delta}) = \tilde\Order(n^{k+\gamma-\delta})$. Thus, by choosing $\gamma = \delta / 2$, the total running time is bounded by $\tilde\Order(n^{k-\gamma})$. By~\autoref{lem:hardness_k3ov}, this contradicts the Sparse \MAXThreeSAT{} hypothesis.
\end{proof}

\subsection{Ruling out Efficient Approximation Schemes}\label{sec:hardness-approx:sec:no-eas}

Recall that we say that $\VOpt(\phi)$ admits an \emph{efficient} approximation scheme, if there is a fixed $\delta > 0$ such that for any $\varepsilon > 0$, an $\varepsilon$-approximation can be computed in time $O(m^{k-\delta})$. In this section we rule out the existence of efficient approximation schemes for all problems which cannot be solved exactly, namely, for all problems $\VOpt(\phi)$ with $3 \leq \Hdeg(\phi)$ or $\Hdeg(\phi) = k$, thereby proving the lower bounds in \autoref{thm:additive-approximation-classification} and \autoref{thm:no-efficient-as}. Since a multiplicative approximation scheme implies an additive one, it suffices to show the result for the additive case. Moreover, by~\autoref{lem:additive-max-min}, we know that an additive approximation scheme for $\VMax(\phi)$ is equivalent to an additive approximation scheme for $\VMin(\phi)$, so we focus on the minimization case.

For $h \geq 3$, consider the natural generalization \kMaxIP{(k,h)} of the \kMaxIP{(k,3)} problem defined in~\autoref{def:k-3-phi}. The main ingredient for our results is the following lemma, which shows that an additive approximation for $\kMaxIP{(k,h)}$ can be obtained from an additive approximation for $\VMin(\phi)$, when $\Hdeg(\phi) = h$.

\begin{lemma}\label{lem:kh-mip-to-hdeg-h-reduction-additive}
Let $\phi: \{0,1\}^k \to \{0,1\}$ have degree hardness $\Hdeg(\phi) = h$. If $\VMin(\phi)$ has an additive approximation scheme, then for any $\varepsilon > 0$, there exists some $\delta > 0$ such that \kMaxIP{(k,h)} has an additive $\varepsilon d$-approximation in time $\Order((nd)^{k-\delta})$.
\end{lemma}
\begin{proof}
  Let $X_1,\dots,X_k \subseteq \{0,1\}^d$ be the $\kMaxIP{(k,h)}$ instance. We describe how to construct an instance $X'_1,\dots,X'_k \subseteq \{0,1\}^{O(d)}$ of $\VMin(\phi)$. For each coordinate $y \in [d]$, let $S_y \in \binom{[k]}h$ denote the set of active indices at $y$. Since $\phi$ has degree hardness $\Hdeg(\psi) = h$, there is an $S_y$-restriction~$\phi_y$ of~$\phi$ such that $\deg(\phi_y) = h$. Let $\alpha_y : [k] \setminus S_y \to \{0, 1\}$ denote the partial assignment which restricts~$\phi$ to~$\phi_y$. By \autoref{lem:coordinate-gadget} there is a coordinate gadget for~$\phi_y$, consisting of functions $\phi_{y,1}, \dots, \phi_{y,\ell}$ with constants $\beta^y_1, \beta^y_2 > 0$. For $j \in [\ell]$, let $\tau_{j,1},\dots,\tau_{j,k}$ denote the unary functions needed to transform~$\phi_y$ into $\phi_{y,j}$. For every $i \in [k]$ and $x_i \in X_i$ we introduce $\ell$ new coordinates $x'_i[(y,1)],\dots,x'_i[(y,\ell)]$ and to be set as follows:
  \begin{equation*}
    x'_i[(y, j)] := 
    \begin{cases}
      \alpha_y(i) & \text{if $i \notin S_y$,} \\
      \tau_{j,i}(x_i[y]) & \text{otherwise.}
    \end{cases} 
  \end{equation*}

  Note that we can make the constants $\beta^y_2$ the same for all $y \in [d]$: Let $\beta^*_2$ be the lowest common multiple of all the $\beta^y_2$'s. Then, by concatenating the $y$-th coordinate gadget $\beta^*_2 / \beta^y_2$ many times, we obtain $\beta^*_2$ for all $y \in [d]$. Thus, letting $\beta^*_1 := \sum_{y \in [d]} \beta^y_1$ and $\beta^*_2$ as indicated, we have that
  \[
      \Val(x'_1,\dots,x'_k) = \sum_{y \in [d]} \Bigg(\beta^y_1 - \beta^y_2 \prod_{i \in S_y} x_i[y]\Bigg) = \beta^*_1 - \beta^*_2 \cdot \Val(x_1,\dots,x_k).
  \]

  Let $d' = O(d)$ be the dimension of the constructed instance $X'_1,\dots,X'_k$ of $\VMin(\phi)$ and let $\OPT'$ be its optimal value. Let $\OPT$ be the optimal for the original instance $X_1,\dots,X_k$ of \kMaxIP{(k,h)}. Since $\beta^*_1, \beta^*_2 > 0$, we have that $\OPT' = \beta^*_1 - \beta^*_2 \cdot \OPT$. Thus, given an additive $\varepsilon d'$-approximation~$V'$ for $\VMin(\phi)$, we compute $V := -(V' - \beta^*_1) / \beta^*_2 \in [\OPT - \varepsilon/\beta^*_2 \cdot d', \OPT]$, which is an additive $\Order(\varepsilon d)$-approximation for \kMaxIP{(k,h)}.

  Since the reduction takes time $O(nd)$ and the produced instance is of size $m = O(nd)$, the result follows.
\end{proof}

Note that if $h = k$ in~\autoref{lem:kh-mip-to-hdeg-h-reduction-additive}, then the $\kMaxIP{(k,k)}$ is the same as the standard \kMaxIP{k} and the result still applies. 

Now we are ready to give the lower bounds. We start with the hardest regime, i.e., when $\Hdeg(\phi) = k$. 

\begin{lemma}[No Efficient Approximation Scheme for $\Hdeg(\phi) = k$]\label{lem:no-efficient-as-seth}
  Let $\phi: \{0,1\}^k \to \{0,1\}$ have degree hardness $\Hdeg(\phi) = k$. Then, for all $\delta > 0$ there is some $\varepsilon > 0$ such that there is no additive $\varepsilon d$-approximation for $\VMin(\phi)$ in time $\Order(m^{k - \delta})$ unless the low-dimensional \kOV{k} hypothesis fails.
\end{lemma}
\begin{proof}
  Suppose there is a $\delta > 0$ and an algorithm $\calA$ computing an additive $\varepsilon d$-approximation for $\VMin(\phi)$ in time $\Order(m^{k - \delta})$ for any $\varepsilon > 0$. We will show how to use~\autoref{lem:reduction_clean} in combination with~\autoref{lem:kh-mip-to-hdeg-h-reduction-additive} to reduce any \kOV{k} instance in dimension $c \log n$ to a small number of $\VMin(\phi)$ instances for which we can use $\calA$ to decide the original instance in time $\Order(n^{k-\delta'})$ for some $\delta' > 0$. 
  
  For every constant $c > 0$, we apply Lemma~\ref{lem:reduction_clean} with parameters $\gamma$ (to be set) and $\tau := 2$ to reduce an instance of \kOV{k} in dimension $d = c \log n$ to $n^{\gamma}$ instances of $\kMaxIP{k}$ on $d' = 2^{\poly(k, c/\gamma)} \log n = \Order(\log n)$ dimensions each. By the properties of Lemma~\ref{lem:reduction_clean}, if we started from a YES instance, then at least one of the constructed $\kMaxIP{k}$ instances achieves a value of $T = c \cdot 2^{c_1} \log n$. Otherwise, all instances have value at most $T/2$. Thus, if we set $\varepsilon$ such that $\varepsilon \cdot d' \leq T/6$, an additive $\varepsilon d'$-approximation can distinguish the two cases. Since $d' = 2^{\poly(k, c/\delta)}\log n$, choosing $\varepsilon < \frac{c \cdot 2^{c_1}}{6 \cdot 2^{\poly(k, c/\gamma)}}$ suffices.

  By~\autoref{lem:kh-mip-to-hdeg-h-reduction-additive}, there exists some $\delta' > 0$ such that we can compute such an additive $\varepsilon d'$-approximation using $\calA$ in time $\Order((nd')^{k-\delta'})$. In this way, we can decide the $\kOV{k}$ instance by constructing the $n^{\gamma}$ instances of $\kMaxIP{k}$ instances in time $\tilde\Order(n^{1+\gamma})$ and computing the additive $\varepsilon d'$-approximation using $\calA$. The overall running time is 
  \begin{equation*}
    \tilde\Order(n^{1 + \gamma} + n^{\gamma} \cdot ((nd') + (nd')^{k - \delta'})) = \tilde\Order(n^{k + \gamma - \delta'}).
  \end{equation*}
  By picking $\gamma < \delta'$, this contradicts the low-dimensional \kOV{k} hypothesis.
\end{proof}

Finally we show the analogous result for $\VMin(\phi)$ with $3 \leq \Hdeg(\phi) \leq k$. Note that the proof is virtually the same as for~\autoref{lem:no-efficient-as-seth}, but we use the stronger Sparse \MAXThreeSAT{} hypothesis, and the machinery developed in~\autoref{sec:hardness-approx:sec:k-3-hardness}.

\begin{lemma}[No Efficient Approximation Scheme for $3 \leq \Hdeg(\phi)$] \label{lem:no-efficient-as-sparse-max3sat}
  Let $\phi: \{0,1\}^k \to \{0,1\}$ have degree hardness $3 \leq \Hdeg(\phi) \leq k$. Then, for all $\delta > 0$, there is some $\varepsilon > 0$ such that there is no additive $\varepsilon d$-approximation for $\VMin(\phi)$ in time $\Order(m^{k - \delta})$ unless the Sparse \MAXThreeSAT{} hypothesis fails.
\end{lemma}
\begin{proof}
  Suppose there exists some $\delta > 0$ and an algorithm $\calA$ such that for any $\varepsilon > 0$, it produces an additive $\varepsilon d$-approximation for $\VMin(\phi)$ in time $\Order(m^{k - \delta})$. We will show how to use the gap introducing reduction from~\autoref{lem:reduction_k3ov_to_maxphi} to reduce any \kOV{(k,3)} instance in dimension $c \log n$ to a few instances of \kMaxIP{(k,h)}, and then use $\calA$ is to distinguish the gap via~\autoref{lem:kh-mip-to-hdeg-h-reduction-additive}.
  
  For every constant $c$, we use \autoref{lem:reduction_k3ov_to_maxphi} with parameters $\gamma$ (to be set) and $\tau := 2$ to reduce an instance of \kOV{(k,3)} in dimension $d = c \log n$ to $n^{\gamma}$ instances of \kMaxIP{(k,3)} each in dimension $d' = \tau^{\poly(k, c/\gamma)}\log n$ and with sparsity $m' = \Order(n \tau^{\poly(k, c/\gamma)} \log n) = \Order(n \log n)$. Note that we can further reduce each of these instances to \kMaxIP{(k,h)}, for any $3 < h \leq k$: For each coordinate we add $h-3$ \emph{arbitrary} active indices and set all entries for these indices to one. The produced \kMaxIP{(k,h)} instances are clearly equivalent since in each coordinate the contribution to the objective only depends on the 3 active sets in the original \kMaxIP{(k,3)} instance.
   
  To decide the \kOV{(k,3)} instance, we need to distinguish whether the optimal value of each constructed $\kMaxIP{(k,h)}$ instance is $\binom k3T$ or $(\binom k3 - 1)T + T/2$. Thus, an additive $\varepsilon d' < T/6$-approximation suffices. Since $T = c \cdot 2^{c_1} \cdot \log n$, setting $\varepsilon < \frac{c 2^{c_1}}{\tau^{\poly(k, c/\gamma)}}$ is enough. By~\autoref{lem:kh-mip-to-hdeg-h-reduction-additive}, we can compute such an additive approximation using $\calA$ in time $\Order((nd')^{k-\delta'})$ for some $\delta' > 0$.

  Hence, by constructing $\tilde\Order(n^{1+\gamma})$ instances of $\kMaxIP{(k,h)}$ and computing the additive $\varepsilon d'$-approximation using $\calA$, we can decide the original \kOV{(k,3)} instance in time
  \begin{equation*}
    \tilde\Order(n^{1+\gamma} + n^{\gamma}(nd' + (nd')^{k-\delta'}))  = \tilde\Order(n^{k + \gamma - \delta'})
  \end{equation*}
  By choosing $\gamma < \delta'$, this contradicts the Sparse \MAXThreeSAT{} hypothesis.
\end{proof}

\section{Distributed PCP and Extensions} \label{sec:pcp}
In this section we sketch how~\autoref{lem:reduction_clean} can be obtained by combining previously known results in~\cite{Chen18}, \cite{KarthikLM19} and~\cite{Rubinstein18}.  

\lemreductionclean*

\autoref{lem:reduction_clean} was shown in~\cite[{Appendix D}]{Chen18}, but for the case of $k = 2$ and using $\kMaxIP{2}$ as the \emph{target} problem. Our version is a simple generalization for the case of arbitrary $k$, and to reduce to any $\VMax(\phi)$ with $\Hand(\phi) = k$. The main ingredient to obtain the parameters of the reduction as in~\autoref{lem:reduction_clean}, is an \emph{efficient MA protocol}\footnote{More precisely, the communication model is the \emph{Simultaneous Messaging Passing} model~\cite{Yao79,BabaiGKL03}, as referred to in~\cite{KarthikLM19}. But we stick to Chen's notation~\cite{Chen18}.} for Set-Disjointness:

\begin{problem}[$\Disj_{d,k}$]
  In the $\Disj_{d,k}$ problem there are $k$ players, each of which receives a vector $x_i \in \{0,1\}^d$ as input. Their goal is to jointly decide whether $\langle x_1,\dots,x_k \rangle = 0$.
\end{problem}

\begin{definition}[$(m, r, \ell, s)$-efficient MA protocols~\cite{KarthikLM19}]
  We say that a communication protocol is $(m, r, \ell, s)$-efficient for a communication problem if the following holds:
  \begin{itemize}
    \item There are $k$ players in the protocol $P_1,\dots,P_k$ and Merlin. Each player $P_i$ holds an input $x_i$.
    \item The protocol is one-round with public randomness. The following happens sequentially:
    \begin{enumerate}
      \item Merlin sends an advice string of length $m$ to $P_1$, which is a function of $x_1,\dots,x_k$.
      \item The players jointly draw $r$ random bits.
      \item Each player $P_i$ for $2 \leq i \leq k$ sends an $\ell$-bit message to $P_1$, which is a deterministic function of the input and randomness.
      \item Upon receiving the messages, $P_1$ decides to accept or reject.
    \end{enumerate}
    \item The protocol has full completeness and soundness $s \in (0,1]$.
    \item The computation of players is polynomial in their input lengths.
  \end{itemize}
\end{definition}

In particular, to obtain~\autoref{lem:reduction_clean} one needs the generalization of the protocol for $\Disj_{m,2}$ by Rubinstein~\cite{Rubinstein18} to the case of $k$ players. This generalization was written explicitly in~\cite{KarthikLM19}, where the authors explicitly acknowledge that Rubinstein suggested them the generalization.

\begin{lemma}[$k$-Disjointness Protocol~{\cite[{Theorem 6.1}]{KarthikLM19}}]\label{lem:ag_disj_protocol}
For every $d, k$ and $\alpha > 0$, there exists a $(d/\alpha, \log_2 d, \poly(k,\alpha),1/2)$-efficient protocol for $\Disj_{d, k}$.
\end{lemma}

Lemma~\ref{lem:ag_disj_protocol} improves upon the Aaronson and Wigderson protocol~\cite{AaronsonW09} used in the \emph{original} distributed PCP reduction~\cite{AbboudRW17}, which if extended to $k$ players, has parameters $(dk/\log d,\allowbreak \Order_k(\log d),\allowbreak \Order_k(\log^3 d), 1/2)$. As mentioned earlier, this key technical improvement was introduced by Rubinstein~\cite{Rubinstein18} via algebraic geometry codes.

Note that the protocol of~\autoref{lem:ag_disj_protocol} uses $\Order(\log d)$ random bits to obtain soundness $1/2$. To boost the soundness to $\varepsilon$, the standard way is to repeat the protocol (with the same advice) $\Order(\log \varepsilon^{-1})$-many times. The overall number of random bits used would thus be $\Order(\log d \cdot \log \varepsilon^{-1})$. Chen reduced this to $\Order(\log d + \log \varepsilon^{-1})$ in the case of $k = 2$ via expander walk sampling and using the aforementioned protocol for $\Disj_{d,2}$ as a blackbox. His proof extends in a straightforward way to general $k$ using the protocol from~\autoref{lem:ag_disj_protocol}:

\begin{lemma}[$k$-Disjointness Protocol with Expander Graphs~{\cite[{Lemma D.2}]{Chen18}}]\label{lem:disj_protocol_expanders}
For every $d$, $\alpha > 0$ and $0 < \varepsilon < 1/2$, there exists a $(d/\alpha, \log_2 d + \Order(\log \varepsilon^{-1}), \poly(k, \alpha)\cdot \log(\varepsilon^{-1}), \varepsilon)$-efficient protocol for $\Disj_{d, k}$.
\end{lemma}

In~\cite{AbboudRW17} it was shown that given \emph{any} communication protocol for $\Disj_{d,2}$, one can obtain a gap introducing reduction from $\kOV{2}$ on dimension $d$, to many instances of $\kMaxIP{2}$ with a gap proportional to the soundness of the protocol. The number of instances produced and the dimension of the constructed instances depend on the parameters of the protocol. This can be easily extended to our case, where we instead start from $\kOV{k}$ and reduce to instances of $\VMax(\phi)$ where $\Hand(\phi) = k$:

\begin{lemma}[Gap Introducing Reduction]\label{lem:pcp-reduction}
  Let $\phi: \{0,1\}^k \mapsto \{0,1\}$ have $\Hand(\phi) = k$. Given a $(m, r, \ell, s)$-efficient protocol for $\Disj_{d,k}$, there is a reduction from \kOV{k} on sets of size $n$ and dimension $d$ which constructs $2^m$ instances of $\VMax(\phi)$ with dimension $d' = 2^{r + (k-1) \ell}$ and sparsity $m' = O(n d')$, such that the following holds:
  \begin{itemize}
    \itemdesc{Completeness:} If we start from a YES instance of \kOV{k}, then one of the 
    constructed instances of $\VMax(\phi)$ has $\OPT = 2^r$.
    \itemdesc{Soundness:} If we start from a NO instance of \kOV{k}, then for \emph{every}
    constructed instance of $\VMax(\phi)$ it holds that $\OPT \leq s \cdot 2^r$.
  \end{itemize}
  The running time of the reduction is $O(2^{r + (k-1) \ell + m} \cdot n \cdot \poly(d))$.
\end{lemma}

\begin{proof}
  By \autoref{def:model-checking-hardness}, note that $\Hand(\phi) = k$ if and only if $\phi$ has exactly one satisfying assignment, which we denote by $(\alpha_1,\dots,\alpha_k) \in \{0,1\}^k$.

  Consider an instance $X_1,\dots,X_k \subseteq \{0,1\}^d$ of $\kOV{k}$. For each advice string $M \subseteq [2^m]$ we construct an instance 
  $X^M_1,\dots,X^M_k$ of $\VMax(\phi)$ as follows: 

  For each random string $R \in [2^r]$ and for each $x_1 \in X_1$ we create a vector $x^{M, R}_1 \in \{0,1\}^{2^{(k-1) \ell}}$. We treat a tuple of messages $(B_2,\dots,B_k)$ where $B_i \in \{0,1\}^\ell$ as an index in $[2^{(k-1)\ell}]$, and each entry of $x^{M, R}_1$ is defined as:
  \begin{equation*}
    x^{M, R}_1[(B_2,\dots,B_k)] := 
    \begin{cases}
      \alpha_1 & \text{\parbox{8.5cm}{if $P_1$ accepts on input $x_1$, advice $M$ from Merlin, randomness $R$ and messages $B_2, \ldots, B_k$ from players $P_2, \dots, P_k$,}} \\[3.5ex]
      1 - \alpha_1 & \text{otherwise.}
    \end{cases}
  \end{equation*}
  Similarly, for each random string $R \in [2^r]$ and for each $x_i \in X_i$ for $2 \leq i \leq k$ we create a vector $x^{M, R}_i \in \{0,1\}^{2^{(k-1) \ell}}$:
  \begin{equation*}
    x^{M, R}_i[(B_2,\dots, B_k)] := 
    \begin{cases}
      \alpha_i & \text{\parbox{8.5cm}{if $P_i$ sends message $B_i$ to $P_1$ on input $x_i$, advice~$M$ from Merlin and randomness $R$,}} \\[2ex]
      1 - \alpha_i & \text{otherwise.}
    \end{cases}
  \end{equation*}

  The key observation is that there is \emph{exactly one position} $j = (B_2,\dots,B_k)$ where $x^{M, R}_i[j] = \alpha_i$ for $2 \leq i \leq k$, corresponding to the \emph{actual} messages $B_2,\dots,B_k$ sent by players $P_2,\dots,P_k$ in the protocol under randomness $R$ and advice $M$. In particular, it follows that if $P_1$ accepts with these messages, this randomness and this advice then $x^{M, R}_1[j] = \alpha_1$. Thus, since $(\alpha_1,\dots,\alpha_k)$ is the unique satisfying assignment of $\phi$, we have that
  \begin{equation*}
    \Val(x_1^{M, R},\dots,x_k^{M, R}) = 
    \begin{cases}
      1 & \text{\parbox{8.5cm}{if $P_1$ accepts on inputs $x_1, \dots, x_k$, advice~$M$ from Merlin and randomness $R$,}} \\[2ex]
      0 & \text{otherwise.}
    \end{cases}
  \end{equation*}

  Finally, each $x_i^M \in X_i^M$ is constructed by concatenating $x_i^{M, R}$ over all random strings $R \in [2^r]$, which means that $X_i^M \subseteq \{0,1\}^{2^{r + (k-1) b}}$. And by the previous argument we have that
  \begin{equation*}
    \Val(x_1^M, \dots, x_k^M) = 2^r \cdot \Pr_{R}\Bigg[\text{\parbox{8.1cm}{\centering$P_1$ accepts on inputs $x_1, \dots, x_k$, advice $M$ from Merlin and randomness $R$}}\Bigg]
  \end{equation*}
  We analyze the completeness, soundness and running time of the reduction:

  \begin{itemize}
    \itemdesc{Completeness:} If $X_1,\dots,X_k$ is a YES instance of $\kOV{k}$, then there is an advice $M'$ which makes $P_1$ accept with probability one, and thus  $\OPT = 2^r$ for the instance $X^{M'}_1,\dots,X^{M'}_k$.
    \itemdesc{Soundness:} If $X_1,\dots,X_k$ is a NO instance of $\kOV{k}$, by the soundness of the protocol we have that $\OPT \leq s \cdot 2^r$ for all instances $X^M_1,\dots,X^M_k$, $M \in [2^m]$. 
    \itemdesc{Running Time:} We iterate over the $2^m$ possible advice strings, the $2^r$ possible random strings and the $2^{(k-1) \ell}$ possible messages that the players send. Then, we simulate the decision of $P_1$ based on each possible combination of advice, randomness and messages which takes $O(\poly(d))$-time per vector (due to the efficiency of the protocol). Since we do this for $O(n)$ vectors, the overall running time is $O(2^{r + (k-1) \ell + m} \cdot n \cdot \poly(d))$. \qedhere
  \end{itemize}
\end{proof}

Finally, we obtain~\autoref{lem:reduction_clean} by using the protocol from~\autoref{lem:disj_protocol_expanders} in the gap reduction from~\autoref{lem:pcp-reduction}:

\begin{proof}[Proof of~\autoref{lem:reduction_clean}]
  Given an instance of $\kOV{k}$ on sets of size $n$ and dimension $c \log n$, we take the 
  $(d/\alpha, \log_2 d + \Order(\log \varepsilon^{-1}), 
    \poly(k, \alpha)\cdot \log(\varepsilon^{-1}), \varepsilon)$-efficient protocol for $\Disj_{d,k}$ from \autoref{lem:disj_protocol_expanders} 
  with parameters $d := c \log n$, $\alpha := c/\gamma$, and $\varepsilon := 1/\tau$. Plugging it into the reduction from~\autoref{lem:pcp-reduction} yields the result.
\end{proof}

\end{document}